\pdfoutput=1 
\documentclass[a4paper,twoside,10pt,fleqn]{article}
\usepackage{helvet}

\usepackage[latin1]{inputenc}
\usepackage{fontenc}
\usepackage{amsfonts}
\usepackage{amsmath}
\usepackage{amssymb,marvosym,units}
\usepackage{amsthm}
\usepackage[cmtip,arrow, matrix, curve]{xy}
\usepackage{pb-diagram,pb-xy}

\usepackage{graphicx}
\usepackage[svgnames,rgb]{xcolor}
\usepackage{epsfig}

\usepackage{hyperref}
\usepackage{zref, xkeyval, ifpdf, ifthen, calc, marginnote, pdfcomment}

\usepackage{longtable}
\usepackage{bbm}

\hypersetup{
    pagebackref=true,
    bookmarks=true,         
    unicode=false,          
    pdftoolbar=true,        
    pdfmenubar=true,        
    pdffitwindow=true,     
    pdfstartview={FitH},    
    pdftitle={The holonomy-flux cross-product $C^*$-algebra of  Loop Quantum Gravity},    
    pdfauthor={Diana Kaminski},     
    pdfsubject={},   
    pdfkeywords={Loop Quantum Gravity, holonomy-flux cross product algebra}, 
    pdfnewwindow=true,      
    colorlinks=false,       
    linkcolor=black,          
    pdfborder= 0 0 0,
}

\title{Algebras of Quantum Variables for Loop Quantum Gravity\\[5pt]
\textbf{III. The holonomy-flux cross-product $C^*$-algebra}}
\author{Diana Kaminski\\[3pt]
kaminski@math.uni-paderborn.de\\ 
\small{Europe - Germany}}
\date{August 19, 2011}

\newcommand{\Ab}{\begin{large}\bar{\mathcal{A}}\end{large}}

\newcommand{\Alg}{\begin{large}\mathfrak{A}\end{large}}

\newcommand{\Aut}{\begin{large}\mathfrak{Aut}\end{large}}

\newcommand{\B}{\mathfrak{B}}
\newcommand{\BAlg}{\begin{large}\mathfrak{B}\end{large}}

\newcommand{\CB}{\mathbb{C}}
\newcommand{\CD}{\mathcal{C}}

\newcommand{\E}{\mathcal{E}}

\newcommand{\Goid}{\mathcal{G}}

\newcommand{\Gop}{\mathbbm{G}}
\newcommand{\GG}{G}

\newcommand{\HS}{\mathcal{H}}

\newcommand{\KD}{\mathcal{K}}

\newcommand{\la}{\langle}
\newcommand{\LD}{\mathcal{L}}

\newcommand{\N}{\mathbb{N}}

\newcommand{\op}{\mathfrak{o}}

\newcommand{\PD}{\mathcal{P}}

\newcommand{\ra}{\rangle}

\newcommand{\FD}{\mathcal{F}}

\newcommand{\R}{\mathbb{R}}

\newcommand{\SimGroup}{\mathfrak{G}}
\newcommand{\surf}{\mathbb{S}}

\newcommand{\WF}{\mathfrak{W}}

\newcommand{\ZD}{\mathcal{Z}}

\newcommand{\Zop}{\mathcal{Z}(\Gop_{\breve S,\gamma})}

\newcommand{\ho}{\mathfrak{h}}

\newcommand{\go}{\mathfrak{g}}

\DeclareMathOperator{\Act}{Act}

\DeclareMathOperator{\dif}{d}
\DeclareMathOperator{\diff}{surf}

\DeclareMathOperator{\disc}{d}

\DeclareMathOperator{\Hom}{Hom}

\DeclareMathOperator{\id}{id}

\DeclareMathOperator{\loc}{loc}

\DeclareMathOperator{\Map}{Map}
\DeclareMathOperator{\Mor}{Mor}

\DeclareMathOperator{\ori}{or}

\DeclareMathOperator{\Rep}{Rep}

\DeclareMathOperator{\tr}{tr}

\newcommand{\gp}{{\gamma^\prime}}

\newcommand{\gpe}{{\gamma^\prime_1}}
\newcommand{\gpz}{{\gamma^\prime_2}}

\newcommand{\tg}{{\tilde\gamma}}

\newcommand{\Gp}{{\Gamma^\prime}}

\newcommand{\Gpp}{\Gamma^{\prime\prime}}
\newcommand{\Gppp}{\Gamma^{\prime\prime\prime}}

\newcommand{\hgi}{\ho_\Gamma(\gamma_1)}

\newcommand{\hgnn}{\ho_\Gamma(\gamma_N)}

\newcommand{\idf}{\mathbbm{1}}

\newcommand{\bra}{[}
\newcommand{\ket}{]}

\newcommand{\beq}{\begin{equation}\begin{aligned}}
\newcommand{\beqs}{\begin{equation*}\begin{aligned}}
\newcommand{\be}{\begin{flalign}}
\newcommand{\bes}{\begin{equation*}}
\newcommand{\eq}{\end{aligned}\end{equation}}
\newcommand{\eqs}{\end{aligned}\end{equation*}}
\newcommand{\ee}{\end{flalign}}
\newcommand{\ees}{\end{equation}}

\newcommand{\limi}{\underset{i\rightarrow\infty}{\underrightarrow{\lim}}}

\newcommand{\limPDi}{\underset{\PD_{\Gamma_i}\in \PD}{\underrightarrow{\lim}}}

\newtheorem{theo}{Theorem }[section]
\newtheorem{lem}[theo]{Lemma}
\newtheorem{rem}[theo]{Remark}
\newtheorem{prop}[theo]{Proposition}
\newtheorem{cor}[theo]{Corollary}

\newtheorem{defi}[theo]{Definition}

\newenvironment{proofs}[1][Proof ]{\noindent\textbf{#1}: }{\ \begin{flushright}
                                                                         \rule{0.5em}{0.5em}
                                                                        \end{flushright}}

\newcounter{exa}[section]
 \newenvironment{exa}{\refstepcounter{exa}
  \textbf{Example} \thesection.\arabic{exa}: }{ {\begin{flushright}
                                                                         \rule{0.2em}{0.2em}
                                                                        \end{flushright}}}

\newcounter{problem}[subsection]
 \newenvironment{problem}{\refstepcounter{problem}
  \textbf{Problem} \thesection.\arabic{problem}: }{{\begin{flushright}
                                                                         \rule{0.2em}{0.2em}
                                                                        \end{flushright}}}
\newcommand{\GGi}{\xymatrix{
  \Goid_1  \ar@<-2pt>[r] \ar@<2pt>[r] &  \Goid^0_1    \\
}}
\newcommand{\GGii}{\xymatrix{
  \Goid_2  \ar@<-2pt>[r] \ar@<2pt>[r] &  \Goid^0_2    \\
}}
\newcommand{\GGm}{\xymatrix{
  \Goid  \ar@<-1pt>[r]^{s} \ar@<1pt>[r]_{t} &  \Goid^0    \\
}}
\newcommand{\GGim}{\xymatrix{
  \Goid_1  \ar@<-1pt>[r]^{s_1} \ar@<1pt>[r]_{t_1} &  \Goid^0_1    \\
}}
\newcommand{\GGiim}{\xymatrix{
  \Goid_2  \ar@<-1pt>[r]^{s_2} \ar@<1pt>[r]_{t_2} &  \Goid^0_2    \\
}}

\newcommand{\PGm}{\xymatrix{
  \PD  \ar@<-1pt>[r]^{s} \ar@<1pt>[r]_{t} &  \Sigma    \\
}}

\newcommand{\PGoS}{\PD\rightrightarrows\Sigma}
\newcommand{\fPGm}{\xymatrix{
  \PD_\Gamma  \ar@<-1pt>[r]^{s} \ar@<1pt>[r]_{t} &  V_\Gamma    \\
}}
\newcommand{\PGsm}{\xymatrix{
  \PD\Sigma \ar@<-1pt>[r]^{s_{\PD\Sigma}} \ar@<1pt>[r]_{t_{\PD\Sigma}} &  \Sigma   \\
}}
\newcommand{\fPGms}{\xymatrix{
  \PD^s_\Gamma  \ar@<-1pt>[r]^{s} \ar@<1pt>[r]_{t} &  V_\Gamma    \\
}}

\newcommand{\fPG}{\PD_\Gamma\Sigma \rightrightarrows V_\Gamma
}

\newcommand{\fPSGm}{\xymatrix{
  \PD_\Gamma\Sigma  \ar@<-1pt>[r]^{s} \ar@<1pt>[r]_{t} &  V_\Gamma    \\
}}
\newcommand{\fPSG}{\xymatrix{
  \PD_\Gamma\Sigma  \ar@<-2pt>[r] \ar@<2pt>[r] &  V_\Gamma    \\
}}
\newcommand{\fgHGm}{\xymatrix{
  H(\Gamma)  \ar@<-1pt>[r]^/0.3em/{\hat s_H} \ar@<1pt>[r]_/0.3em/{\hat t_H} &  V_\Gamma    \\
}}

\newcommand{\fHGm}{\xymatrix{
  H_\Gamma  \ar@<-1pt>[r]^/0.3em/{\hat s_H} \ar@<1pt>[r]_/0.3em/{\hat t_H} &  V_\Gamma    \\
}}

\newcommand{\fGGm}{\xymatrix{
  \G^G_\Gamma  \ar@<-1pt>[r]^/0.3em/{s_P} \ar@<1pt>[r]_/0.3em/{t_P} &  V_\Gamma    \\
}}
\newcommand{\fGHm}{\xymatrix{
  \G^H_\Gamma  \ar@<-1pt>[r]^/0.3em/{s_P} \ar@<1pt>[r]_/0.3em/{t_P} &  V_\Gamma    \\
}}

\newcounter{count}
\newcommand{\citetableF}{\cite[table 11.6]{KaminskiPHD}}

\begin{document}
\maketitle
\begin{abstract}\noindent In this article a new $C^*$-algebra derived from the basic quantum variables: holonomies along paths and group-valued quantum flux operators in the framework of Loop Quantum Gravity is constructed. This development is based on the theory of cross-products and $C^*$-dynamical systems. In \cite{Kaminski1} the author has presented a set of actions of the flux group associated to a surface set on the analytic holonomy $C^*$-algebra, which define $C^*$-dynamical systems. These objects are used to define the holonomy-flux cross-product $C^*$-algebra associated to a surface set. Furthermore surface-preserving path- and graph-diffeomorphism-invariant states of the new $C^*$-algebra are analysed. Finally the holonomy-flux cross-product $C^*$-algebra is extended such that the graph-diffeomorphisms generate among other operators the holonomy-flux-graph-diffeomorphism cross-product $C^*$-algebra associated to a surface set. 
\end{abstract}

\thispagestyle{plain}
\pdfbookmark[0]{\contentsname}{toc}
\tableofcontents

\section{Introduction}

In \cite{Kaminski0,KaminskiPHD} the quantum configuration variables, the  quantum momentum variables and the spatial diffeomorphisms have been introduced briefly. These objects have been used to define the \textit{Weyl $C^*$-algebra for surfaces} in \cite{Kaminski1,KaminskiPHD}.  In this article the quantum variables are presented in section \ref{sec quantvar}. The quantum configuration variables are \textit{holonomies along paths in a graph}. The finite set of subgraphs of a graph forms a \textit{finite graph system}. The configuration space is denoted by $\Ab_\Gamma$ and is naturally identified with $G^{\vert\Gamma\vert}$, where $G$ is the structure group of a principal fibre bundle $P(\Sigma,G,\pi)$ and $\vert\Gamma\vert$ denotes the number of independent edges of the graph $\Gamma$. For generality it is assumed that, the structure group $G$ is a unimodular locally compact group. The quantum momentum variables are given by the \textit{group-valued quantum flux operators}, which depend on a surface and a graph. For a certain surface set these operators form a group, which is called the \textit{flux group associated to a surface set and a graph} and which is denoted by $\bar G_{\breve S,\Gamma}$. For a certain fixed surface set this group is identified with $G^M$, where $M \leq \vert\Gamma\vert$.

The Weyl algebra of Quantum Geometry \cite{Fleischhack06} or the Weyl algebra for surfaces \cite{Kaminski1},\cite[Chapt. 6]{KaminskiPHD} are not the only $C^*$-algebras which will be constructed from the quantum configuration and momentum operators of the theory of Loop Quantum Gravity. The significant choice for a construction of the Weyl algebra for surfaces was the requirement of the group-valued quantum flux operators to be unitary Hilbert space operators. If this choice is not made, then the flux operators can be represented on a Hilbert space by the \textit{generalised group-valued quantum flux operators}, which are given by the integrated representations of the flux group associated to a surface set. In this article even more general objects, which are given by algebra-valued functions depending on the flux group associated to a surface set, are introduced. The algebras are derived from these objects are cross-product $C^*$-algebras, which have been studied intensively by Williams \cite{Williams07}, Hewitt and Ross \cite{HewittRoss} or Pedersen \cite{Pedersen}. For a short overview refer to Blackadar \cite{Blackadar}. 

\paragraph*{Algebras derived from either quantum configuration or momentum variables\\}\hspace*{10pt}

Let $G$ be a locally compact unimodular group. To start with consider the quantum momentum space, which contains all flux groups associated to surfaces. A certain flux group $\bar G_{\breve S,\Gamma}$ associated to a fixed surface set $\breve S$ can be identified with $G^{\vert\Gamma\vert}$. 
Then the following algebras are studied in section \ref{subsec fluxgroupalg}:
The \textit{convolution flux $^*$-algebra $\CD(\bar G_{\breve S,\Gamma})$ associated to a surface set and a graph}. This algebra is for a  in general a non-commutative $^*$-algebra. Moreover, the \textit{flux group $C^*$-algebra $C^*(\bar G_{\breve S,\Gamma})$ associated to a surface set and a graph} is derived from the generalised group-valued quantum flux operators. 
Finally, a particular cross-product $C^*$-algebra is given by the \textit{flux transformation group $C^*$-algebra $C^*(\bar G_{\breve S,\Gamma},\bar G_{\breve S,\Gamma})$ associated to a surface set and a graph}, which contains algebra-valued functions on the flux group $\bar G_{\breve S,\Gamma}$. 

It is assumed that, the configuration space $\Ab_\Gamma$ restricted to a fixed graph system $\PD_\Gamma$ is naturally identified with $G^{\vert\Gamma\vert}$. Then the \textit{convolution holonomy $^*$-algebra $\CD(\Ab_\Gamma)$ associated to a graph}, the \textit{non-commutative holonomy $C^*$-algebra $C^*(\Ab_\Gamma)$ associated to a graph} and the \textit{heat-kernel holonomy $C^*$-algebra $C^*(\Ab_\Gamma,\Ab_\Gamma)$ associated to a graph} is introduced. Note that, the analytic holonomy $C^*$-algebra $C(\Ab_\Gamma)$ differs from the non-commutative holonomy $C^*$-algebra $C^*(\Ab_\Gamma)$ in the multiplication operation and involution.

The construction of cross-products for the quantum configuration variables restricted to a graph is related to the observation, which has been noticed by Ashtekar and Lewandowski \cite{AshLewDG95}. In the context of heat kernels the authors Lewandowski and Ashtekar \cite[section 6.2]{AshLewDG95} have presented an object, which can be understood as a generalised heat kernel representation $\pi_I^H$ of the non-commutative holonomy $C^*$-algebra $C^*(\Ab_\Gamma)$ associated to a graph $\Gamma$ on the Hilbert space $\HS_\Gamma:=L^2(\Ab_\Gamma,\dif\mu_\Gamma)$. This representation is given by
\beq \pi_I^{H}(\rho_{t,\Gamma})\psi_\Gamma
&=\int_{\Ab_\Gamma}\dif\mu_{\Gamma}(\hat\ho_\Gamma)\rho_{t,\Gamma}(\hat\ho_\Gamma^{-1}\ho_\Gamma)\psi_\Gamma(\ho_\Gamma)\\
&= \rho_{t,\Gamma}\ast \psi_\Gamma
\eq where $\hat\ho_\Gamma,\ho_\Gamma$ are two different holonomies along paths of a graph $\Gamma$, $\rho_{t,\Gamma}\in C^*(\Ab_\Gamma)$ and $\psi_\Gamma\in\HS_\Gamma$. 

The inductive limit of the inductive family $\{C^*(\Ab_\Gamma), \beta_{\Gamma,\Gp}\}$ of $C^*$-algebras is called the \textit{non-commutative holonomy $C^*$-algebra} $C^*(\Ab)$. Furthermore the inductive limit $C^*$-algebra of the inductive family\\ $\{C^*(\Ab_\Gamma,\Ab_\Gamma),\tilde\beta_{\Gamma,\Gp}\}$ is called the \textit{heat-kernel holonomy $C^*$-algebra} $C^*(\Ab,\Ab)$. 

\paragraph*{Algebras derived from quantum configuration and momentum variables\\}\hspace*{10pt}

In section \ref{subsec holflux} algebras are constructed from holonomy along paths and group-valued quantum fluxes by using $C^*$-dynamical systems. The concept of $C^*$-dynamical systems in LQG has been introduced in \cite{Kaminski1,KaminskiPHD}. In particular the analytic holonomy $C^*$-algebra restricted to a finite graph system $\PD_\Gamma$ and an action $\alpha$ of a certain flux group $\bar G_{\breve S,\Gamma}$ associated to a surface set $\breve S$ on this algebra, form a $C^*$-dynamical system. In this article different holonomy-flux cross-product $C^*$-algebras associated to certain surface sets are constructed for these $C^*$-dynamical systems.

Then the \textit{holonomy-flux cross-product $C^*$-algebra $C_0(\Ab_\Gamma)\rtimes_\alpha \bar G_{\breve S,\Gamma}$  associated to a surface set and a graph} contains $C_0(\Ab_\Gamma)$-valued functions depending on the flux group $\bar G_{\breve S,\Gamma}$. Let $G$ be a compact group. Then there is an limit $C^*$-algebra of the inductive family $\{C(\Ab_\Gamma)\rtimes_\alpha \bar G_{\breve S,\Gamma},\hat\beta_{\Gamma,\Gp}\}$, which is called the the \textit{holonomy-flux cross-product $C^*$-algebra} $C(\Ab)\rtimes_\alpha \bar G_{\breve S}$ associated to a surface set. 

Since there are many different $C^*$-dynamical systems presented in \cite{Kaminski1,KaminskiPHD}, there are a lot of different holonomy-flux cross-product $C^*$-algebras associated to suitable surface sets. These algebras are compared with the Weyl $C^*$-algebra for surfaces in section \ref{end}. In particular in theorem \ref{prop multilpiercrossprod} it is proved that, the \textit{multiplier algebra of the holonomy-flux cross-product $C^*$-algebra associated to a certain surface set and a graph $\Gamma$} contains all operators, which are contained in the other cross-product $C^*$-algebras asociated to suitable surface sets and the graph $\Gamma$, in the analytic holonomy $C^*$-algebra restricted to the finite graph system $\PD_\Gamma$ and all Weyl elements for suitable surface sets and the graph $\Gamma$. 

All algebras presented in the previous paragraphs are constructed from the basic quantum variables, which are given by the holonomy along paths and the group-valued quantum fluxes. Hence they are possible algebras of a quantum theory of gravity. 

\paragraph*{Simplified algebras derived from certain quantum configuration and momentum variables\\}\hspace*{10pt}

If both quantum variables: the quantum configuration and momentum variables restricted to a fixed graph $\Gamma$ and a fixed suitable surface set $\breve S$ are considered simultaneously, then the following simplifications can be studied. 

In section \ref{subsec transformalg} the flux transformation $C^*$-algebra associated to a surface set and a graph is presented. Similarly the flux group of a fixed graph $\Gamma$ and a fixed suitable surface set $\breve S$ and the configuration space $\Ab_\Gamma$ are identified with $G^{\vert\Gamma\vert}$. Hence in both cases the cross-product $C^*$-algebras are simplified to $C_0(G^{\vert\Gamma\vert})\rtimes_\alpha G^{\vert\Gamma\vert}$. 

Then it is verified in theorem \ref{Generalised Stone- von Neumann theorem} that the cross-product $C^*$-algebra $C_0(G^{\vert\Gamma\vert})\rtimes_\alpha G^{\vert\Gamma\vert}$ is Morita equivalent to the $C^*$-algebra of compact operators on the Hilbert space $L^2(G^{\vert\Gamma\vert},\mu_\Gamma)$, where $\mu_\Gamma$ denotes the product of Haar measures. Therefore the representation theory of both $C^*$-algebras is the same and, hence, there is only one irreducible representation of the cross-product $C^*$-algebra up to unitary equivalence. 

But this identification is only true for certain surface sets. The cross-product $C^*$-algebra is derived from the quantum momentum variables, which depend on different surface sets. In particular the flux group associated to a suitable surface set is identified with a product group $G^M$ where $M\leq \vert\Gamma\vert$. Then there exists a purely left (or right) action of $G^M$ on the $C^*$-algebra $C_0(G^{\vert\Gamma\vert})$. For $M< \vert\Gamma\vert$ a Morita equivalent $C^*$-algebra is not found in this project. In theorem \ref{theo moritaequivgroup} a Morita equivalent algebra for the $C^*$-algebra $C_0(G^{N})\rtimes_\alpha G^{M}$ whenever $N<M$, is given.

In this article the general case of arbitrary surfaces is studied. Hence the quantum configuration and the momentum variables of the theory are manifestly distinguished from each other. The quantum configuration variables only depends on graphs and holonomy mappings, whereas the quantum momentum variables depend on graphs, maps from graphs to products of the structure group and the intersection behavior of the paths of the graphs and surfaces. But nevertheless the elements of the holonomy-flux cross-product $C^*$-algebra are understood as compact operators on the flux group associated to a surface set with values in the analytic holonomy $C^*$-algebra restricted to a graph, which are acting on the Hilbert space $L^2(\Ab_\Gamma,\mu_\Gamma)$. 

\paragraph*{States of algebras derived from certain quantum configuration and momentum variables\\}\hspace*{10pt}

There exists several inductive limit holonomy-flux cross-product $C^*$-algebras, which are given by the inductive families of holonomy-flux cross-product $C^*$-algebras associated to graphs and a suitable surface set. The states on these algebras always depend on the choice of the surface set and, hence, they are not path- or graph-diffeomorphism invariant. 

\paragraph*{Algebras derived from quantum configuration and momentum variables and quantum spatial diffeomorphisms\\}\hspace*{10pt}

In the last paragraphs new $C^*$-algebras of a special kind have been constructed. All these algebras are based on new operators, which are more general than group-valued quantum flux operators and which take in particluar values in the analytic holonomy $C^*$-algebra. Until now the quantum diffeomorphisms are implemented only as automorphisms on these algebras. In section \ref{subsec holfluxdiffcrossalg} one of the previous algebras is extended such that functions on the group of bisections of a finite graph system to the holonomy-flux cross-product $C^*$-algebra, form this new $C^*$-algebra.

The cross-product $C^*$-algebra construction is particularly based on $C^*$-dynamical systems. In the article \cite[Section 3.2]{Kaminski1},\cite[Section 6.2]{KaminskiPHD} it has been argued that, the action of the group of bisections of a finite graph system on the analytic holonomy $C^*$-algebra restricted to a finite graph system defines a $C^*$-dynamical system, too. Furthermore there is also an action of the group of certain bisections of a finite graph system on the holonomy-flux cross-product $C^*$-algebra associated to the surface set $\breve S$ and a graph. These objects define another $C^*$-dynamical system and a new cross-product $C^*$-algebra, which is called the \textit{holonomy-flux-graph-diffeomorphism cross-product $C^*$-algebra}.

There exists a covariant representation of this $C^*$-dynamical system on a Hilbert space. This pair is given by a unitary representation of the group of surface-orientation-preserving bisections of a finite graph system on the Hilbert space $\HS_\Gamma$ and the multiplication representation $\Phi_M$ of the analytic holonomy $C^*$-algebra restricted to the finite graph system $\PD_\Gamma$ on $\HS_\Gamma$. The unitaries are called the \textit{unitary bisections of a finite graph system and surfaces} in the project \textit{AQV}. Then each unitary bisections of a finite graph system and surfaces is contained in the \textit{multiplier algebra of the holonomy-flux-graph-diffeomorphism cross-product $C^*$-algebra associated to a graph and the surface set}. The remarkable point is that the multiplier algebra of the holonomy-flux cross-product $C^*$-algebra associated to a graph and the surface set does not contain these unitaries. 

In general the multiplier algebra of the holonomy-flux cross-product $C^*$-algebra associated to a fixed surface set contains all operators of holonomy-flux cross-product $C^*$-algebra for other suitable surface sets, elements of the analytic holonomy $C^*$-algebra and all Weyl elements associated to other suitable surface sets. The Weyl $C^*$-algebra for surfaces contains elements of the analytic holonomy $C^*$-algebra and all Weyl elements. The multiplier algebra of the holonomy-flux cross-product $C^*$-algebra associated to the surface set $\breve S$ contains the Weyl algebra for suitable surface sets. The Lie algebra-valued quantum flux operators and the right-invariant vector fields are affiliated with the holonomy-flux cross-product $C^*$-algebra, but they are not affiliated with the Weyl $C^*$-algebra for surfaces. For a detailed overview about the multiplier algebras and affiliated elements with the $C^*$-algebras of quantum variables refer to \citetableF.\\

\section{The basic quantum operators}\label{sec quantvar}
\subsection{Finite path groupoids and graph systems}\label{subsec fingraphpathgroup}

Let $c:[0,1]\rightarrow\Sigma$ be continuous curve in the domain $\bra 0,1\ket$, which is (piecewise) $C^k$-differentiable ($1\leq k\leq \infty$), analytic ($k=\omega$) or semi-analytic ($k=s\omega$) in $\bra 0,1\ket$ and oriented such that the source vertex is $c(0)=s(c)$ and the target vertex is $c(1)=t(c)$. Moreover assume that, the range of each subinterval of the curve $c$ is a submanifold, which can be embedded in $\Sigma$. An \textbf{edge} is given by a \hyperlink{rep-equiv}{reparametrisation invariant} curve of class (piecewise) $C^k$. The maps $s_{\Sigma},t_{\Sigma}:P\Sigma\rightarrow\Sigma$ where $P\Sigma$ is the path space are surjective maps and are called the source or target map.    

A set of edges $\{e_i\}_{i=1,...,N}$ is called \textbf{independent}, if the only intersections points of the edges are source $s_{\Sigma}(e_i)$ or $t_{\Sigma}(e_i)$ target points. Composed edges are called \textbf{paths}. An \textbf{initial segment} of a path $\gamma$ is a path $\gamma_1$ such that there exists another path $\gamma_2$ and $\gamma=\gamma_1\circ\gamma_2$. The second element $\gamma_2$ is also called a \textbf{final segment} of the path $\gamma$.

\begin{defi}
A \textbf{graph} $\Gamma$ is a union of finitely many independent edges $\{e_i\}_{i=1,...,N}$ for $N\in\N$. The set $\{e_1,...,e_N\}$ is called the \textbf{generating set for $\Gamma$}. The number of edges of a graph is denoted by $\vert \Gamma\vert$. The elements of the set $V_\Gamma:=\{s_{\Sigma}(e_k),t_{\Sigma}(e_k)\}_{k=1,...,N}$ of source and target points are called \textbf{vertices}.
\end{defi}

A graph generates a finite path groupoid in the sense that, the set $\PD_\Gamma\Sigma$ contains all independent edges, their inverses and all possible compositions of edges. All the elements of $\PD_\Gamma\Sigma$ are called paths associated to a graph. Furthermore the surjective source and target maps $s_{\Sigma}$ and $t_{\Sigma}$ are restricted to the maps $s,t:\PD_\Gamma\Sigma\rightarrow V_\Gamma$, which are required to be surjective.

\begin{defi}\label{path groupoid} Let $\Gamma$ be a graph. Then a \textbf{finite path groupoid} $\PD_\Gamma\Sigma$ over $V_\Gamma$ is a pair $(\PD_\Gamma\Sigma, V_\Gamma)$ of finite sets equipped with the following structures: 
\begin{enumerate}
 \item two surjective maps \(s,t:\PD_\Gamma\Sigma\rightarrow V_\Gamma\), which are called the source and target map,
\item the set \(\PD_\Gamma\Sigma^2:=\{ (\gamma_i,\gamma_j)\in\PD_\Gamma\Sigma\times\PD_\Gamma\Sigma: t(\gamma_i)=s(\gamma_j)\}\) of finitely many composable pairs of paths,
\item the  composition \(\circ :\PD_\Gamma^2\Sigma\rightarrow \PD_\Gamma\Sigma,\text{ where }(\gamma_i,\gamma_j)\mapsto \gamma_i\circ \gamma_j\), 
\item the inversion map \(\gamma_i\mapsto \gamma_i^{-1}\) of a path,
\item the object inclusion map \(\iota:V_\Gamma\rightarrow\PD_\Gamma\Sigma\) and
\item $\PD_\Gamma\Sigma$ is defined by the set $\PD_\Gamma\Sigma$ modulo the algebraic equivalence relations generated by
\beq\label{groupoid0} \gamma_i^{-1}\circ \gamma_i\simeq \idf_{s(\gamma_i)}\text{ and }\gamma_i\circ \gamma_i^{-1}\simeq \idf_{t(\gamma_i)}
\eq 
\end{enumerate}
Shortly write $\fPSGm$. 
\end{defi} 
Clearly, a graph $\Gamma$ generates freely the paths in $\PD_\Gamma\Sigma$. Moreover the map $s \times t: \PD_\Gamma\Sigma\rightarrow V_\Gamma\times V_\Gamma$ defined by $(s\times t)(\gamma)=(s(\gamma),t(\gamma))$ for all $\gamma\in\PD_\Gamma\Sigma$ is assumed to be surjective ($\PD_\Gamma\Sigma$ over $V_\Gamma$ is a transitive groupoid), too. 

A general groupoid $\GG$ over $\GG^{0}$ defines a small category where the set of morphisms is denoted in general by $\GG$ and the set of objects is denoted by $\GG^{0}$. Hence in particular the path groupoid can be viewed as a category, since,
\begin{itemize}
\item the set of morphisms is identified with $\PD_\Gamma\Sigma$,
\item the set of objects is given by $V_\Gamma$ (the units) 
\end{itemize}

From the condition (\ref{groupoid0}) it follows that, the path groupoid satisfies additionally 
\begin{enumerate}
 \item $ s(\gamma_i\circ \gamma_j)=s(\gamma_i)\text{ and } t(\gamma_i\circ \gamma_j)=t(\gamma_j)\text{ for every } (\gamma_i,\gamma_j)\in\PD_\Gamma^2\Sigma$
\item $s(v)= v= t(v)\text{ for every } v\in V_\Gamma$
\item\label{groupoid1} $ \gamma \circ\idf_{s(\gamma)} = \gamma = \idf_{t(\gamma)}\circ \gamma\text{ for every } \gamma\in \PD_\Gamma\Sigma\text{ and }$
\item $\gamma \circ (\gamma_i\circ \gamma_j)=(\gamma \circ \gamma_i) \circ \gamma_j$
\item $\gamma \circ (\gamma^{-1}\circ \gamma_1)=\gamma_1= (\gamma_1 \circ \gamma) \circ \gamma^{-1}$
\end{enumerate}

The condition \ref{groupoid1} implies that the vertices are units of the groupoid. 

\begin{defi}
Denote the set of all finitely generated paths by
\beqs \PD_\Gamma\Sigma^{(n)}:=\{(\gamma_1,...,\gamma_n)\in \PD_\Gamma\times ...\PD_\Gamma: (\gamma_i,\gamma_{i+1})\in\PD^{(2)}, 1\leq i\leq n-1 \}\eqs
The set of paths with source point $v\in V_\Gamma$ is given by
\beqs \PD_\Gamma\Sigma^{v}:=s^{-1}(\{v\})\eqs
The set of paths with target  point $v\in V_\Gamma$ is defined by
\beqs \PD_\Gamma\Sigma_{v}:=t^{-1}(\{v\})\eqs
The set of paths with source point $v\in V_\Gamma$ and target point $u\in V_\Gamma$ is 
\beqs \PD_\Gamma\Sigma^{v}_u:=\PD_\Gamma\Sigma^{v}\cap \PD_\Gamma\Sigma_{u}\eqs
\end{defi}

A graph $\Gamma$ is said to be \hypertarget{disconnected}{\textbf{disconnected}} if it contains only mutually pairs $(\gamma_i,\gamma_j)$ of non-composable independent paths $\gamma_i$ and $\gamma_j$ for $i\neq j$ and $i,j=1,...,N$. In other words for all $1\leq i,l\leq N$ it is true that $s(\gamma_i)\neq t(\gamma_l)$ and $t(\gamma_i)\neq s(\gamma_l)$ where $i\neq l$ and $\gamma_i,\gamma_l\in\Gamma$.

\begin{defi}
Let $\Gamma$ be a graph. A \textbf{subgraph $\Gp$ of $\Gamma$} is a given by a finite set of independent paths in $\PD_\Gamma\Sigma$. 
\end{defi}
For example let $\Gamma:=\{\gamma_1,...,\gamma_N\}$ then $\Gp:=\{\gamma_1\circ\gamma_2,\gamma_3^{-1},\gamma_4\}$ where $\gamma_1\circ\gamma_2,\gamma_3^{-1},\gamma_4\in\PD_\Gamma\Sigma$ is a subgraph of $\Gamma$, whereas the set $\{\gamma_1,\gamma_1\circ\gamma_2\}$ is not a subgraph of $\Gamma$. Notice if additionally $(\gamma_2,\gamma_4)\in\PD_\Gamma^{(2)}$ holds, then $\{\gamma_1,\gamma_3^{-1},\gamma_2\circ\gamma_4\}$ is a subgraph of $\Gamma$, too. Moreover for $\Gamma:=\{\gamma\}$ the graph $\Gamma^{-1}:=\{\gamma^{-1}\}$ is a subgraph of $\Gamma$. As well the graph $\Gamma$ is a subgraph of $\Gamma^{-1}$. A subgraph of $\Gamma$ that is generated by compositions of some paths, which are not reversed in their orientation, of the set $\{\gamma_1,...,\gamma_N\}$ is called an \textbf{orientation preserved subgraph of a graph}. For example for $\Gamma:=\{\gamma_1,...,\gamma_N\}$ orientation preserved subgraphs are given by $\{\gamma_1\circ\gamma_2\}$, $\{\gamma_1,\gamma_2,\gamma_N\}$ or $\{\gamma_{N-2}\circ\gamma_{N-1}\}$ if $(\gamma_1,\gamma_2)\in\PD_\Gamma\Sigma^{(2)}$ and $(\gamma_{N-2},\gamma_{N-1})\in\PD_\Gamma\Sigma^{(2)}$.   

\begin{defi} 
A \textbf{finite graph system $\PD_\Gamma$ for $\Gamma$} is a finite set of subgraphs of a graph $\Gamma$. A finite graph system $\PD_{\Gp}$ for $\Gp$ is a \hypertarget{finite graph subsystem}{\textbf{finite graph subsystem}} of $\PD_\Gamma$ for $\Gamma$ if the set $\PD_{\Gp}$ is a subset of $\PD_{\Gamma}$ and $\Gp$ is a subgraph of $\Gamma$. Shortly write $\PD_{\Gp}\leq\PD_{\Gamma}$.

A \hypertarget{finite orientation preserved graph system}{\textbf{finite orientation preserved graph system}} $\PD^{\op}_\Gamma$ for $\Gamma$ is a finite set of orientation preserved subgraphs of a graph $\Gamma$. 
\end{defi}

Recall that, a finite path groupoid is constructed from a graph $\Gamma$, but a set of elements of the path groupoid need not be a graph again. For example let $\Gamma:=\{\gamma_1\circ\gamma_2\}$ and $\Gp=\{\gamma_1\circ\gamma_3\}$, then $\Gpp=\Gamma\cup\Gp$ is not a graph, since this set is not independent. Hence only appropriate unions of paths, which are elements of a fixed finite path groupoid, define graphs. The idea is to define a suitable action on elements of the path groupoid, which corresponds to an action of diffeomorphisms on the manifold $\Sigma$. The action has to be transfered to graph systems. But the action of bisection, which is defined by the use of the groupoid multiplication, cannot easily generalised for graph systems. 

\begin{problem}\label{problem group structure on graphs systems}
Let $\breve\Gamma:=\{\Gamma_i\}_{i=1,..,N}$ be a finite set such that each $\Gamma_i$ is a set of not necessarily independent paths such that 
\begin{enumerate}
\item the set contains no loops and
\item each pair of paths satisfies one of the following conditions
\begin{itemize}
\item the paths intersect each other only in one vertex,
\item the paths do not intersect each other or
\item one path of the pair is a segment of the other path.
\end{itemize}
\end{enumerate}

Then there is a map $\circ:\breve\Gamma\times \breve\Gamma\rightarrow\breve\Gamma$ of two elements $\Gamma_1$ and $\Gamma_2$ defined by
\beqs \{\gamma_1,...,\gamma_M\}\circ\{\tg_1,...,\tg_M\}:= &\Big\{ \gamma_i\circ\tg_j:t(\gamma_i)=s(\tg_j)\Big\}_{1\leq i,j\leq M}\\
\eqs for $\Gamma_1:=\{\gamma_1,...,\gamma_M\},\Gamma_2:=\{\tg_1,...,\tg_M\}$. 
Moreover define a map $^{-1}:\breve\Gamma\rightarrow\breve\Gamma$ by
\beqs  \{\gamma_1,...,\gamma_M\}^{-1}:= \{\gamma^{-1}_1,...,\gamma^{-1}_M\}\eqs 

Then the following is derived
\beqs \{\gamma_1,...,\gamma_M\}\circ\{\gamma^{-1}_1,...,\gamma^{-1}_M\}&=\Big\{ \gamma_i\circ\gamma^{-1}_j: t(\gamma_i)=t(\gamma_j)\Big\}_{1\leq i,j\leq M}\\
&=\Big\{ \gamma_i\circ\gamma^{-1}_j:t(\gamma_i)=t(\gamma_j)\text{ and }i\neq j\Big\}_{1\leq i,j\leq M}\\
&\quad\cup\{\idf_{s_{\gamma_j}}\}_{1\leq j\leq M}\\
\neq &\quad\cup\{\idf_{s_{\gamma_j}}\}_{1\leq j\leq M}
\eqs The equality is true, if the set $\breve\Gamma$ contains only graphs such that all paths are mutually non-composable. Consequently this does not define a well-defined multiplication map. Notice that, the same is discovered if a similar map and inversion operation are defined for a finite graph system $\PD_\Gamma$. 
\end{problem}

Consequently the property of paths being independent need not be dropped for the definition of a suitable multiplication and inversion operation. In fact the independence property is a necessary condition for the construction of the holonomy algebra for analytic paths. Only under this circumstance each analytic path is decomposed into a finite product of independent piecewise analytic paths again. 

\begin{defi}
A finite path groupoid $\PD_{\Gp}\Sigma$ over $V_{\Gp}$ is a \textbf{finite path subgroupoid} of $\PD_{\Gamma}\Sigma$ over $V_\Gamma$ if the set $V_{\Gp}$ is contained in $V_\Gamma$ and the set $\PD_{\Gp}\Sigma$ is a subset of $\PD_{\Gamma}\Sigma$. Shortly write $\PD_{\Gp}\Sigma\leq\PD_{\Gamma}\Sigma$.
\end{defi}

Clearly for a subgraph $\Gamma_1$ of a graph $\Gamma_2$, the associated path groupoid $\PD_{\Gamma_1}\Sigma$ over $V_{\Gamma_1}$ is a subgroupoid of $\PD_{\Gamma_2}\Sigma$ over $V_{\Gamma_2}$.  This is a consequence of the fact that, each path in $\PD_{\Gamma_1}\Sigma$ is a composition of paths or their inverses in $\PD_{\Gamma_2}\Sigma$. 

\begin{defi}
A \textbf{family of finite path groupoids} $\{\PD_{\Gamma_i}\Sigma\}_{i=1,...,\infty}$, which is a set of finite path groupoids $\PD_{\Gamma_i}\Sigma$ over $V_{\Gamma_i}$, is said to be \textbf{inductive} if for any $\PD_{\Gamma_1}\Sigma,\PD_{\Gamma_2}\Sigma$ exists a $\PD_{\Gamma_3}\Sigma$ such that $\PD_{\Gamma_1}\Sigma,\PD_{\Gamma_2}\Sigma\leq\PD_{\Gamma_3}\Sigma$.

A \textbf{family of graph systems} $\{\PD_{\Gamma_i}\}_{i=1,...,\infty}$, which is a set of finite path systems $\PD_{\Gamma_i}$ for $\Gamma_i$, is said to be \textbf{inductive} if for any $\PD_{\Gamma_1},\PD_{\Gamma_2}$ exists a $\PD_{\Gamma_3}$ such that $\PD_{\Gamma_1},\PD_{\Gamma_2}\leq \PD_{\Gamma_3}$.
\end{defi}

\begin{defi}
Let $\{\PD_{\Gamma_i}\Sigma\}_{i=1,...,\infty}$ be an inductive family of path groupoids and $\{\PD_{\Gamma_i}\}_{i=1,...,\infty}$ be an inductive family of graph systems.

The \textbf{inductive limit path groupoid $\PD$ over $\Sigma$} of an inductive family of finite path groupoids such that $\PD:=\limi\PD_{\Gamma_i}\Sigma$ is called the \textbf{(algebraic) path groupoid} $\PGoS$.

Moreover there exists an \textbf{inductive limit graph $\Gamma_\infty$} of an inductive family of graphs such that $\Gamma_\infty:=\limi \Gamma_i$.

The \textbf{inductive limit graph system} $\PD_{\Gamma_\infty}$ of an inductive family of graph systems such that $\PD_{\Gamma_\infty}:=\limi \PD_{\Gamma_i}$
\end{defi}

Assume that, the inductive limit $\Gamma_\infty$ of a inductive family of graphs is a graph, which consists of an infinite countable number of independent paths. The inductive limit $\PD_{\Gamma_\infty}$ of a inductive family $\{\PD_{\Gamma_i}\}$ of finite graph systems contains an infinite countable number of subgraphs of $\Gamma_\infty$ and each subgraph is a finite set of arbitrary independent paths in $\Sigma$. 

\subsection{Holonomy maps for finite path groupoids, graph systems and transformations}\label{subsec holmapsfinpath}
In section \ref{subsec fingraphpathgroup} the concept of finite path groupoids for analytic paths has been given. Now the holonomy maps are introduced for finite path groupoids and finite graph systems. The ideas are familar with those presented by Thiemann \cite{Thiembook07}. But for example the finite graph systems have not been studied before. Ashtekar and Lewandowski \cite{AshLew93} have defined the analytic holonomy $C^*$-algebra, which they have based on a finite set of independent hoops. The hoops are generalised for path groupoids and the independence requirement is implemented by the concept of finite graph systems. 

\subsubsection{Holonomy maps for finite path groupoids}\label{subsubsec holmap}

\paragraph*{Groupoid morphisms for finite path groupoids}\hspace{10pt} 

Let $\GGim, \GGiim$ be two arbitrary groupoids.

\begin{defi}
A \hypertarget{groupoid-morphism}{\textbf{groupoid morphism}} between two groupoids $\GG_1$ and $\GG_2$ consists of two maps  $\ho:\GG_1\rightarrow\GG_2$  and $h:\GG_1^0\rightarrow\GG_2^0$ such that
\beqs (\hypertarget{G1}{G1})\qquad \ho(\gamma\circ\gp)&= \ho(\gamma)\ho(\gp)\text{ for all }(\gamma,\gp)\in \GG_1^{(2)}\eqs
\beqs (\hypertarget{G2}{G2})\qquad s_{2}(\ho(\gamma))&=h(s_{1}(\gamma)),\quad t_2(\ho(\gamma))=h(t_{1}(\gamma))\eqs 
 
A \textbf{strong groupoid morphism} between two groupoids $\GG_1$ and $\GG_2$ additionally satisfies
\beqs (\hypertarget{SG2}{SG})\qquad \text{ for every pair }(\ho(\gamma),\ho(\gp))\in\GG_2^{(2)}\text{ it follows that }(\gamma,\gp)\in \GG_1^{(2)}\eqs
\end{defi}

Let $G$ be a Lie group. Then $G$ over $e_G$ is a groupoid, where the group multiplication $\cdot: G^2\rightarrow G$ is defined for all elements  $g_1,g_2,g\in G$ such that $g_1\cdot g_2 = g$. A groupoid morphism between a finite path groupoid $\PD_\Gamma\Sigma$ to $G$ is given by the maps
\[\ho_\Gamma: \PD_\Gamma\Sigma\rightarrow G,\quad h_\Gamma:V_\Gamma\rightarrow e_G \] Clearly
\beq \ho_\Gamma(\gamma\circ\gp)&= \ho_\Gamma(\gamma)\ho_\Gamma(\gp)\text{ for all }(\gamma,\gp)\in \PD_\Gamma\Sigma^{(2)}\\
s_G(\ho_\Gamma(\gamma))&=h_\Gamma(s_{\PD_\Gamma\Sigma}(\gamma)),\quad t_G(\ho_\Gamma(\gamma))=h_\Gamma(t_{\PD_\Gamma\Sigma}(\gamma))
\eq But for an arbitrary pair $(\ho_\Gamma(\gamma_1),\ho_\Gamma(\gamma_2))=:(g_1,g_2)\in G^{(2)}$ it does not follows that, $(\gamma_1,\gamma_2)\in \PD_\Gamma\Sigma^{(2)}$ is true. Hence $\ho_\Gamma$ is not a strong groupoid morphism.

\begin{defi}\label{def sameholanal}Let $\fPG$ be a finite path groupoid.

Two paths $\gamma$ and $\gp$ in $\PD_\Gamma\Sigma$ have the \textbf{same-holonomy for all connections} iff 
\beqs \ho_\Gamma(\gamma)=\ho_\Gamma(\gp)\text{ for all }&(\ho_\Gamma,h_\Gamma)\text{ groupoid morphisms }\\ & \ho_\Gamma:\PD_\Gamma\Sigma\rightarrow G, h:V_\Gamma\rightarrow\{e_G\}
\eqs Denote the relation by $\sim_{\text{s.hol.}}$.
\end{defi}
\begin{lem}
The same-holonomy for all connections relation is an equivalence relation. 
\end{lem}
Notice that, the quotient of the finite path groupoid and the same-holonomy relation for all connections replace the hoop group, which has been used in \cite{AshLew93}.
\begin{defi}\label{genrestgroupoidforgraph}
Let $\fPG$ be a finite path groupoid modulo same-holonomy for all connections equivalence.

A \hypertarget{holonomy map for a finite path groupoid}{\textbf{holonomy map for a finite path groupoid}} $\PD_\Gamma\Sigma$ over $V_\Gamma$ is a groupoid morphism consisting of the maps $(\ho_\Gamma,h_\Gamma)$, where
\(\ho_\Gamma:\PD_\Gamma\Sigma\rightarrow G,h_\Gamma:V_\Gamma\rightarrow \{e_G\}\). 
The set of all holonomy maps is abbreviated by $\Hom(\PD_\Gamma\Sigma,G)$.
\end{defi}

For a short notation observe the following.
In further sections it is always assumed that, the finite path groupoid $\fPG$ is considered modulo same-holonomy for all connections equivalence although it is not stated explicitly.

\subsubsection{Holonomy maps for finite graph systems}\label{subsec graphhol}

Ashtekar and Lewandowski \cite{AshLew93} have presented the loop decomposition into a finite set of independent hoops (in the analytic category). This structure is replaced by a graph, since a graph is a set of independent edges. Notice that, the set of hoops that is generated by a finite set of independent hoops, is generalised to the set of finite graph systems. A finite path groupoid is generated by the set of edges, which defines a graph $\Gamma$, but a set of elements of the path groupoid need not be a graph again. The appropriate notion for graphs constructed from sets of paths is the finite graph system, which is defined in section \ref{subsec fingraphpathgroup}. Now the concept of holonomy maps is generalised for finite graph systems. Since the set, which is generated by a finite number of independent edges, contains paths that are composable, there are two possibilities to identify the image of the holonomy map for a finite graph system on a fixed graph with a subgroup of $G^{\vert\Gamma\vert}$. One way is to use the generating set of independend edges of a graph, which has been also used in \cite{AshLew93}. On the other hand, it is also possible to identify each graph with a disconnected subgraph of a fixed graph, which is generated by a set of independent edges. Notice that, the author implements two situations. One case is given by a set of paths that can be composed further and the other case is related to paths that are not composable. This is necessary for the definition of an action of the flux operators. Precisely the identification of the image of the holonomy maps along these paths is necessary to define a well-defined action of a flux element on the configuration space. This issue has been studied in \cite{Kaminski1,KaminskiPHD}.

First of all consider a graph $\Gamma$ that is generated by the set $\{\gamma_1,...,\gamma_N\}$ of edges. Then each subgraph of a graph $\Gamma$ contain paths that are composition of edges in $\{\gamma_1,...,\gamma_N\}$ or inverse edges. For example the following set $\Gp:=\{\gamma_1\circ\gamma_2\circ\gamma_3,\gamma_4\}$ defines a subgraph of $\Gamma:=\{\gamma_1,\gamma_2,\gamma_3,\gamma_4\}$. Hence there is a natural identification available.

\begin{defi}
A subgraph $\Gp$ of a graph $\Gamma$ is always generated by a subset $\{\gamma_1,...,\gamma_M\}$ of the generating set $\{\gamma_1,...,\gamma_N\}$ of independent edges that generates the graph $\Gamma$. Hence each subgraph is identified with a subset of $\{\gamma_1^{\pm 1},...,\gamma_N^{\pm 1}\}$. This is called the \hypertarget{natural identification}{\textbf{natural identification of subgraphs}}.
\end{defi}

\begin{exa}\label{exa natidentif}
For example consider a subgraph $\Gp:=\{\gamma_1\circ\gamma_2,\gamma_3\circ\gamma_4,...,\gamma_{M-1}\circ\gamma_M\}$, which is identified naturally with a set $\{\gamma_1,...,\gamma_M\}$. The set $\{\gamma_1,...,\gamma_M\}$ is a subset of $\{\gamma_1,...,\gamma_N\}$ where $N=\vert \Gamma\vert$ and $M\leq N$. 

Another example is given by the graph $\Gpp:=\{\gamma_1,\gamma_2\}$ such that $\gamma_2=\gpe\circ\gpz$, then $\Gpp$ is identified naturally with $\{\gamma_1,\gpe,\gpz\}$. This set is a subset of $\{\gamma_1,\gpe,\gpz,\gamma_3,...,\gamma_{N-1}\}$. 
\end{exa}

\begin{defi}
Let $\Gamma$ be a graph, $\PD_\Gamma$ be the finite graph system. Let $\Gp:=\{\gamma_1,...,\gamma_M\}$be a subgraph of $\Gamma$.

A \hypertarget{holonomy map for a finite graph system}{\textbf{holonomy map for a finite graph system}} $\PD_\Gamma$ is a given by a pair of maps $(\ho_\Gamma,h_\Gamma)$ such that there exists a holonomy map\footnote{In the work the holonomy map for a finite graph system and the holonomy map for a finite path groupoid is denoted by the same pair $(\ho_\Gamma,h_\Gamma)$.} $(\ho_\Gamma,h_\Gamma)$ for the finite path groupoid $\fPG$ and
\beqs &\ho_\Gamma:\PD_\Gamma\rightarrow G^{\vert \Gamma\vert},\quad \ho_\Gamma(\{\gamma_1,...,\gamma_M\})=(\ho_\Gamma(\gamma_1),...,\ho_\Gamma(\gamma_M), e_G,...,e_G)\\
&h_\Gamma:V_\Gamma\rightarrow \{e_G\}
\eqs 
The set of all holonomy maps for the finite graph system is denoted by $\Hom(\PD_\Gamma,G^{\vert \Gamma\vert})$.

The image of a map $\ho_\Gamma$ on each subgraph $\Gp$ of the graph $\Gamma$ is given by
\beqs (\ho_\Gamma(\gamma_1),...,\ho_\Gamma(\gamma_M),e_G,...,e_G)
\eqs is an element of $G^{\vert \Gamma\vert}$. The set of all images of maps on subgraphs of $\Gamma$ is denoted by $\Ab_\Gamma$.
\end{defi}
The idea is now to study two different restrictions of the set $\PD_\Gamma$ of subgraphs. For a short notation of a ''set of  holonomy maps for a certain restricted set of subgraphs of a graph'' in this article the following notions are introduced.
\begin{defi}
If the subset of all disconnected subgraphs of the finite graph system $\PD_\Gamma$ is considered, then the restriction of $\Ab_\Gamma$, which is identified with $G^{\vert \Gamma\vert}$ appropriately, is called the \hypertarget{non-standard identification}{\textbf{non-standard identification of the configuration space}}. If the subset of all natural identified subgraphs of the finite graph system $\PD_\Gamma$ is considered, then the restriction of $\Ab_\Gamma$, which is identified with $G^{\vert \Gamma\vert}$ appropriately, is called the \hypertarget{natural identification}{\textbf{natural identification of the configuration space}}.
\end{defi}

A comment on the non-standard identification of $\Ab_\Gamma$ is the following. If $\Gp:=\{\gamma_1\circ\gamma_2\}$ and $\Gpp:=\{\gamma_2\}$ are two subgraphs of $\Gamma:=\{\gamma_1,\gamma_2,\gamma_3\}$. The graph $\Gp$ is a subgraph of $\Gamma$. Then evaluation of a map $\ho_\Gamma$ on a subgraph $\Gp$ is given by
\beqs \ho_\Gamma(\Gp)=(\ho_\Gamma(\gamma_1\circ\gamma_2),\ho_\Gamma(s(\gamma_2)),\ho_\Gamma(s(\gamma_3)))=(\ho_\Gamma(\gamma_1)\ho_\Gamma(\gamma_2),e_G,e_G)\in G^3
\eqs and the holonomy map of the subgraph $\Gpp$ of $\Gp$ is evaluated by
\beqs \ho_\Gamma(\Gpp)=(\ho_\Gamma(s(\gamma_1)),\ho_\Gamma(s(\gamma_2))\ho_\Gamma(\gamma_2),\ho_\Gamma(s(\gamma_3)))=(\ho_\Gamma(\gamma_2),e_G,e_G)\in G^3
\eqs

\begin{exa}
Recall example \thesubsection.\ref{exa natidentif}.
For example for a subgraph $\Gp:=\{\gamma_1\circ\gamma_2,\gamma_3\circ\gamma_4,...,\gamma_{M-1}\circ\gamma_M\}$, which is naturally identified with a set $\{\gamma_1,...,\gamma_M\}$. Then the holonomy map is evaluated at $\Gp$ such that \[\ho_\Gamma(\Gp)=(\ho_\Gamma(\gamma_1),\ho_\Gamma(\gamma_2),....,\ho_\Gamma(\gamma_M),e_G,...,e_G)\in G^N\] where $N=\vert \Gamma\vert$. For example, let $\Gp:=\{\gamma_1,\gamma_2\}$ such that $\gamma_2=\gpe\circ\gpz$ and which is naturally identified with $\{\gamma_1,\gpe,\gpz\}$. Hence \[\ho_\Gamma(\Gp)=(\ho_\Gamma(\gamma_1),\ho_\Gamma(\gpe),\ho_\Gamma(\gpz),e_G,...,e_G)\in G^N\] is true.

Another example is given by the disconnected graph $\Gp:=\{\gamma_1\circ\gamma_2\circ\gamma_3,\gamma_4\}$, which is a subgraph of $\Gamma:=\{\gamma_1,\gamma_2,\gamma_3,\gamma_4\}$. Then the non-standard identification is given by
\[\ho_\Gamma(\Gp)=(\ho_\Gamma(\gamma_1\circ\gamma_2\circ\gamma_3),\ho_\Gamma(\gamma_4),e_G,e_G)\in G^4\]

If the natural identification is used, then $\ho_\Gamma(\Gp)$ is idenified with 
\[(\ho_\Gamma(\gamma_1),\ho_\Gamma(\gamma_2),\ho_\Gamma(\gamma_3),\ho_\Gamma(\gamma_4))\in G^4\]

Consider the following example. Let $\Gppp:=\{\gamma_1,\alpha,\gamma_2,\gamma_3\}$ be a graph such that 
 \begin{center}
\includegraphics[width=0.2\textwidth]{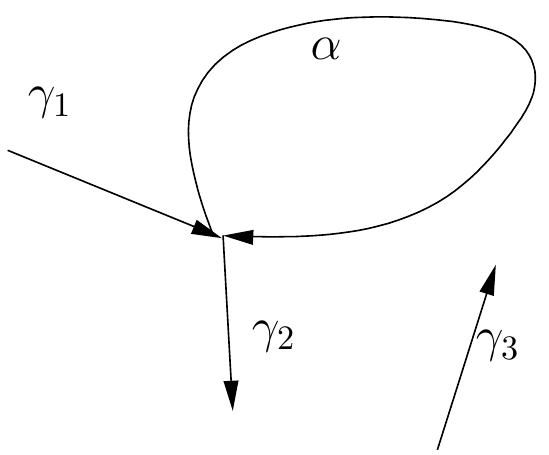}
\end{center}
Then notice the sets $\Gamma_1:=\{\gamma_1\circ\alpha,\gamma_3\}$ and $\Gamma_2:=\{\gamma_1\circ\alpha^{-1},\gamma_3\}$. In the non-standard identification of the configuration space $\Ab_{\Gppp}$ it is true that,
\beqs \ho_{\Gppp}(\Gamma_1)=(\ho_{\Gppp}(\gamma_1\circ\alpha),\ho_{\Gppp}(\gamma_3),e_G,e_G)\in G^4,\\
\ho_{\Gppp}(\Gamma_2)=(\ho_{\Gppp}(\gamma_1\circ\alpha^{-1}),\ho_{\Gppp}(\gamma_3),e_G,e_G)\in G^4
\eqs holds. Whereas in the natural identification of $\Ab_{\Gppp}$
 \beqs \ho_{\Gppp}(\Gamma_1)=(\ho_{\Gppp}(\gamma_1),\ho_{\Gppp}(\alpha),\ho_{\Gppp}(\gamma_3),e_G)\in G^4,\\
\ho_{\Gppp}(\Gamma_2)=(\ho_{\Gppp}(\gamma_1),\ho_{\Gppp}(\alpha^{-1}),\ho_{\Gppp}(\gamma_3),e_G)\in G^4
\eqs yields.
\end{exa}

The equivalence class of similar or equivalent groupoid morphisms defined in definition \ref{def similargroupoidhom} allows to define the following object.
The set of images of all holonomy maps of a finite graph system modulo the similar or equivalent groupoid morphisms equivalence relation is denoted by $\Ab_\Gamma/\bar\SimGroup_\Gamma$. 

\subsection{The group-valued quantum flux operators associated to surfaces and graphs}\label{subsec fluxdef}

Let $G$ be the structure group of a principal fibre bundle $P(\Sigma,G,\pi)$. Then the quantum flux operators, which are associated to a fixed surface $S$, are $G$-valued operators. For the construction of the quantum flux operator $\rho_S(\gamma)$ different maps from a graph $\Gamma$ to a direct product $G\times G$ are considered. This is related to the fact that, one distinguishes between paths that are ingoing and paths that are outgoing with resepect to the surface orientation of $S$. If there are no intersection points of the surface $S$ and the source or target vertex of a path $\gamma_i$ of a graph $\Gamma$, then the map maps the path $\gamma_i$ to zero in both entries. For different surfaces or for a fixed surface different maps refer to different quantum flux operators.  
 
\begin{defi}
Let $\breve S$ be a finite set $\{S_i\}$ of surfaces in $\Sigma$, which is closed under a flip of orientation of the surfaces. Let $\Gamma$ be a graph such that each path in $\Gamma$ satisfies one of the following conditions 
\begin{itemize}
 \item the path intersects each surface in $\breve S$ in the source vertex of the path and there are no other intersection points of the path and any surface contained in $\breve S$,
 \item the path intersects each surface in $\breve S$ in the target vertex of the path and there are no other intersection points of the path and any surface contained in $\breve S$,
 \item the path intersects each surface in $\breve S$ in the source and target vertex of the path and there are no other intersection points of the path and any surface contained in $\breve S$,
 \item the path does not intersect any surface $S$ contained in $\breve S$.
\end{itemize} Finally let $\PD_\Gamma$ denotes the finite graph system associated to $\Gamma$. 

Then define the intersection functions $\iota_L:\breve S\times \Gamma\rightarrow \{\pm 1,0\}$ such that
\beqs \iota_L(S,\gamma):=
\left\{\begin{array}{ll}
1 &\text{ for a path }\gamma\text{ lying above and outgoing w.r.t. }S\\
-1 &\text{ for a path }\gamma\text{ lying below and outgoing w.r.t. }S\\
0 &\text{ the path }\gamma\text{ is not outgoing w.r.t. }S
\end{array}\right.
\eqs
and the intersection functions $\iota_R:\breve S\times \Gamma\rightarrow\{\pm 1,0\}$ such that
\beqs \iota_L(S,\gamma):= \left\{\begin{array}{ll}
-1 &\text{ for a path }\gp\text{ lying above and ingoing w.r.t. }S\\
1 &\text{ for a path }\gp\text{ lying below and ingoing w.r.t. }S\\
0 &\text{ the path }\gp\text{ is not ingoing w.r.t. }S
\end{array}\right.
\eqs whenever $S\in\breve S$ and $\gamma\in\Gamma$.

Define a map $o_L:\breve S\rightarrow G$ such that
\beqs o_L(S)&=o_L(S^ {-1})
\eqs whenever $S\in\breve S$ and $S^ {-1}$ is the surface $S$ with reversed orientation. Denote the set of such maps by $\breve o_L$. Respectively the map $o_R:\breve S\rightarrow G$ such that
\beqs o_R(S)&=o_R(S^ {-1})
\eqs whenever $S\in\breve S$. Denote the set of such maps by $\breve o_R$.
Moreover there is a map $o_L\times o_R:\breve S\rightarrow G\times G$ such that
\beqs (o_L,o_R)(S)&=(o_L,o_R)(S^ {-1})
\eqs whenever $S\in\breve S$. Denote the set of such maps by $\breve o$.

Then define the \textbf{group-valued quantum flux set for paths}
\beqs  \Gop_{\breve S,\Gamma}
:=\bigcup_{o_L\times o_R\in\breve o}\bigcup_{S\in\breve S}\Big\{& (\rho^L,\rho^R)\in\Map(\Gamma,G\times G): 
&(\rho^L, \rho^R)(\gamma):=(o_L(S)^{\iota_L(S,\gamma)},o_R(S)^{\iota_R(S,\gamma)})\Big\}
\eqs
where $\Map(\Gamma,G\times G)$ denotes the set of all maps from the graph $\Gamma$ to the direct product $G\times G$.

Define the \textbf{set of group-valued quantum fluxes for graphs}
\beqs G_{\breve S,\Gamma}:= \bigcup_{o_L\times o_R\in\breve o}\bigcup_{S\in\breve S}\Big\{ \rho_{S,\Gamma}\in\Map(\PD^{\op}_\Gamma,G^{\vert\Gamma\vert}\times G^{\vert\Gamma\vert}):\quad 
&\rho_{S,\Gamma}:=\rho_S\times...\times \rho_S\\&\text{ where }\rho_S(\gamma):=(o_L(S)^{\iota_L(\gamma,S)},o_R(S)^{\iota_R(\gamma,S)}),\\
&\rho_S\in\Gop_{\breve S,\Gamma},S\in\breve S,\gamma\in\Gamma\Big\}\eqs 
\end{defi}
Notice if $H$ is a closed subgroup of $G$, then $H_{\breve S,\Gamma}$ can be defined in analogy to $G_{\breve S,\Gamma}$.
In particular if the group $H$ is replaced by the center $\ZD(G)$ of the group $G$, then the set $\Gop_{\breve S,\Gamma}$ is replaced by $\Zop_{\breve S,\Gamma}$ and $G_{\breve S,\Gamma}$ is changed to $\ZD_{\breve S,\Gamma}$. 

Furthermore observe that, $(\iota_L\times \iota_R)(S^{-1},\gamma)=(-\iota_L\times -\iota_R)(S,\gamma)$ for every $\gamma\in\Gamma$ holds. Remark that, the condition $\rho^L(\gamma)=\rho^R(\gamma^{-1})$ is not required. 

\begin{exa}\label{exa Exa1}
For example the following example can be analysed. Consider a graph $\Gamma$ and two disjoint surface sets $\breve S$ and $\breve T$.
\begin{center}
\includegraphics[width=0.45\textwidth]{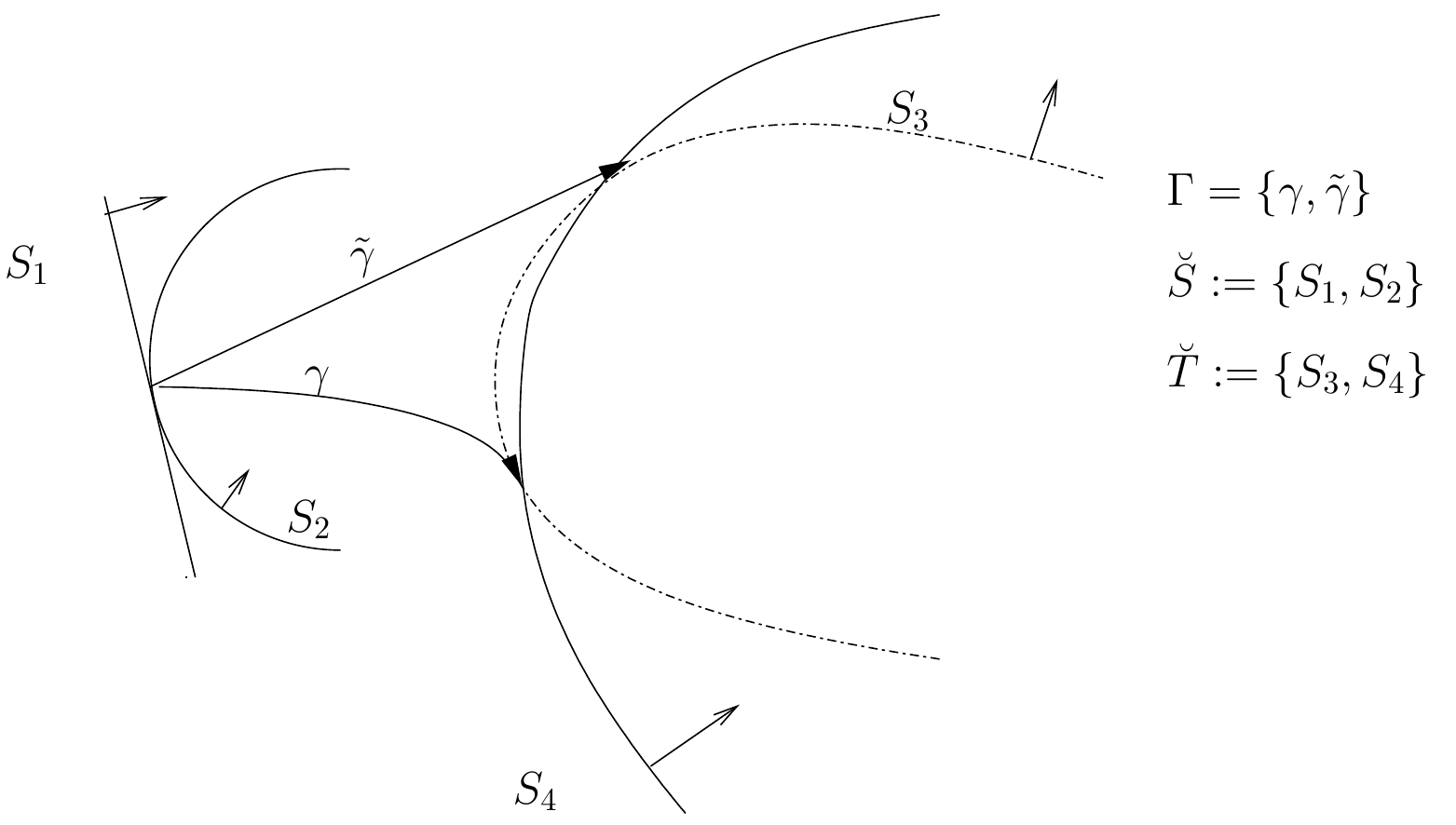}
\end{center}
Then the elements of $\Gop_{\breve S,\Gamma}$ are for example the maps $\rho^L_{i}\times \rho^R_{i}$ for $i=1,2$ such that 
\beqs 
\rho_1(\gamma)&:= (\rho^L_{1}, \rho^R_{1})(\gamma)=(\sigma_L(S_1)^{\iota_L(S_1,\gamma)},\sigma_R(S_1)^{\iota_R(S_1,\gamma)})=(g_{1},0)\\
\rho_1(\tg)&:= (\rho^L_{1}, \rho^R_{1})(\tg)=(\sigma_L(S_1)^{\iota_L(S_1,\tg)},\sigma_R(S_1)^{\iota_R(S_1,\tg)})= (g_1,0)\\
\rho_2(\gamma)&:= (\rho^L_{2}, \rho^R_{2})(\gamma)=(\sigma_L(S_2)^{\iota_L(S_2,\gamma)},\sigma_R(S_2)^{\iota_R(S_2,\gamma)})
=(g_{2},0)\\
\rho_2(\tg)&:= (\rho^L_{2}, \rho^R_{2})(\tg)=(\sigma_L(S_2)^{\iota_L(S_2,\tg)},\sigma_R(S_2)^{\iota_R(S_2,\tg)})
=(g_{2},0)\\
\rho_3(\gamma)&:= (\rho^L_{3}, \rho^R_{3})(\gamma)=(\sigma_L(S_3)^{\iota_L(S_3,\gamma)},\sigma_R(S_3)^{\iota_R(S_3,\gamma)})
=(0,h_{3}^{-1})\\
\rho_3(\tg)&:= (\rho^L_{3}, \rho^R_{3})(\tg)=(\sigma_L(S_3)^{\iota_L(S_3,\tg)},\sigma_R(S_3)^{\iota_R(S_3,\tg)})= (0,h_3^{-1})\\
\rho_4(\gamma)&:= (\rho^L_{4}, \rho^R_{4})(\gamma)=(\sigma_L(S_4)^{\iota_L(S_4,\gamma)},\sigma_R(S_4)^{\iota_R(S_4,\gamma)})
=(0,h_{4})\\
\rho_4(\tg)&:= (\rho^L_{4}, \rho^R_{4})(\tg)=(\sigma_L(S_4)^{\iota_L(S_4,\tg)},\sigma_R(S_4)^{\iota_R(S_4,\tg)})= (0,h_4)
\eqs 

This example shows that, the surfaces $\{S_1,S_2\}$ are similar, whereas the surfaces $\{T_1,T_2\}$ produce different signatures for different paths. Moreover the set of surfaces are chosen such that one component of the direct sum is always zero. 
\end{exa}

For a particular surface set $\breve S$, the following set is defined
\beqs\bigcup_{\sigma_L\times\sigma_R\in\breve\sigma}\bigcup_{S\in\breve S}
\Big\{ (\rho^L,\rho^R)\in\Map(\Gamma,G\times G): \quad(\rho^L, \rho^R)(\gamma):=(\sigma_L(S)^{\iota_L(S,\gamma)},0)\Big\}\eqs can be identified with 
\beqs\bigcup_{\sigma_L\in\breve\sigma_L}\bigcup_{S\in\breve S}\Big\{\rho\in\Map(\Gamma,G): \quad
\rho(\gamma):=\sigma_L(S)^{\iota_L(S,\gamma)}\Big\}
\eqs 
The same is observed for another surface set $\breve T$ and the set $\Gop_{\breve T,\Gamma}$ is identifiable with 
\beqs\bigcup_{\sigma_R\in\breve\sigma_R}\bigcup_{T\in\breve T}
\Big\{\rho\in\Map(\Gamma,G): \quad
\rho(\gamma):=\sigma_R(T)^{\iota_R(T,\gamma)}\Big\}
\eqs

The intersection behavoir of paths and surfaces plays a fundamental role in the definition of the flux operator. There are exceptional configurations of surfaces and paths in a graph. One of them is the following.

\begin{defi}
A surface $S$ has the \textbf{surface intersection property for a graph} $\Gamma$, if the surface intersects each path of $\Gamma$ once in the source or target vertex of the path and there are no other intersection points of $S$ and the path. 
\end{defi}

This is for example the case for the surface $S_1$ or the surface $S_3$, which are presented in example \thesection.\ref{exa Exa1}. Notice that in general, for the surface $S$ there are $N$ intersection points with $N$ paths of the graph. In the example the evaluated map $\rho_1(\gamma)=(g_1,0)=\rho_1(\tg)$ for $\gamma,\tg\in\Gamma$ if the surface $S_1$ is considered.

The property of a path lying above or below is not important for the definition of the surface intersection property for a surface. This indicates that the surface $S_4$ in the example \thesection.\ref{exa Exa1} has the surface intersection property, too.

Let a surface $S$ does not have the surface intersection property for a graph $\Gamma$, which contains only one path $\gamma$. Then for example the path $\gamma$ intersects the surface $S$ in the source and target vertices such that the path lies above the surface $S$. Then the map $\rho^ L\times \rho^ R$ is evaluated for the path $\gamma$ by
\beqs (\rho^ L\times \rho^ R)(\gamma)=(g,h^{-1})
\eqs
Hence simply speaking the surface intersection property reduces the components of the map $\rho^ L\times \rho^ R$, but for different paths to different components.

Now, consider a bunch of sets of surfaces such that for each surface there is only one intersection point.
\begin{defi}\label{def intprop}
A set $\breve S$ of $N$ surfaces has the \textbf{surface intersection property for a graph $\Gamma$} with $N$ independent edges, if it contain only surfaces, for which each path $\gamma_i$ of a graph $\Gamma$ intersects each surface $S_i$ only once in the same source or target vertex of the path $\gamma_i$, there are no other intersection points of each path $\gamma_i$ and each surface in $\breve S$ and there is no other path $\gamma_j$ that intersects the surface $S_i$ for $i\neq j$ where $1 \leq i,j\leq N$.
\end{defi}
Then for example consider the following configuration.

\begin{exa} 
\begin{center}
\includegraphics[width=0.45\textwidth]{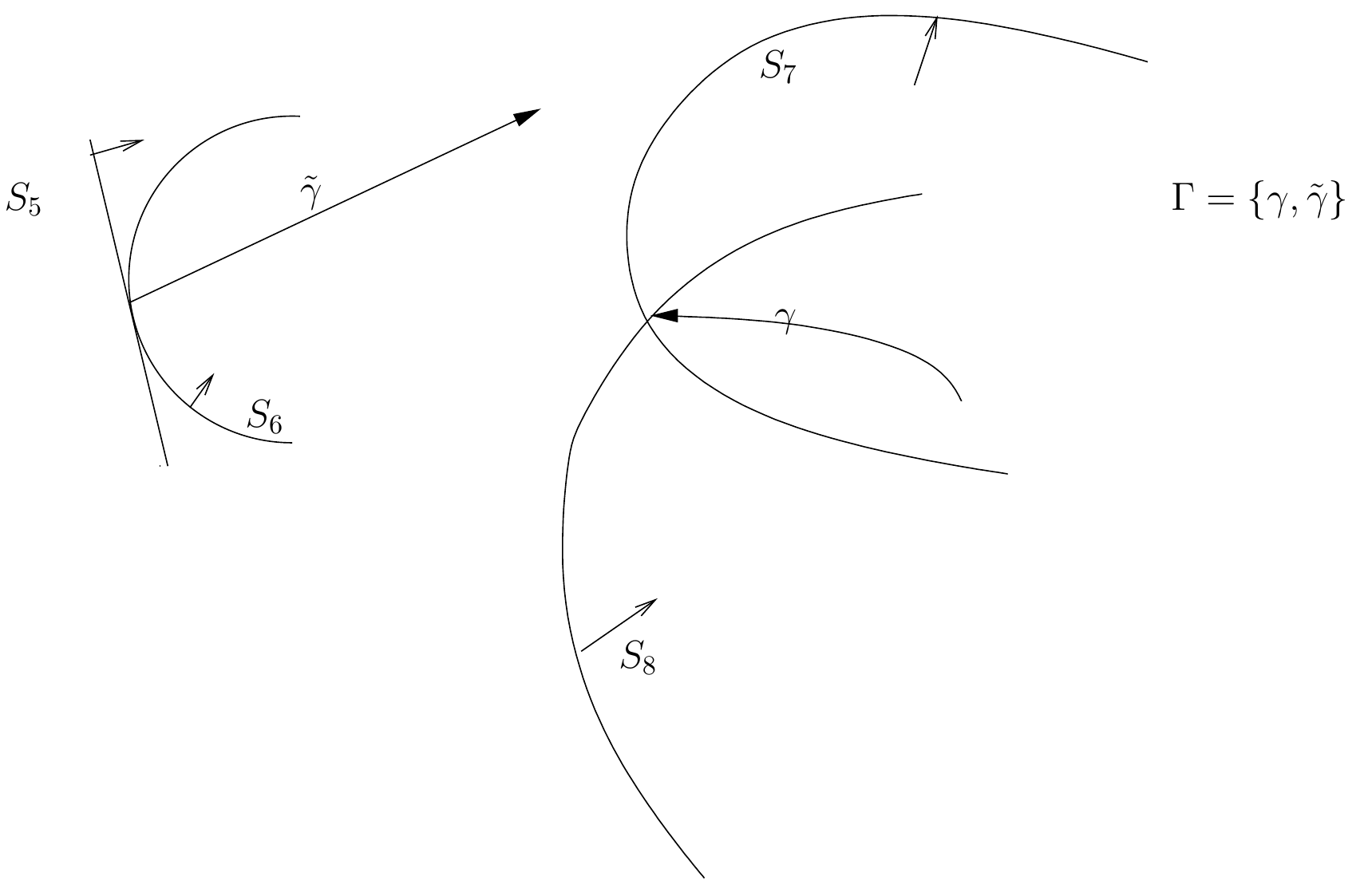}
\end{center} 
The sets $\{S_{6},S_{7}\}$ or $\{S_5,S_{8}\}$ have the surface intersection property for the graph $\Gamma$. 
The images of a map $E$ is
\beqs \rho_5(\tg)=(g_5,0),\quad \rho_{8}(\gamma)=(0,h_{8})
\eqs
\end{exa}
Note that simply speaking the property indicates that each map reduces to a component of $\rho^ L\times \rho^R$.

A set of surfaces that has the surface intersection property for a graph is further specialised by restricting the choice to paths lying ingoing and below with respect to the surface orientations. 
\begin{defi}
A set $\breve S$ of $N$ surfaces has the \hypertarget{simple surface intersection property for a graph}{\textbf{simple surface intersection property for a graph $\Gamma$}} with $N$ independent edges, if it contains only surfaces, for which each path $\gamma_i$ of a graph $\Gamma$ intersects only one surface $S_i$ only once in the target vertex of the path $\gamma_i$, the path $\gamma_i$ lies above and there are no other intersection points of each path $\gamma_i$ and each surface in $\breve S$. 
\end{defi}
\begin{exa}Consider the following example.
\begin{center}
\includegraphics[width=0.45\textwidth]{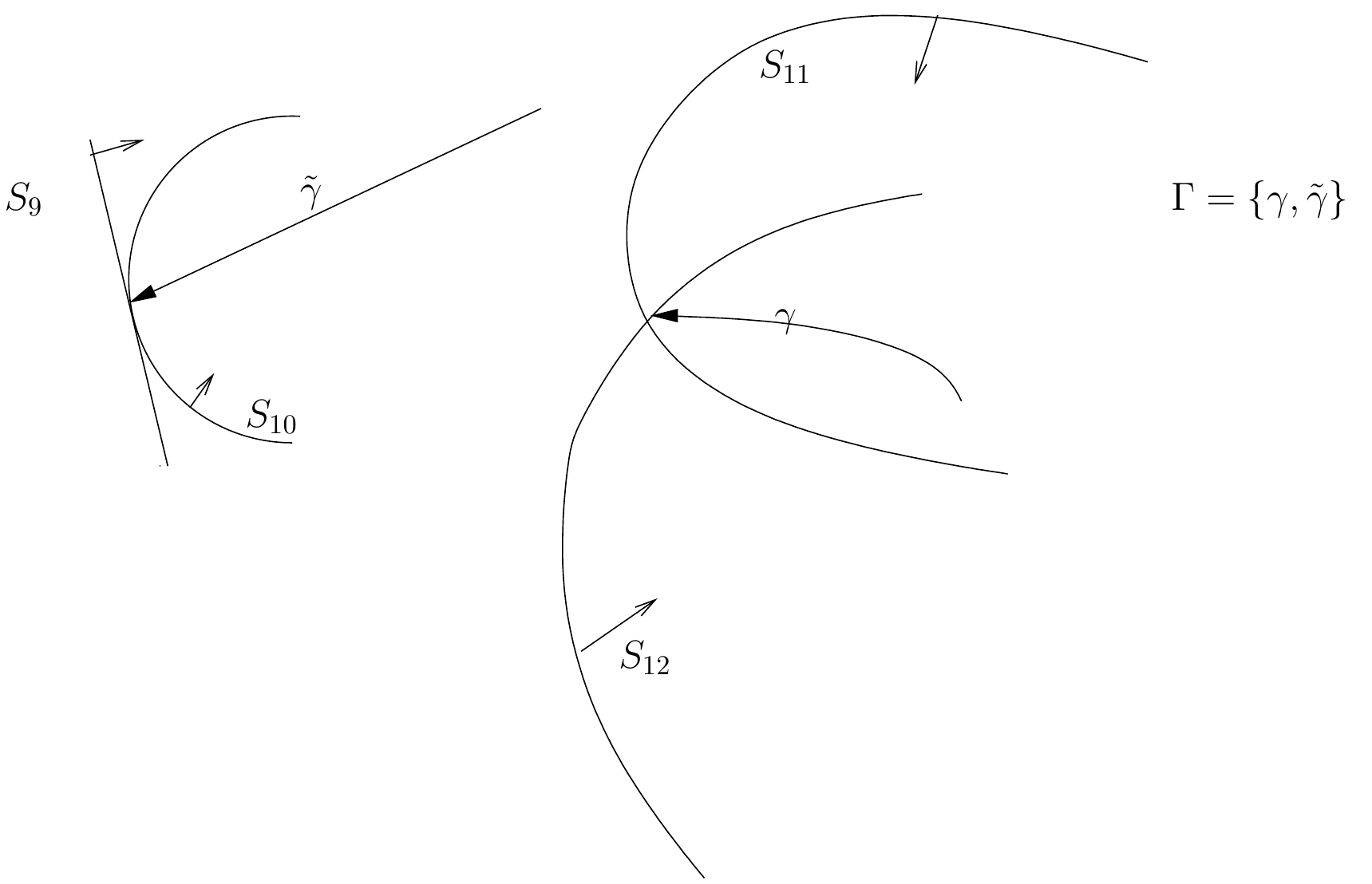}
\end{center} The sets $\{S_{9},S_{11}\}$ or $\{S_{10},S_{12}\}$ have the simple surface intersection property for the graph $\Gamma$.
Calculate
\beqs \rho_{9}(\tg)=(0,h_{9}^{-1}),\quad \rho_{11}(\gamma)=(0,h_{11}^{-1})
\eqs
\end{exa}
In this case the set $\Gop_{\breve S,\Gamma}$ reduces to
\beqs\bigcup_{\sigma_R\in\breve\sigma_R}\bigcup_{S\in\breve S}\Big\{\rho\in\Map(\Gamma,\go): \quad
\rho(\gamma):=\sigma_R(S)^{-1}\text{ for }\gamma\cap S= t(\gamma)\Big\}
\eqs Notice that, the set $\Gamma\cap \breve S=\{t(\gamma_i)\}$ for a surface $S_i\in\breve S$ and $\gamma_i\cap S_j\cap S_i =\{\varnothing\}$ for a path $\gamma_i$ in $\Gamma$ and $i\neq j$.

On the other hand, there exists a set of surfaces such that each path of a graph intersects all surfaces of the set in the same vertex. This contradicts the assumption that each path of a graph intersects only one surface once. 
\begin{defi}Let $\Gamma$ be a graph that contains no loops.

A set $\breve S$ of surfaces has the \hypertarget{same intersection property}{\textbf{same surface intersection property for a graph}} $\Gamma$ iff each path $\gamma_i$ in $\Gamma$ intersects with all surfaces of $\breve S$ in the same source vertex $v_i\in V_\Gamma$ ($i=1,..,N$), all paths are outgoing and lie below each surface $S\in\breve S$ and there are no other intersection points of each path $\gamma_i$ and each surface in $\breve S$. 

A surface set $\breve S$ has the \textbf{same right surface intersection property for a graph} $\Gamma$ iff each path $\gamma_i$ in $\Gamma$ intersects with all surfaces of $\breve S$ in the same target vertex $v_i\in V_\Gamma$ ($i=1,..,N$), all paths are ingoing and lie above each surface $S\in\breve S$ and there are no other intersection points of each path $\gamma_i$ and each surface in $\breve S$. 
\end{defi}

Recall the example \thesection.\ref{exa Exa1}. Then the set $\{S_{1},S_{2}\}$ has the same surface intersection property for the graph $\Gamma$.

Then the set $\Gop_{\breve S,\Gamma}$ reduces to
\beqs\bigcup_{\sigma_L\in\breve\sigma_L}\bigcup_{S\in\breve S}\Big\{\rho\in\Map(\Gamma,\go): \quad
\rho(\gamma):= \sigma_L(S)^{-1}\text{ for }\gamma\cap S= s(\gamma)\Big\}
\eqs Notice that, $\gamma\cap S_1\cap ...\cap S_N=s(\gamma)$ for a path $\gamma$ in $\Gamma$ whereas $\Gamma\cap\breve S=\{s(\gamma_i)\}_{1\leq i\leq N}$. Clearly $\Gamma\cap S_i=s(\gamma_i)$ for a surface $S_i$ in $\breve S$ holds. 
Simply speaking the physical intution behind that is given by fluxes associated to different surfaces that should act on the same path.
 
A very special configuration is the following.
\begin{defi}
A set $\breve S$ of surfaces has the \textbf{same surface intersection property for a graph $\Gamma$ containing only loops} iff each loop $\gamma_i$ in $\Gamma$ intersects with all surfaces of $\breve S$ in the same vertices $s(\gamma_i)=t(\gamma_i)$ in $V_\Gamma$ ($i=1,..,N$), all loops lie below each surface $S\in\breve S$ and there are no other intersection points of each loop in $\Gamma$ and each surface in $\breve S$. 
\end{defi}

Notice that, both properties can be restated for other surface and path configurations. Hence a surface set have the simple or same surface intersection property for paths that are outgoing and lie above (or ingoing and below, or outgoing and below). The important fact is related to the question if the intersection vertices are the same for all surfaces or not.

Finally for the definition of the quantum flux operators notice the following objects.
\begin{defi}
 The set of all images of maps in $\Gop_{\breve S,\Gamma}$ for a fixed surface set $\breve S$ and a fixed path $\gamma$ in $\Gamma$ is denoted by $\bar\Gop_{\breve S,\gamma}$. 

The set of all finite products of images of maps in $G_{\breve S,\Gamma}$ for a fixed surface set $\breve S$ and a fixed graph $\Gamma$ is denoted by $\bar G_{\breve S,\Gamma}$. 
\end{defi}
The product $\cdot$ on $\bar G_{\breve S,\Gamma}$ is given by
\beqs \rho_{S_1,\Gamma}(\Gamma)\cdot \rho_{S_2,\Gamma}(\Gamma)&=(\rho_{S_1}(\gamma_1)\cdot \rho_{S_2}(\gamma_1),...,\rho_{S_1}(\gamma_N)\cdot \rho_{S_2}(\gamma_N))\\
&=(o_L(S_1)^{-1}o_L(S_2)^{-1},...,o_L(S_2)^{-1}o_L(S_1)^{-1})\\
&=((o_L(S_2)o_L(S_1))^{-1},...,(o_L(S_2)o_L(S_1))^{-1})\\
&=\rho_{S_3,\Gamma}(\Gamma)
\eqs

\begin{defi}
Let $S$ be a surface and $\Gamma$ be a graph such that the only intersections of the graph and the surface in $S$ are contained in the vertex set $V_\Gamma$. Moreover let $\fPG$ be a finite path groupoid associated to $\Gamma$. 

Then define the set for a fixed surface $S$ by
\beqs &\Map_S(\PD_\Gamma\Sigma,G\times G)\\&
:= \bigcup_{o_L\times o_R\in\breve o}\bigcup_{S\in\breve S}\Big\{& (\rho^L,\rho^R)\in\Map(\PD_\Gamma\Sigma,G\times G): 
&(\rho^L, \rho^R)(\gamma):=(o_L(S)^{\iota_L(S,\gamma)},o_R(S)^{\iota_R(S,\gamma)})\Big\}
\eqs
\end{defi}

Then the quantum flux operators are elements of the following group.
\begin{prop}Let $\breve S$ a set of surfaces and $\Gamma$ be a fixed graph, which contains no loops, such that the set $\breve S$ has the same surface intersection property for the graph $\Gamma$. 

The set $\bar\Gop_{\breve S,\gamma}$ has the structure of a group.

The group $\bar\Gop_{\breve S,\gamma}$ is called the \textbf{flux group associated a path and a finite set of surfaces}.
\end{prop}
\begin{proofs}This follows easily from the observation that in this case 
$\Gop_{\breve S,\gamma}$ reduces to
\beqs\bigcup_{o_L\in\breve o_L}\bigcup_{S\in\breve S}\Big\{& \rho^L\in\Map(\Gamma,G): 
&\rho^L(\gamma):=o_L(S)^{-1}\text{ for }\gamma\cap S=s(\gamma)\Big\}
\eqs

There always exists a map $\rho^L_{S,3}\in\Gop_{\breve S,\gamma}$ such that the following equation defines a  multiplication operation 
\beqs \rho^L_{S,1}(\gamma)\cdot\rho^L_{S,2}(\gamma)=g_{1}g_{2}:=\rho^L_{S,3}(\gamma)\in\bar\Gop_{\breve S,\gamma}\eqs with inverse $(\rho^L_S(\gamma))^{-1}$ such that
\beqs \rho^L_S(\gamma)\cdot (\rho^L_S(\gamma))^{-1}=(\rho^L_S(\gamma))^{-1}\cdot \rho^L_S(\gamma)=e_G\quad\forall\gamma\in\Gamma\eqs
\end{proofs}

Notice that for a loop $\alpha$ an element $\rho_S(\alpha)\in\bar\Gop_{\breve S,\gamma}$ is defined by \beqs\rho_S(\alpha):=(\rho_S^L\times\rho_S^R)(\alpha)= (g, h)\in G^2 
\eqs In the case of a path $\gp$ that intersects a surface $S$ in the source and target vertex there is also an element $\rho_S(\gp)\in\bar\Gop_{\breve S,\gamma}$ defined by \beqs\rho_S(\gp):=(\rho_S^L\times\rho_S^R)(\gp)= (g, h)\in G^2 
\eqs

\begin{prop}
Let $\breve S$ be a set of surfaces and $\Gamma$ be a fixed graph, which contains no loops, such that the set $\breve S$ has the same surface intersection property for the graph $\Gamma$. Let $\PD^{\op}_\Gamma$ be a finite orientation preserved graph system such that the set $\breve S$.

The set $\bar G_{\breve S,\Gamma}$ has the structure of a group.

The  set $\bar G_{\breve S,\Gamma}$ is called the \textbf{flux group associated a graph and a finite set of surfaces}.
\end{prop}
\begin{proofs}
This follows from the observation that the set $G_{\breve S,\Gamma}$ is identified with
\beqs  \bigcup_{\sigma_L\in\breve\sigma_L}\bigcup_{S\in\breve S}\Big\{ \rho_{S,\Gamma}\in\Map(\PD^{\op}_\Gamma,G^{\vert E_\Gamma\vert}):\quad 
&\rho_{S,\Gamma}:=\rho_S\times...\times \rho_S\\&\text{ where }\rho_S(\gamma):=o_L(S)^{-1},
\rho_S\in \Gop_{\breve S,\Gamma},S\in\breve S,\gamma\in\Gamma\Big\}\eqs 

Let $\breve S$ be a surface set having the same intersection property for a fixed graph $\Gamma:=\{\gamma_1,...,\gamma_N\}$. Then for two surfaces $S_1,S_2$ contained in $\breve S$ define
\beqs \rho_{S_1,\Gamma}(\Gamma)\cdot \rho_{S_2,\Gamma}(\Gamma)
&=(\rho_{S_1}(\gamma_1)\cdot \rho_{S_2}(\gamma_1),...,\rho_{S_1}(\gamma_N)\cdot \rho_{S_2}(\gamma_N))\\
&= (g_{S_1},...,g_{S_1})\cdot (g_{S_2},..., g_{S_2}) =( g_{S_1}g_{S_2},..., g_{S_1}g_{S_2})
\eqs where $\Gamma=\{\gamma_1,...,\gamma_N\}$.
Note that, since the maps $o_L$ are arbitrary maps from $\breve S$ to $G$, it is assumed that the maps satisfy $o_L(S_i):=g^{-1}_{S_i}\in G$ for $i=1,2$. 
Clearly this is related to in this particular case of the graph $\Gamma$ and can be generalised. 

The inverse operation is given by
\beqs (\rho_{S,\Gamma}(\Gamma))^{-1}=((\rho_S(\gamma_1))^{-1},...,(\rho_S(\gamma_N))^{-1})
\eqs
where $N=\vert \Gamma\vert$ and $\rho_S\in\Gop_{\breve S,\gamma}$ for $S\in\breve S$. Since it is true that
\beqs \rho_{S,\Gamma}(\Gamma)\cdot \rho_{S,\Gamma}(\Gamma)^{-1}&= (g_{S},...,g_{S})\cdot (g_{S}^{-1},..., g_{S}^{-1})\\
&=(\rho_{S}(\gamma_1)\cdot \rho_{S}(\gamma_1)^{-1},...,\rho_{S}(\gamma_N)\cdot \rho_{S}(\gamma_N)^{-1})\\
& =( g_{S}g_{S}^{-1},..., g_{S}g_{S}^{-1})=(e_G,...,e_G)
\eqs yields.
\end{proofs}
Notice that, it is not defined that
\beqs &\rho_{S_1,\Gamma}(\Gamma)\bullet_R \rho_{S_2,\Gamma}(\Gamma)\\&=(o_L(S_2)^{-1}o_L(S_1)^{-1},...,o_L(S_2)^{-1}o_L(S_1)^{-1})
=((o_L(S_1)o_L(S_2))^{-1},...,(o_L(S_1)o_L(S_2))^{-1})\\
&=\rho_{S_3,\Gamma}(\Gamma)
\eqs is true.
Moreover observe that, if all subgraphs of a finite orientation preserved graph system are \hyperlink{natural identification}{naturally identified}, then $\bar G_{\breve S,\Gp\leq\Gamma}$ is a subgroup of $\bar G_{\breve S,\Gamma}$ for all subgraphs $\Gp$ in $\PD_\Gamma$. If $G$ is assumed to be a compact Lie group, then the flux group $\bar G_{\breve S,\Gamma}$ is called the Lie flux group.

There is another group, if another surface set is considered.
\begin{prop}Let $\breve T$ be a set of surfaces and $\Gamma$ be a fixed graph such that the set $\breve T$ has the simple surface intersection property for the graph $\Gamma$. Let $\PD^{\op}_\Gamma$ be a finite orientation preserved graph system.

The set $\bar G_{\breve T,\Gamma}$ has the structure of a group.
\end{prop}
The same arguments using the identification of $\bar G_{\breve T,\Gamma}$ with
\beqs  \bigcup_{\sigma_R\in\breve\sigma_R}\Big\{ \rho_{T,\Gamma}\in\Map(\PD^{\op}_\Gamma,G^{\vert E_\Gamma\vert}):\quad 
&\rho_{\breve T,\Gamma}:=\rho_{T_1}\times...\times \rho_{T_N}\\&\text{ where }\rho_{T_i}(\gamma):=o_R(T_i)^{-1},
\rho_{T_i}\in \Gop_{\breve T,\Gamma},T_i\in\breve T,\gamma\in\Gamma\Big\}\eqs 
which is given by
\beqs \rho_{T_1,\Gamma}(\Gamma)\cdot ...\cdot \rho_{T_N,\Gamma}(\Gamma)
&=(\rho_{T_1}(\gamma_1)e_G, e_G,...,e_G)\cdot 
(e_G, \rho_{T_2}(\gamma_2) e_G, e_G,...,e_G)\cdot 
 ...\cdot (e_G,...,e_G,\rho_{T_N}(\gamma_N)e_G)\\
&=(\rho^1_{T_1}(\gamma_1),...,\rho^1_{T_N}(\gamma_N))= (g_{1},...,g_{N})\in G^N\\
&=:\rho_{\breve T,\Gamma}(\Gamma)
\eqs
Then the multiplication operation is presented by
\beqs \rho^1_{\breve T,\Gamma}(\Gamma)\cdot \rho^2_{\breve T,\Gamma}(\Gamma)
&=(\rho^1_{T_1}(\gamma_1)\cdot \rho^2_{T_1}(\gamma_1),...,\rho^1_{T_N}(\gamma_N)\cdot \rho^2_{T_N}(\gamma_N))\\
&= (g_{1,1},...,g_{1,N})\cdot (g_{2,1},..., g_{2,N}) =( g_{1,1}g_{2,1},..., g_{1,N}g_{2,N})\in G^N
\eqs where $\Gamma=\{\gamma_1,...,\gamma_N\}$.

It is also possible that, the fluxes are located only in a vertex and do not depend on ingoing or outgoing, above or below orientation properties.
\begin{defi}\label{def Gloc}Let $\PD_\Gamma$ be a finite graph groupoid associated to a graph $\Gamma$ and let $N$ be the number of edges of the graph $\Gamma$. 

Define the set of maps 
\beqs G^{\loc}_{\Gamma}:=
\Big\{\textbf{g}_\Gamma\in\Map(\PD_\Gamma,G^{\vert \Gamma\vert}):& \textbf{g}_\Gamma:=g^1_\Gamma\circ s\times ...\times g^N_\Gamma\circ s\\  
&g^i_\Gamma\in\Map(\Gamma,G)\Big\}
\eqs 

Then $\bar G^{\loc}_{\Gamma}$ is the set of all images of maps in $G^{\loc}_{\Gamma}$ for all graphs in $\PD_\Gamma$ and $\bar G^{\loc}_{\Gamma}$ is called the \textbf{local flux group associated a finite graph system}. 
\end{defi}

\section{The flux and flux transformation group, n.c. and heat-kernel holonomy $C^*$-algebra}\label{subsec fluxgroupalg}

In this section new algebras are constructed from either the quantum configuration or the quantum momentum variables of LQG. In the following the focus is based on the quantum momentum variables, which are given by the group-valued quantum flux operators associated to surfaces and paths.

\subsection{The flux group $C^*$-algebra associated to graphs and a surface set}

Let $\CD(G)$ be the convolution $^*$-algebra of continuous functions $C_c(G)$ on a locally compact unimodular group $G$ equipped with the convolution product, an inversion and supremum norm.

Recall that a surface $S$ has the same surface intersection property for a graph $\Gamma$, if each path of $\Gamma$ intersect the surface $S$ exactly once in a source vertex of the path and the path is outgoing and lies below.

\begin{cor}\label{cor convsame}Let $S$ be a surface with same surface intersection property for a finite graph system associated to a graph $\Gamma$. Let $G$ be a unimodular locally compact group and let $\bar G_{S,\Gamma}$ be the flux group.

Then the \textbf{convolution flux $^*$-algebra $\CD(\bar G_{S,\Gamma})$ associated to a surface} and a graph $\Gamma$ is defined by the following product
\beqs 
&(f_1 \ast f_2)(\rho_{S,\Gamma}(\Gamma))=(f_1 \ast f_2)(\rho_{S}(\gamma_1),...,\rho_{S}(\gamma_N))\\
&=\int_{G_{S,\Gamma}} f_1(\rho_{S}(\gamma_1)\hat\rho_{S}(\gamma_1)^{-1},...,\rho_{S}(\gamma_N)\hat\rho_{S}(\gamma_N)^{-1})f_2(\hat\rho_{S}(\gamma_1),...,\hat\rho_{S}(\gamma_N))\\&\qquad\qquad\dif\mu_{S,\Gamma}(\hat\rho_{S}(\gamma_1),...,\hat\rho_{S}(\gamma_N))\\
&=\int_{G_{S,\Gamma}} f_1(\rho_{S,\Gamma}(\Gamma)\hat\rho_{S,\Gamma}(\Gamma)^{-1})f_2(\hat\rho_{S,\Gamma}(\Gamma))\dif\mu_{S,\Gamma}(\hat\rho_{S}(\Gamma))
\eqs for $\Gamma=\{\gamma_1,...,\gamma_N\}$ and  where $\rho_{S,\Gamma},\hat\rho_{S,\Gamma}\in G_{S,\Gamma}$ and $\rho_{S},\hat\rho_{S}\in \Gop_{S,\Gamma}$
which reduces to 
\beqs (f_1 \ast f_2)(\rho_{S,\Gamma}(\Gp))=(f_1 \ast f_2)(\rho_{S}(\gamma_i))=\int_{G_{S,\Gamma}} f_1(\rho_{S}(\gamma_i)\hat\rho_{S}(\gamma_i)^{-1})f_2(\hat\rho_{S}(\gamma_i))\dif\mu_{S,\Gamma}(\hat\rho_{S}(\gamma_i))
\eqs for any $i=1,...,N$ and $\Gp=\{\gamma_i\}\in\PD_\Gamma$, the involution 
\beqs
f_\Gamma(\rho_{S,\Gamma}(\Gp))^*:= \overline{ f_\Gamma(\rho_{S,\Gamma}(\Gp)^{-1})}
\eqs for any $i=1,...,N$
and equipped with the supremum norm.
\end{cor}

Set $\Gp:=\{\gamma_i\}$. Remark that, if all paths $\gamma_i$ are ingoing and above, then the product reads
\beq 
&(f_1 \ast f_2)(\rho_{S,\Gamma}(\Gp))=\int_{G} f_1(g_S\hat g_S^{-1})f_2(\hat g_S)\dif\mu(\hat g_S)
\eq otherwise
\beq 
&(f_1 \ast f_2)(\rho_{S,\Gamma}(\Gp))=\int_{G} f_1(g_S^{-1}\hat g_S)f_2(\hat g_S)\dif\mu(\hat g_S)
\eq This implies that, only for one surface the structure is identified with $\CD(G)$. The convolution algebra $\CD(\bar G_{\breve S,\Gamma})$ is defined similarly to the one defined in corollary \ref{cor convsame} for a surface set $\breve S$ with same surface intersection property for a finite graph system associated to a graph $\Gamma$.

Recall that, a set $\breve S$ of $N$ surfaces has the \hyperlink{simple surface intersection property for a graph}{simple surface intersection property for a graph $\Gamma$} with $N$ independent edges, if it contain only surfaces, for which each path $\gamma_i$ of a graph $\Gamma$ intersects only one surface $S_i$ only once in the target vertex of the path $\gamma_i$, the path $\gamma_i$ lies above and there are no other intersection points of each path $\gamma_i$ and each surface in $\breve S$. Then the convolution algebra can be defined as follows.

\begin{cor}
Let $\breve S:=\{S_i\}_{1\leq i\leq N}$ be a set of surfaces with \hypertarget{simple surface intersection property for a graph system}{simple surface intersection property for a finite graph system} associated to a graph $\Gamma$.

Then the \textbf{convolution flux $^*$-algebra $\CD(\bar G_{\breve S,\Gamma})$ associated to a surface set} and a graph $\Gamma$ is defined by the following product
\beqs 
&(f_1 \ast f_2)(\rho_{S_1}(\gamma_1),...,\rho_{S_N}(\gamma_N))\\
&=\int_{G} f_1(\rho_{S_1}(\gamma_1)\hat\rho_{S_1}(\gamma_1)^{-1},...,\rho_{S_N}(\gamma_N)\hat\rho_{S_N}(\gamma_N)^{-1})f_2(\hat\rho_{S_1}(\gamma_1),...,\hat\rho_{S_N}(\gamma_N))\dif\mu(\hat\rho_{\breve S,\Gamma}(\Gamma))
\eqs where $\rho_{\breve S,\Gamma}\in G_{\breve S,\Gamma}$, $\rho_{S_i}\in \Gop_{\breve S,\Gamma}$ for $i=1,...,N$, $\rho_{\breve S,\Gamma}(\Gamma):=(\rho_{S_1}(\gamma_1),...,\rho_{S_N}(\gamma_N))$, the involution is defined by
\beqs
f_\Gamma(\rho_{S_1}(\gamma_1),...,\rho_{S_N}(\gamma_N))^*:= \overline{ f_\Gamma(\rho_{S_1}(\gamma_1)^{-1},...,\rho_{S_N}(\gamma_N)^{-1})}
\eqs 
and equipped with the supremum norm.
\end{cor}

Clearly $\bar G_{\breve S,\Gamma}$ is identified with $G^N$ for $N$ being the number of independent paths in $\Gamma$ such that each of the path $\gamma_i$ intersects a surface $S_i$. 

The convolution algebra $\CD(\bar G_{\breve S,\Gamma})$ is also studied for other situations as far as the surface set $\breve S$ has one of the surface intersection properties, which have been given in section \ref{subsec fluxdef}.

The dual space $C_0(\bar G_{\breve S,\Gamma})^*$ is identified by the Riesz-Markov theorem with the Banach space of bounded complex Baire measures on $\bar G_{\breve S,\Gamma}$. Moreover, each Baire measure has a unique extension to a regular Borel measure on $\bar G_{\breve S,\Gamma}$. The Banach space of all regular Borel measures is denoted by $\textbf{M}(\bar G_{\breve S,\Gamma})$. There is a convolution multiplication
\beq &\int_{\bar G_{\breve S,\Gamma}}f(\rho_{\breve S,\Gamma}(\Gp))\dif (\mu\ast\nu)(\rho_{\breve S,\Gamma}(\Gp))\\
&=\int_{\bar G_{\breve S,\Gamma}}\int_{\bar G_{\breve S,\Gamma}}f(\rho_{\breve S,\Gamma}(\Gp)\hat\rho_{\breve S,\Gamma}(\Gp))\dif\mu(\rho_{\breve S,\Gamma}(\Gp))\dif\nu(\hat\rho_{\breve S,\Gamma}(\Gp))
\eq where $\rho_{\breve S,\Gamma}\in G_{\breve S,\Gamma}$, $\rho_{S_i}\in \Gop_{\breve S,\Gamma}$ for $i=1,...,N$, $\rho_{\breve S,\Gamma}(\Gp):=(\rho_{S_1}(\gamma_1),...,\rho_{S_N}(\gamma_M))$, $\Gp:=\{\gamma_1,...,\gamma_M\}$, $\mu,\nu\in\textbf{M}(\bar G_{\breve S,\Gamma})$ and $f\in C_0(\bar G_{\breve S,\Gamma})$ and an inversion
\beq \int_{\bar G_{\breve S,\Gamma}}f(\rho_{\breve S,\Gamma}(\Gp))\dif\mu^*(\rho_{\breve S,\Gamma}(\Gp))
=\overline{\int_{\bar G_{\breve S,\Gamma}}\bar f(\rho_{\breve S,\Gamma}(\Gp)^{-1})\dif\mu(\rho_{\breve S,\Gamma}(\Gp))}
\eq which transfers $\textbf{M}(\bar G_{\breve S,\Gamma})$ to a Banach $^*$-algebra. Then restrict $\textbf{M}(\bar G_{\breve S,\Gamma})$ to the norm closed subspace consisting of measures absolutely continuous w.r.t. the Haar measure $\mu_{\breve S,\Gamma}$, which is identified with $L^1(\bar G_{\breve S,\Gamma})$ by $\dif\mu(\rho_{\breve S,\Gamma}(\Gp))=f_\Gamma(\rho_{\breve S,\Gamma}(\Gp))\dif\mu_{\breve S,\Gamma}(\rho_{\breve S,\Gamma}(\Gp))$ for $f_\Gamma\in L^1(\bar G_{\breve S,\Gamma})$.

\begin{cor}Let $\breve S:=\{S_i\}_{1\leq i\leq N}$ be a set of surfaces with same surface intersection property for a finite graph system associated to a graph $\Gamma$.

The Banach $^*$-algebra $L^1(\bar G_{\breve S,\Gamma},\mu_{\breve S,\Gamma})$ is the continuous extension of $\CD (\bar G_{\breve S,\Gamma})$ in the $L^1$-norm. 
\end{cor}

There is a non-degenerate $^*$-representation $\pi_0$ of $L^1(\bar G_{\breve S,\Gamma},\mu_{\breve S,\Gamma})$ on the Hilbert space $\HS_\Gamma=L^2(\bar G_{\breve S,\Gamma},\mu_{\breve S,\Gamma})$, which is of the form
\beq \pi_0(f_\Gamma):=\int_{\bar G_{\breve S,\Gamma}} f_\Gamma(\rho_{S_1}(\gamma_1),...,\rho_{S_N}(\gamma_N))\dif\mu_{\breve S,\Gamma}(\rho_{S_1}(\gamma_1),...,\rho_{S_N}(\gamma_N))\eq for $f_\Gamma\in L^1(\bar G_{\breve S,\Gamma},\mu_{\breve S,\Gamma})$ (defined in the sense of a Bochner integral).

Notice that, the Banach $^*$-algebra $L^1(\bar G_{\breve S,\Gamma},\mu_{\breve S,\Gamma})$ has an approximate unit. Then for a $^*$-representation $\pi_0$ of $L^1(\bar G_{\breve S,\Gamma},\mu_{\breve S,\Gamma})$ on $\HS_\Gamma$ exists a GNS-triple $(\HS_\Gamma,\pi_0,\Omega_0)$ and an associated state $\omega_0$ \cite[section 8.6]{Schmuedgen90}. Furthermore, there is a left regular unitary representation $U_{\overleftarrow{L}}^N$ of $\bar G_{\breve S,\Gamma}$ on $\HS_\Gamma$ associated to an action $\alpha_{\overleftarrow{L}}^N$ \cite[Lemma 3.4 and Lemma 3.16]{Kaminski1} or \cite[Lemma 6.1.4 and Lemma 6.1.16]{KaminskiPHD}.
Then observe that, for $f_\Gamma\in L^1(\bar G_{\breve S,\Gamma},\mu_{\breve S,\Gamma})$ and $\rho_{\breve S,\Gamma}^N,\hat\rho_{\breve S,\Gamma}^N\in\bar G_{\breve S,\Gamma}$ the unitary $U_{\overleftarrow{L}}^N$ satisfies
\beq
&U_{\overleftarrow{L}}^N(\hat\rho_{\breve S,\Gamma}^N)\pi_0(f_\Gamma)\Omega_0\\
&= \int_{\bar G_{\breve S,\Gamma}}d\mu_{\breve S,\Gamma}(\rho_{\breve S,\Gamma}^N(\Gamma))U_{\overleftarrow{L}}^N(\hat\rho_{S,\Gamma}^N)f_\Gamma(\rho_{\breve S,\Gamma}^N)\Omega_0\\
&=\int_{\bar G_{\breve S,\Gamma}}d\mu_{\breve S,\Gamma}(\rho_{S_1}(\gamma_1),...,\rho_{S_N}(\gamma_N))f_\Gamma(\hat\rho_{S_1}(\gamma_1)\rho_{S_1}(\gamma_1),...,\hat\rho_{S_N}(\gamma_N)\rho_{S_N}(\gamma_N))\Omega_0\\ 
&= \pi_0(\alpha_{\overleftarrow{L}}^N(\rho_{\breve S,\Gamma}^N)f_\Gamma)\Omega_0\eq 
whenever $\rho_{\breve S,\Gamma}^N(\Gamma):=(\rho_{S_1}(\gamma_1),...,\rho_{S_N}(\gamma_N))$ and $\Omega_0$ is the cyclic vector. This implies 
\beq\omega_0(\alpha_{\overleftarrow{L}}^N(\hat\rho_{\breve S,\Gamma}^N)f_\Gamma)=\omega_0(f_\Gamma)
\eq and hence that the state $\omega_0$ on the Banach $^*$-algebra $L^1(\bar G_{\breve S,\Gamma},\mu_{\breve S,\Gamma})$ associated to the representation $\pi_0$ is $\bar G_{\breve S,\Gamma}$-invariant.

The same is true if all paths in $\Gamma$ intersect in vertices of the set $V_\Gamma$ with a surface $S$ such that all paths are outgoing and lie below the surface $S$ and the unitaries $U_{\overleftarrow{L}}^1(\hat\rho_{S,\Gamma}^1)$ are analysed. Clearly this can be also studied for other situations presented in \cite[Section 3.1]{Kaminski1} or \cite[Section 6.1]{KaminskiPHD}.

Notice that, the Banach $^*$-algebra $L^1(\bar G^{\disc}_{\breve S,\Gamma})$ is generated by all Dirac point measures $\{\delta(\rho_{S,\Gamma}(\Gp)):\rho_{S,\Gamma}(\Gp)\in \bar G_{\breve S,\Gamma}\}$ such that
\beqs &\delta(\rho_{S,\Gamma}(\Gp))\ast\delta(\hat\rho_{S,\Gamma}(\Gp))=\delta(\rho_{S,\Gamma}(\Gp)\hat\rho_{S,\Gamma}(\Gp))\\
&\delta^*(\rho_{S,\Gamma}(\Gp))=\delta(\rho_{S,\Gamma}(\Gp)^{-1})
\eqs
Moreover, recognize that,
\beqs &(\delta(\rho_{S,\Gamma}(\Gp))\ast f_\Gamma)(\hat\rho_{S,\Gamma})=f_\Gamma(\rho_{S,\Gamma}(\Gp)^{-1}\rho_{S,\Gamma}(\Gp))\\
&(f_\Gamma\ast \delta(\rho_{S,\Gamma}(\Gp)))(\hat\rho_{S,\Gamma})=f_\Gamma(\rho_{S,\Gamma}(\Gp)\rho_{S,\Gamma}^{-1}(\Gp))
\eqs yields for all $f_\Gamma\in L^1(\bar G_{\breve S,\Gamma},\mu_{\breve S,\Gamma})$ and $\rho_{S,\Gamma}\in \bar G_{\breve S,\Gamma}$.

Observe that, for $A=\sum_{i=1}^n a_i\delta(\rho_{S_i,\Gamma}(\Gp)))\in L^1(\bar G^{\disc}_{\breve S,\Gamma})$ and $\breve S:=\{S_i\}_{1\leq i\leq N}$, 
there is a state $\hat\omega_0$ on $L^1(\bar G^{\disc}_{\breve S,\Gamma})$ such that
\beq \hat\omega_0(A^*A)&=\sum_{n,m}\overline{a_n}a_m\hat\omega_0(\delta^*(\rho_{S_n,\Gamma}(\Gp))\delta(\rho_{S_m,\Gamma}(\Gp)))\\
&=\sum_{n,m}\overline{a_n}a_m\hat\omega_0(\delta(\rho_{S_n,\Gamma}(\Gp)^{-1}\rho_{S_m,\Gamma}(\Gp)))
\eq
Moreover, for an action $\alpha$ of $\bar G^{\disc}_{\breve S,\Gamma}$ on $L^1(\bar G^{\disc}_{\breve S,\Gamma})$ the action is automorphic and point-norm continuous. 
The state is defined by
\beqs \hat\omega_0(\delta(\rho_{S,\Gamma}(\Gp))):=\left\{\begin{array}{ll} 1&\text{ for  }\rho_{S,\Gamma}(\Gp)=e_G\\
0 &\text{ for  }\rho_{S,\Gamma}(\Gp)\neq e_G\end{array}
\right.
\eqs
Derive
\beq \hat\omega_0(\alpha(\tilde\rho_{S,\Gamma})(\delta(\rho_{S_n,\Gamma}(\Gp))))
&=\hat\omega_0(\delta(\tilde\rho_{S_n,\Gamma}(\Gp)\rho_{S_n,\Gamma}(\Gp)\tilde\rho_{S_n,\Gamma}(\Gp)^{-1}))\\
&=\hat\omega_0(\delta(\rho_{S_n,\Gamma}(\Gp)))
\eq
and conclude that, the state $\hat\omega_0$ is $\bar G^{\disc}_{\breve S,\Gamma}$-invariant.

\begin{defi}\label{pi rep U}Let the surface $S$ has the same surface intersection property for a graph $\Gamma$, let $\breve S$ be a set of surfaces $S_1,...,S_N$ having the same surface intersection property for a graph $\Gamma$. 

The \textbf{generalised group-valued quantum flux operator for a surface $S$} is given by the following non-degenerate representation $\pi_{S,\Gamma}$ of $L^1(\bar G_{S,\Gamma},\mu_{S,\Gamma})$ on the Hilbert space $L^2(\bar G_{S,\Gamma},\mu_{S,\Gamma})$, which satisfies $\|\pi_{S,\Gamma}(f_\Gamma)\|_2\leq\| f_\Gamma\|_1$ and is defined as a $L^2(\bar G_{S,\Gamma},\mu_{S,\Gamma})$-valued Bochner integral
\beqs &\pi_{S,\Gamma}(f_\Gamma)\psi_\Gamma
&:=\int_{\bar G_{S,\Gamma}} \dif\mu_{S,\Gamma}(\rho_{S,\Gamma}(\Gamma)) f_\Gamma(\rho_{S,\Gamma}(\Gamma))U(\rho_{S,\Gamma}(\Gamma))\psi_\Gamma\quad\text{ for }f_\Gamma\in L^1(\bar G_{S,\Gamma},\mu_{S,\Gamma})
\eqs 
and a weakly continuous unitary representation $U$ of $\bar G_{S,\Gamma}$ acting on a vector $\psi_\Gamma$ in $L^2(\bar G_{S,\Gamma},\mu_{S,\Gamma})$ is considered.

The \textbf{generalised group-valued quantum flux operator for a set of surfaces $\check S$} is given by the following non-degenerate representation $\pi_{\breve S,\Gamma}$ of $L^1(\bar G_{\breve S,\Gamma},\mu_{\breve S,\Gamma})$ on the Hilbert space $L^2(\bar G_{\breve S,\Gamma},\mu_{\breve S,\Gamma})$, which satisfies $\|\pi_{\breve S,\Gamma}(f_\Gamma)\|_2\leq\| f_\Gamma\|_1$ and is defined as a $L^2(\bar G_{\breve S,\Gamma},\mu_{\breve S,\Gamma})$-valued Bochner integral
\beqs &\pi_{\breve S,\Gamma}(f_\Gamma)\psi_\Gamma\\
&:=\int_{\bar G_{\breve S,\Gamma}} \dif\mu_{\breve S,\Gamma}(\rho_{S_1}(\gamma_1),...,\rho_{S_N}(\gamma_N))\\&\qquad\qquad f_\Gamma(\rho_{S_1}(\gamma_1),...,\rho_{S_N}(\gamma_N))U_{\breve S,\Gamma}(\rho_{S_1}(\gamma_1),...,\rho_{S_N}(\gamma_N))\psi_\Gamma(\hat\rho_{S_1}(\gamma_1),...,\hat\rho_{S_N}(\gamma_N))
\eqs 
whenever $f_\Gamma\in L^1(\bar G_{\breve S,\Gamma},\mu_{\breve S,\Gamma})$ and a weakly continuous unitary representation $U$ of $G_{\breve S,\Gamma}$ acting on a vector $\psi_\Gamma$ in $L^2(\bar G_{\breve S,\Gamma},\mu_{\breve S,\Gamma})$ is considered. 
\end{defi}
It is easy to show that, for example the representation associated to a left regular representation $U_{\overleftarrow{L}}^N$ of $\bar G_{\breve S,\Gamma}$ on $L^2(\bar G_{S,\Gamma},\mu_{S,\Gamma})$ fulfill 
\beq &\Phi_M(f_\Gamma)\psi_\Gamma(\hat\rho_{S_1}(\gamma_1),..,\hat\rho_{S_N}(\gamma_N))\\
&\qquad=\int_{\bar G_{\breve S,\Gamma}} d\mu_{\breve S,\Gamma}(\rho_{S_1}(\gamma_1),...,\rho_{S_N}(\gamma_N)) f_\Gamma(\rho_{S_1}(\gamma_1),...,\rho_{S_N}(\gamma_N))\\ &\qquad\qquad\qquad U_{\overleftarrow{L}}^N(\rho_{S_1}(\gamma_1),...,\rho_{S_N}(\gamma_N))\psi_\Gamma(\hat\rho_{S_1}(\gamma_1),..,\hat\rho_{S_N}(\gamma_N))\\
&\qquad= \int_{\bar G_{\breve S,\Gamma}} d\mu_{\breve S,\Gamma}(\rho_{S_1}(\gamma_1),...,\rho_{S_N}(\gamma_N)) f_\Gamma(\rho_{S_1}(\gamma_1),...,\rho_{S_N}(\gamma_N))\\ &\qquad\qquad\qquad\psi_\Gamma(\rho_{S_1}(\gamma_1)^{-1}\hat\rho_{S_1}(\gamma_1),..,\rho_{S_N}(\gamma_N)^{-1}\hat\rho_{S_N}(\gamma_N))\\
&\qquad= f_\Gamma\ast \psi_\Gamma\text{ for }\psi_\Gamma\in L^2(\bar G_{S,\Gamma},\mu_{S,\Gamma})
\eq 
It is a $^*$-representation on the Hilbert space $L^2(\bar G_{S,\Gamma},\mu_{S,\Gamma})$, since it is true that,
\beq &\Phi_M(f^1_\Gamma\ast f^2_\Gamma)\psi_\Gamma=\Phi_M(f^1_\Gamma)\Phi_M(f^2_\Gamma)\psi_\Gamma\\
&\Phi_M(\lambda_1f_\Gamma^1+\lambda_2f_\Gamma^2)\psi_\Gamma=\lambda_1\Phi_M(f_\Gamma^1)\psi_\Gamma+\lambda_2\Phi_M(f_\Gamma^2)\psi_\Gamma\\
&\Phi_M(f_\Gamma^*)\psi_\Gamma=\Phi_M(f_\Gamma)^*\psi_\Gamma
\eq yields whenever $f_\Gamma,f^1_\Gamma,f^2_\Gamma\in L^1(\bar G_{\breve S,\Gamma},\mu_{\breve S,\Gamma})$ and $\lambda_1,\lambda_2\in\CB$.

The representation associated to the right regular representation $U_{\overleftarrow{R}}^N$ of $\bar G_{\breve S,\Gamma}$ on $L^2(\bar G_{S,\Gamma},\mu_{S,\Gamma})$ is equivalent to  
\beq &\Phi_M(f_\Gamma)\psi_\Gamma(\hat\rho_{S_1}(\gamma_1),...,\hat\rho_{S_N}(\gamma_N))
\\&:=\int_{\bar G_{\breve S,\Gamma}} \dif\mu_{\breve S,\Gamma}(\rho_{\breve S,\Gamma}(\Gamma))\\&\qquad f_\Gamma(\rho_{\breve S,\Gamma}(\Gamma))
U_{\overleftarrow{R}}^N(\rho_{\breve S,\Gamma}^N)\psi_\Gamma(\hat\rho_{S_1}(\gamma_1),...,\hat\rho_{S_N}(\gamma_N))\\
&=\psi_\Gamma\ast f_\Gamma 
\eq 

Clearly there is a representation of $L^1(\bar G_{S,\Gamma},\mu_{\breve S,\Gamma})$, which correponds to the situation of all paths intersecting with one surface $S$ and such that all paths are outgoing and lie below the surface $S$, on the Hilbert space $L^2(\bar G_{\breve S,\Gamma},\mu_{\breve S,\Gamma})$. The representation is illustrated by
\beq &\pi_{\overleftarrow{L},N}(f_\Gamma)\psi_\Gamma(\hat\rho_{S_1}(\gamma_1),..,\hat\rho_{S_N}(\gamma_N))\\
&\qquad=\int_{G} d\mu(\rho_{S}(\gamma_1),...,\rho_{S}(\gamma_N)) f_\Gamma(\rho_{S}(\gamma_1),...,\rho_{S}(\gamma_N))\\ &\qquad\qquad\qquad U_{\overleftarrow{L}}^{N}(\rho_{S}(\gamma_1),...,\rho_{S}(\gamma_N))\psi_\Gamma(\hat\rho_{S_1}(\gamma_1),..,\hat\rho_{S_N}(\gamma_N))\\
&\qquad= \int_{G} d\mu(\rho_{S}(\gamma_1),...,\rho_{S}(\gamma_N)) f_\Gamma(\rho_{S}(\gamma_1),...,\rho_{S}(\gamma_N))\\ &\qquad\qquad\qquad\psi_\Gamma(\rho_{S}(\gamma_1)\hat\rho_{S_1}(\gamma_1),..,\rho_{S}(\gamma_N)\hat\rho_{S_N}(\gamma_N))\\
&\qquad= \int_{G} d\mu(g_S) f_\Gamma(g_S)\psi_\Gamma(g_S \hat g_{S_1},..,g_S\hat g_{S_N})\\
&\qquad= f_\Gamma\ast \psi_\Gamma\text{ for }\psi_\Gamma\in \HS_\Gamma
\eq for any $i=1,...,N$ and where all surfaces $S_i$ are elements of the surface set $\breve S$. 

Another representation of $L^1(\bar G_{S,\Gamma},\mu_{S,\Gamma})$ on the Hilbert space $L^2(\bar G_{S,\Gamma},\mu_{S,\Gamma})$ is defined by
\beq &\pi_{\overleftarrow{L},1}(f_\Gamma)\psi_\Gamma(\hat\rho_{S}(\gamma_1),...,\hat\rho_{S}(\gamma_N))\\
&\qquad=\int_{G} d\mu_{S,\Gamma}(\rho_{S}(\gamma_1),...,\rho_{S}(\gamma_N)) f_\Gamma(\rho_{S}(\gamma_1),...,\rho_{S}(\gamma_N))\\&\qquad\qquad U_{\overleftarrow{L}}^1(\rho_{S}(\gamma_1),...,\rho_{S}(\gamma_N))\psi_\Gamma(\hat\rho_{S}(\gamma_1),...,\hat\rho_{S}(\gamma_N))\\
&\qquad= \int_{G} d\mu_{S,\Gamma}(\rho_{S}(\gamma_i)) f_\Gamma(\rho_{S}(\gamma_i))\psi_\Gamma(\rho_{S}(\gamma_i)\hat\rho_{S}(\gamma_i))\\
&\qquad= \int_{G} d\mu_{S,\Gamma}(g_S) f_\Gamma(g_S)\psi_\Gamma(g_S \hat g_{S})\\
&\qquad= f_\Gamma\ast \psi_\Gamma\text{ for }\psi_\Gamma\in L^2(\bar G_{S,\Gamma},\mu_{S,\Gamma})
\eq 
Moreover, a general representation $\pi_{\breve S,\Gamma}$ is a faithful regular \footnote{A representation $(\pi,\HS)$ of a $C^*$-algebra $\Alg$ of the form (\ref{pi rep U}) is called \textbf{regular} iff the unitary representation $U$ of a locally compact group $G$ is weak operator continuous on $\HS$. } $^*$-representation of $\CD_r^*(\bar G_{\breve S,\Gamma})$ in $L^2(\bar G_{\breve S,\Gamma},\mu_{\breve S,\Gamma})$. It is a faithful representation, since from $f_\Gamma\ast \psi_\Gamma=0$ it is deducible that, $f_\Gamma=0$ holds. The left and the right regular representations $U_{\overleftarrow{L}}^{1}$ and $U_{\overleftarrow{R}}^{1}$ are unitarily equivalent, hence the generalised representations $\pi_{\overleftarrow{L},1}$ and $\pi_{\overleftarrow{R},1}$ are unitarily equivalent, too.

\begin{defi}
Let $S$ be a surface and $\breve S$ be a set of surfaces such that $S$ and $\breve S$ have the same surface intersection property for a graph $\Gamma$. 

The \textbf{reduced flux group $C^*$-algebra $\CD_r^*(\bar G_{S,\Gamma})$ for a surface $S$ or $\CD_r^*(\bar G_{\breve S,\Gamma})$ for a set $\breve  S$ of surfaces} is defined as the closure of $L^1(\bar G_{S,\Gamma},\mu_{S,\Gamma})$, or respectively $L^1(\bar G_{\breve S,\Gamma},\mu_{\breve S,\Gamma})$, in the norm $\|f_\Gamma\|_r:=\|\pi_{S,\Gamma}(f_\Gamma)\|_2$ or $\|f_\Gamma\|_r:=\|\pi_{\breve S,\Gamma}(f_\Gamma)\|_2$.
\end{defi}

In fact all continuous unitary representations $U$ of the flux group $\bar G_{S,\Gamma}$ on $L^2(\bar G_{S,\Gamma},\mu_{S,\Gamma})$ give a non-degenerate representation $\pi_{S,\Gamma}$ of $L^1(\bar G_{S,\Gamma},\mu_{S,\Gamma})$. Each representation is given by
\beq\label{pi} &\pi_{S,\Gamma}(f_\Gamma):=\int_{\bar G_{S,\Gamma}} \dif\mu_{S,\Gamma}(\rho_{S,\Gamma}(\Gamma)) f_\Gamma(\rho_{S,\Gamma}(\Gamma))U(\rho_{S,\Gamma}(\Gamma))\eq

\begin{defi}
Let $S$ be a surface and $\breve S$ be a set of surfaces such that $S$ and $\breve S$ have the same surface intersection property for a graph $\Gamma$. 

The \textbf{flux group $C^*$-algebras $C^*(\bar G_{S,\Gamma})$ for a surface $S$ or $C^*(\bar G_{\breve S,\Gamma})$ for a set $\breve S$ of surfaces} is the closure of $L^1(\bar G_{S,\Gamma},\mu_{S,\Gamma})$ or $L^1(\bar G_{\breve S,\Gamma},\mu_{\breve S,\Gamma})$ in the norm $\|f_\Gamma\|:=\sup_{\pi}\|\pi(f_\Gamma)\|_2$ where the supremum is taken over all non-degenerate $L^1$-norm decreasing\footnote{A norm $\|.\|$ of $\Alg$ is called $L^1$-norm decreasing if $\|\pi_I(f)\|_2\leq \|f\|$ for all $f\in \Alg$. } $^*$-representations of $L^1(\bar G_{S,\Gamma},\mu_{S,\Gamma})$, or respectively all representations $\pi$ of the form \eqref{pi}, where $U$ is a continuous unitary representation (one representative of each equivalence class) of the flux group $\bar G_{S,\Gamma}$ on a Hilbert space.
\end{defi}

\begin{rem}\label{rem peterweyl}
In the case of a (second countable) compact group $G$ the structures above are well known. Let $\hat G$ be the unitary dual consisting of all unitary equivalence classes of irreducible, continuous and unitary and therefore finite-dimensional representations $\pi_{s,\gamma_i}$ of $G$ w.r.t. a graph $\Gamma:=\{\gamma_i\}$ on a finite dimensional Hibert space $\HS_{s,\gamma_i}$. Notice that, every element of $\hat G$ is one-dimensional, iff $G$ is commutative. The dual $\hat G$ is discrete and countable. The set $\hat G$ is finite, iff $G$ is finite. The finite-dimensional representation $U_{s,\gamma_i}$ is equivalent to the left-regular representation $U_L:G\rightarrow U(L^2(G))$.

There exists an isomorphisms betweeen Hilbert spaces such that
\beqs \HS_\Gamma:=L^2_{\gamma_i}(G)\simeq L^2_{\gamma_i}(\hat G):=\hat\HS_\Gamma 
= \bigoplus_{s\in\hat G}M_{d_{s,\gamma_i}}(\CB)
\eqs where $d_{s,\gamma_i}$ is the dimension of $s$ in $\hat G$, given by the unitary Plancherel transform $\FD: L^2_{\gamma_i}(G)\rightarrow  L^2_{\gamma_i}(\hat G)$ with
\beq \hat \psi_{\gamma_i}(s):=(\FD\psi_{\gamma_i})(s)=\sqrt{d_{s,\gamma_i}}\int_G\dif \mu(\rho_{S_i}(\gamma_i)) U_{s,\gamma_i}(\rho_{S_i}(\gamma_i))\psi_{\gamma_i}(\rho_{S_i}(\gamma_i))
\eq where $\rho_{S_i}(\gamma_i)\in\bar\Gop_{\breve S,\Gamma}$ is identified
 with $G$ if $\breve S:=\{S_i\}_{1\leq i\leq N}$ has the same intersection property for $\Gamma$. The inverse transform is given by 
\beqs \FD^{-1}\hat \psi_{\gamma_i}(s):=\sum_{s\in\hat G}\sqrt{\dim \pi_{s,\gamma_i}}\tr(\hat \psi_{\gamma_i}(s)U_{s,\gamma_i}(\rho_{S_i}(\gamma_i))^*)
\eqs
Clearly if $\psi_\Gamma\in L^2_{\gamma_i}(G)$ and $\hat\psi_\Gamma\in  L^2_{\gamma_i}(\hat G)$ it is true that,
\beq \int_G\vert \psi_\Gamma(\rho_{S_i}(\gamma_i))\vert^2\dif\mu(\rho_{S_i}(\gamma_i))=\sum_{s\in\hat G}(\dim\pi_\sigma)\tr(\hat\psi_\Gamma(s)\hat\psi_\Gamma(s)^*)\eq holds.
Let $\Gamma$ be equivalent to $\{\gamma\}$ and $S$ has the same intersection property for $\Gamma$. The representation $\pi_{S,\Gamma}$ of the $C^*$-algebra $C^*_r(\bar G_{S,\Gamma})$ on the Hilbert space $\HS_\Gamma:=L^2(\bar G_{S,\gamma},\mu_{S,\gamma})$ is given for a path $\gamma$ that intersects $S$ such that the path is outgoing and lies below by 
\beq \pi_{S,\Gamma}(f_\Gamma)\psi_\Gamma &:= \int_G\dif \mu_{S,\gamma}(\rho_{S}(\gamma)) 
f_{\Gamma}(\rho_{S}(\gamma))U_{s,\Gamma}(\rho_{S}(\gamma))\psi_\Gamma(\rho_{S}(\gamma))
\eq for $\psi_\Gamma\in\HS_\Gamma$. 
Notice that, for an abelian (locally) compact flux group $\bar G_{S,\Gamma}$ there is an isomorphism
$\FD:C_r^*(\bar G_{S,\Gamma})\rightarrow C_0(\widehat{ \bar G_{S,\Gamma}})$ given by
\beqs \FD(f_\Gamma)(s)&:= \int_G\dif \mu_{S,\gamma}(\rho_{S}(\gamma)) 
f_{\Gamma}(\rho_{S}(\gamma))U_{s,\Gamma}(\rho_{S}(\gamma))
\eqs which is called the generalised Fourier transform. The set of characters is denoted by $\widehat{ \bar G_{S,\Gamma}}$.
\end{rem}

\begin{exa}
For an abelian locally compact group $G$ the group algebra $C^*(G)$ coincides with $C^*_r(G)$. This is true, since for $s\in\hat G$ the representation $\pi_s$ of $G$ on $L^2(G)$ coincides with $\hat f(s)\in\CB$ and consequently the norm $\|.\|_r$ and $\|.\|$ are the same. 

Moreover, since $\R$ and $\hat \R$ are equal, there are the following isomorphisms
\beqs C_0(\R)\simeq C^*(\R)\simeq C^*_r(\R)
\eqs     
Notice that, this statement generalises for an abelian locally compact group $G$. There is an isomorphism $C^*(G)$ and $C(\hat G)$.

For a general locally compact group $\bar G_{\breve S,\Gamma}$ it is true that,
\beqs C^*_r(\bar G_{\breve S,\Gamma}):=\pi_{S,\Gamma}(C^*(\bar G_{\breve S,\Gamma}))\simeq C^*(\bar G_{\breve S,\Gamma})\setminus \ker (\pi_{S,\Gamma})
\eqs holds. Therefore a Lie group is called amenable, if $C^*(\bar G_{\breve S,\Gamma})$ coincides with $C^*_r(\bar G_{\breve S,\Gamma})$ and hence iff $\pi_{S,\Gamma}$ is faithful.  Since for locally compact groups, the representation $\pi_{S,\Gamma}$ is always faithful, these groups are always amenable.  
\end{exa}

\begin{prop}\label{prop compLiefluxgroup}Let $S$ be a surface with the same surface intersection property for a graph $\Gamma$.

For a compact Lie group $G$ the flux group $C^*$-algebras for surface $S$ and a graph $\Gamma:=\{\gamma\}$ is given by
\beqs C^*_r(\bar G_{\breve S,\Gamma})\simeq C^*(\bar G_{\breve S,\Gamma})\simeq\bigoplus_{\pi_{s,\Gamma}\in \hat G} M_{d_{s,\Gamma}}(\CB)=:M_\Gamma\eqs or
\beqs C^*_r(\bar G_{S,\Gamma})\simeq C^*(\bar G_{S,\Gamma})\simeq\bigoplus_{\pi_{s,\Gamma}\in \hat{\bar G}_{S,\Gamma}} M_{d_{s,\Gamma}}(\CB)\eqs
and, hence, $\bar G_{\breve S,\Gamma}$ is amenable.
\end{prop}
\begin{proofs}
This is due to the remark \ref{rem peterweyl}. 
\end{proofs}

\subsection{The flux transformation group $C^*$-algebra associated to graphs and a surface set}\label{subsec transformalg}

In the general theory for arbitrary locally compact groups the left regular representation $\pi_{\overleftarrow{L},1}$ of $\bar G_{S,\Gamma}$ is defined by $\pi_{\overleftarrow{L},1}(f_\Gamma)\psi_\Gamma:=f_\Gamma\ast\psi_\Gamma$ for $f_{\Gamma}\in L^1(\bar G_{S,\Gamma},\mu_{S,\Gamma})$ on the Hilbert space $L^2(\bar G_{S,\Gamma},\mu_{S,\Gamma})$. The operator $\pi_{\overleftarrow{L},1}(f_\Gamma)$ is compact for every $f_{\Gamma}\in L^1(\bar G_{S,\Gamma},\mu_{S,\Gamma})$. 
The set of functions $\CD(\bar G_{S,\Gamma},\bar G_{S,\Gamma})$ for a locally compact group $G$ is a linear subspace of $\CD(\bar G_{S,\Gamma},C_0(\bar G_{S,\Gamma}))$.
\begin{theo}\label{Generalised Stone- von Neumann theorem}\cite[Theorem 4.24]{Williams07} \textbf{(Generalised Stone- von Neumann theorem): }\\ 
Let $\breve S$ be be a set of surfaces with the simple surface intersection property for a graph $\Gamma$. 

Let $G$ be a locally compact unimodular group, $\bar G_{\breve S,\Gamma}$ be the flux group and let $U$ be a continuous, irreducible and unitary representation of $\bar G_{\breve S,\Gamma}$ on $L^2(\bar G_{\breve S,\Gamma},\mu_{\breve S,\Gamma})=:\HS_\Gamma$. Hence $U\in\Rep(\bar G_{\breve S,\Gamma},\KD(\HS_\Gamma))$.

Let $C_0(\bar G_{\breve S,\Gamma})$ be the $C^*$-algebra of continuous functions vanishing at infinity on $\bar G_{\breve S,\Gamma}$ with a pointwise multiplication and $\sup$-norm and let $\Phi_M$ is the multiplication representation of $C_0(\bar G_{\breve S,\Gamma})$ on $\HS_\Gamma$. Therefore $\Phi_M\in\Mor(C_0(\bar G_{\breve S,\Gamma}),\LD(\HS_\Gamma))$.

Then the linear map $\pi_I: \CD(\bar G_{\breve S,\Gamma},C_0(\bar G_{\breve S,\Gamma}))\rightarrow \LD(L^2(\bar G_{\breve S,\Gamma},\mu_{\breve S,\Gamma}))$ of the form
\beq\label{eq formrepL1} &(\pi_{I}(F_\Gamma)\psi_\Gamma)(\hat\rho_{S_1}(\gamma_1),...,\hat\rho_{S_N}(\gamma_N))
\\&:=\int_{\bar G_{\breve S,\Gamma}} \dif\mu_{\breve S,\Gamma}(\rho_{\breve S,\Gamma}(\Gamma))\Phi_M(F_\Gamma(\rho_{S_1}(\gamma_1),....,\rho_{S_N}(\gamma_N);\hat\rho_{S_1}(\gamma_1),...,\hat\rho_{S_N}(\gamma_N)))\\&\qquad\qquad 
U(\rho_{\breve S,\Gamma}(\Gamma))\psi_\Gamma(\hat\rho_{S_1}(\gamma_1),...,\hat\rho_{S_N}(\gamma_N))
\eq is a faithful and irreducible representation of the convolution $^*$-algebra $\CD(\bar G_{\breve S,\Gamma},C_0(\bar G_{\breve S,\Gamma}))$ of continuous functions $\bar G_{\breve S,\Gamma}\rightarrow C_0(\bar G_{\breve S,\Gamma})$ with compact support acting on the Hilbert space $L^2(\bar G_{\breve S,\Gamma},\mu_{\breve S,\Gamma})$. The convolution $^*$-algebra $\CD(\bar G_{\breve S,\Gamma},C_0(\bar G_{\breve S,\Gamma}))$ is equipped with a norm $\|.\|_1$ such that its completion is given by the Banach $^*$-algebra $L^1(\bar G_{\breve S,\Gamma},C_0( \bar G_{\breve S,\Gamma}))$. Consequently $\pi_I\in\Rep(L^1(\bar G_{\breve S,\Gamma},C_0(\bar G_{\breve S,\Gamma})),\LD(\HS_\Gamma))$.

Set $\|F_\Gamma\|_u:= \sup_{\pi_I} \|\pi_I(F_\Gamma)\|$, where the supremum is taken over all non-degenerate $L^1$-norm decreasing\footnote{A norm $\|.\|$ of $\Alg$ is called $L^1$-norm decreasing if $\|\pi_I(f)\|\leq \|f\|_1$ for all $f\in \Alg$. } $^*$-representations $\pi_I$ of the Banach $^*$-algebra $L^1(\bar G_{\breve S,\Gamma},C_0(\bar G_{\breve S,\Gamma}))$, or respectively over all representations $\pi_I$ of the form \eqref{eq formrepL1} where $(\Phi_M,U_{\overleftarrow{L}}^N)$ is a covariant Hilbert space representation of the $C^*$-dynamical system\\ $(C_0(\bar G_{\breve S,\Gamma}),\alpha_{\overleftarrow{L}}^N,\bar G_{\breve S,\Gamma})$.

Then the range of the closure of $L^1(\bar G_{\breve S,\Gamma},C_0(\bar G_{\breve S,\Gamma}))$ w.r.t. the norm $\|.\|_u$ is called the \textbf{flux transformation group $C^*$-algebra $C^*(\bar G_{\breve S,\Gamma},\bar G_{\breve S,\Gamma})$ for a set $\breve S$ of surfaces and a graph $\Gamma$}.

Moreover, $C^*(\bar G_{\breve S,\Gamma},\bar G_{\breve S,\Gamma})$ is isomorphic to the $C^*$-algebra $\KD(L^2(\bar G_{\breve S,\Gamma},\mu_{\breve S,\Gamma}))$ of compact operators. The $C^*$-algebras $C^*(\bar G_{\breve S,\Gamma},\bar G_{\breve S,\Gamma})$ and $\KD(L^2(\bar G_{\breve S,\Gamma},\mu_{\breve S,\Gamma}))$ are Morita equivalent $C^*$-algebras.
\end{theo}
\begin{proofs}
\textit{Step A.: Existence of the flux transformation group algebra for a graph}\\
The convolution $^*$-algebra $\CD(\bar G_{\breve S,\Gamma},C_0(\bar G_{\breve S,\Gamma}))$ is given by the convolution product
\beqs &(F_\Gamma^1\ast F_\Gamma^2)(\rho_{S_1}(\gamma_1),...,\rho_{S_N}(\gamma_N),\hat\rho_{S_1}(\gamma_1),...,\hat\rho_{S_N}(\gamma_N))\\&=\int_{\bar G_{\breve S,\Gamma}}\dif\mu((\gamma_1),...,\rho_{S_N}(\gamma_N))F_\Gamma^1(\tilde \rho_{S_1}(\gamma_1),...,\tilde \rho_{S_N}(\gamma_N),\hat\rho_{S_1}(\gamma_1),...,\hat\rho_{S_N}(\gamma_N))\\
&\qquad\qquad F_\Gamma^2(\tilde \rho_{S_1}(\gamma_1)^{-1}\rho_{S_1}(\gamma_1),...,\tilde \rho_{S_N}(\gamma_N)^{-1}\rho_{S_N}(\gamma_N),\tilde \rho_{S_1}(\gamma_1)^{-1}\hat\rho_{S_1}(\gamma_1),...,\tilde \rho_{S_N}(\gamma_N)^{-1}\hat\rho_{S_N}(\gamma_N))
\eqs and involution
\beqs F_\Gamma(\rho_{S_1}(\gamma_1),...,\rho_{S_N}(\gamma_N),\hat\rho_{S_1}(\gamma_1),...,\hat\rho_{S_N}(\gamma_N))^*=\overline{F(\rho_{S_1}(\gamma_1)^{-1},...,\rho_{S_N}(\gamma_N)^{-1},\hat\rho_{S_1}(\gamma_1)^{-1},...,\hat\rho_{S_N}(\gamma_N)^{-1})}\
\eqs
Equipp the convolution $^*$-algebra $\CD(\bar G_{\breve S,\Gamma},C_0(\bar G_{\breve S,\Gamma}))$ with the $\|.\|_1$-norm, which is defined by
\beqs &\|F_\Gamma\|_1\\&=\int_{\bar G_{\breve S,\Gamma}}\dif\mu((\gamma_1),...,\rho_{S_N}(\gamma_N))\sup_{\overset{(\hat\rho_{S_1}(\gamma_1),...,\hat\rho_{S_N}(\gamma_N))}{\in \bar G_{\breve S,\Gamma}}}\vert F_\Gamma(\rho_{S_1}(\gamma_1),...,\rho_{S_N}(\gamma_N),\hat\rho_{S_1}(\gamma_1),...,\hat\rho_{S_N}(\gamma_N))\vert
\eqs
and complete the algebra to the Banach $^*$-algebra $L^1(\bar G_{\breve S,\Gamma},C_0(\bar G_{\breve S,\Gamma}))$.

Set $\HS_\Gamma:=L^2(\bar G_{\breve S,\Gamma},C_0(\bar G_{\breve S,\Gamma}))$. Assume that, the surface set has the simple surface property for a graph $\Gamma$ and all paths lie below and are outgoing. Let $U_{\overleftarrow{L}}^N\in\Rep(\bar G_{\breve S,\Gamma},\KD(\HS_\Gamma))$, $F_\Gamma\in \CD(\bar G_{\breve S,\Gamma},\bar G_{\breve S,\Gamma})$. Then the map
\beqs &\pi_{I,\overleftarrow{L}}^N(F_\Gamma)\psi_\Gamma(\hat\rho_{S_1}(\gamma_1),...,\hat\rho_{S_N}(\gamma_N))
\\&:=\int_{\bar G_{\breve S,\Gamma}} \dif\mu(\rho_{S_1}(\gamma_1))...\dif\mu(\rho_{S_N}(\gamma_N))\\
&\qquad F_\Gamma(\rho_{S_1}(\gamma_1),....,\rho_{S_N}(\gamma_N);\hat\rho_{S_1}(\gamma_1),...,\hat\rho_{S_N}(\gamma_N))
U_{\overleftarrow{L}}^N(\rho_{\breve S,\Gamma}^N)\psi_\Gamma(\hat\rho_{S_1}(\gamma_1),...,\hat\rho_{S_N}(\gamma_N)) 
\\&=\int_{\bar G_{\breve S,\Gamma}} \dif\mu(\rho_{S_1}(\gamma_1))...\dif\mu(\rho_{S_N}(\gamma_N))\\
&\qquad F_\Gamma(\rho_{S_1}(\gamma_1),....,\rho_{S_N}(\gamma_N);\hat\rho_{S}(\gamma_1),...,\hat\rho_{S}(\gamma_N))\psi_\Gamma(\rho_{S_1}(\gamma_1)\hat\rho_{S_1}(\gamma_1),...,\rho_{S_N}(\gamma_N)\hat\rho_{S_N}(\gamma_N)) 
\eqs defines a $^*$-homomorphism $\pi_I:\CD(\bar G_{\breve S,\Gamma},\bar G_{\breve S,\Gamma})\rightarrow \LD(\HS_\Gamma)$, which is extended to a $^*$-homomorphism from $\CD(\bar G_{\breve S,\Gamma},C_0(\bar G_{\breve S,\Gamma}))$ to $\LD(\HS_\Gamma)$. Therefore this defines a $^*$-representation of $\CD(\bar G_{\breve S,\Gamma},C_0(\bar G_{\breve S,\Gamma}))$. Furthermore, it extends to a $^*$-representation of $L^1(\bar G_{\breve S,\Gamma},C_0(\bar G_{\breve S,\Gamma}))$ on $\HS_\Gamma$. The representation is faithful, since from $\pi_{I,\overleftarrow{L}}^N(F_\Gamma)\psi_\Gamma=F_\Gamma\ast \psi_\Gamma=0$ it follows that, $F_\Gamma=0$ holds. Clearly this investigation carries over for arbitrary surface sets, which have the simple surface intersection property for $\Gamma$.

\textit{Step B.: Isomorphism between $C^*(\bar G_{\breve S,\Gamma},\bar G_{\breve S,\Gamma})$ and $\KD(L^2(\bar G_{\breve S,\Gamma},\mu_{\breve S,\Gamma}))$}\\ 
Secondly $\pi_I(F_\Gamma)$ is Hilbert-Schmidt if $\|\pi_I(F_\Gamma)\|_2^2<\infty$, which is verified by the following computation
\beqs &\|\pi_{I,\overleftarrow{L}}^N(F_\Gamma)\|_2^2
=\int_{\bar G_{\breve S,\Gamma}} \dif\mu(\rho_{\breve S,\Gamma}(\Gamma))\int_{\bar G_{\breve S,\Gamma}}\dif\mu(\hat\rho_{\breve S,\Gamma}(\Gamma))\\
&\qquad\qquad\qquad\quad \vert F_\Gamma(\rho_{S_1}(\gamma_1),....,\rho_{S_N}(\gamma_N);
\rho_{S_1}(\gamma_1)^{-1}\hat\rho_{S_1}(\gamma_1),...,\rho_{S_N}(\gamma_N)^{-1}\hat\rho_{S_N}(\gamma_N))\vert^2 
\eqs which is finite for every $F_\Gamma\in C^*(\bar G_{\breve S,\Gamma},\bar G_{\breve S,\Gamma})$. Consequently $\pi_{I,\overleftarrow{L}}^N(\CD(\bar G_{\breve S,\Gamma},\bar G_{\breve S,\Gamma}))$ is a subset of the Hilbert Schmidt class $\KD_{HS}(L^2(\bar G_{\breve S,\Gamma}))$, which is a dense subspace (w.r.t. in the usual operator norm) of the $C^*$-algebra $\KD(L^2(\bar G_{\breve S,\Gamma}))$ of compact operators. Hence the closure of $\pi_{I,\overleftarrow{L}}^N(\CD(\bar G_{\breve S,\Gamma},\bar G_{\breve S,\Gamma}))$ is equivalent to $\KD(L^2(\bar G_{\breve S,\Gamma}))$ in the operator norm and equality of the $C^*$-algebra $\pi_{I,\overleftarrow{L}}^N(\CD(\bar G_{\breve S,\Gamma},C_0(\bar G_{\breve S,\Gamma}))$ and $\KD(L^2(\bar G_{\breve S,\Gamma}))$ is due to the fact that, $\pi_{I,\overleftarrow{L}}^N$ is faithful.

\textit{Step C.: All non-degenerate representations of $L^1(\bar G_{\breve S,\Gamma},C_0(\bar G_{\breve S,\Gamma}))$ are unitarily equivalent to $\pi_{I,\overleftarrow{L}}^N$}\\
To show that, there is an isomorphism between the categories of representations of $C^*(\bar G_{\breve S,\Gamma},\bar G_{\breve S,\Gamma})$ and $\KD(L^2(\bar G_{\breve S,\Gamma}))$, which is isomorphic to the representations of $\CB$, on a Hilbert space. This is equivalent to the property of $C^*(\bar G_{\breve S,\Gamma},\bar G_{\breve S,\Gamma})$ and $\KD(L^2(\bar G_{\breve S,\Gamma}))$ being Morita equivalent $C^*$-algebras.
 
\textit{Step 1.: two pre-$C^*$-algebras $\Alg_\Gamma$,$\BAlg$ and a full pre-Hilbert $\BAlg$-module $\E_\Gamma$ }\\  
Assume that, the surface set has the simple surface property for a graph $\Gamma$ and all paths lie below and are outgoing. Let $U_{\overleftarrow{L}}^N\in\Rep(\bar G_{\breve S,\Gamma},\KD(\HS_\Gamma))$.

Consider the pre-$C^*$-algebras $\Alg_\Gamma=\CD(\bar G_{\breve S,\Gamma},\bar G_{\breve S,\Gamma})$ and $\BAlg=\CB$. Moreover, let $\E_\Gamma=C_c(\bar G_{\breve S,\Gamma})$ be a full pre-Hilbert $\CB$-module, which is defined by the $\CB$-action $\pi_R$ on $C_c(\bar G_{\breve S,\Gamma})$, i.o.w. $\pi_R(\lambda)\psi_\Gamma=\psi_\Gamma\lambda$, and the inner product
\beqs \la \psi_\Gamma,\phi_\Gamma\ra_{\CB}:=\la\psi_\Gamma,\phi_\Gamma\ra_2\eqs

\textit{Step 2.: full right Hilbert $\BAlg$-module $\E_\Gamma$}\\ 
The completition of $\E_\Gamma$ is a Hilbert $\CB$-module. 

\textit{Step 3.: left-action of $\Alg_\Gamma$ on $\E_\Gamma$ s.t. $\E_\Gamma$ is a full left pre-Hilbert $\Alg_\Gamma$-module }\\
The left-action of $\Alg_\Gamma$ on $\E_\Gamma$ is defined by $F_\Gamma\psi_\Gamma:=\pi_{I,\overleftarrow{L}}^N(F_\Gamma)\psi_\Gamma$ and therefore 
\beqs &(\pi_{I,\overleftarrow{L}}^N(F_\Gamma)\psi_\Gamma)(\hat\rho_{S_1}(\gamma_1),...,\hat\rho_{S_N}(\gamma_N))
\\&=\int_{\bar G_{\breve S,\Gamma}} \dif\mu_{\breve S,\Gamma}(\rho_{\breve S,\Gamma}(\Gamma))
F_\Gamma(\rho_{S_1}(\gamma_1),....,\rho_{S_N}(\gamma_N);
\rho_{S_1}(\gamma_1)^{-1}\hat\rho_{S_1}(\gamma_1),...,\rho_{S_N}(\gamma_N)^{-1}\hat\rho_{S_N}(\gamma_N))\\
&\qquad\psi_\Gamma(\hat\rho_{S_1}(\gamma_1),...,\hat\rho_{S_N}(\gamma_N))\eqs
and for $F^*_\Gamma\psi_\Gamma:=\pi_{I,\overleftarrow{L}}^N(F^*_\Gamma)\psi_\Gamma$
\beqs 
&(\pi_{I,\overleftarrow{L}}^N(F^*_\Gamma)\psi_\Gamma)(\rho_{S_1}(\gamma_1),...,\rho_{S_N}(\gamma_N))\\
&\qquad=\int_{\bar G_{\breve S,\Gamma}} \dif\mu_{\breve S,\Gamma}(\rho_{\breve S,\Gamma}(\Gamma))
\psi_\Gamma(\hat\rho_{S_1}(\gamma_1),...,\hat\rho_{S_N}(\gamma_N))\\
&\qquad\qquad F_\Gamma(\hat\rho_{S_1}(\gamma_1)\rho_{S_1}(\gamma_1)^{-1},....,\hat\rho_{S_1}(\gamma_N)\rho_{S_N}(\gamma_N)^{-1};
\rho_{S_1}(\gamma_1),....,\rho_{S_N}(\gamma_N))^*\\
&\qquad=\int_{\bar G_{\breve S,\Gamma}} \dif\mu_{\breve S,\Gamma}(\rho_{\breve S,\Gamma}(\Gamma))
\psi_\Gamma(\rho_{S_1}(\gamma_1)^{-1}\hat\rho_{S_1}(\gamma_1),...,\rho_{S_N}(\gamma_N)^{-1}\hat\rho_{S_N}(\gamma_N))\\
&\qquad\qquad \overline{F_\Gamma(\hat\rho_{S_1}(\gamma_1)^{-1},....,\hat\rho_{S_1}(\gamma_N)^{-1};
\rho_{S_1}(\gamma_1),....,\rho_{S_N}(\gamma_N))}\\
\eqs
and there is a $\CD(\bar G_{\breve S,\Gamma}\times \bar G_{\breve S,\Gamma})$-valued inner product on $C_c(\bar G_{\breve S,\Gamma})$ given by
\beqs &\la \phi_\Gamma,\varphi_\Gamma\ra_{\CD(\bar G_{\breve S,\Gamma}\times \bar G_{\breve S,\Gamma})}\\&:=\phi_\Gamma(\rho_{\breve S,\Gamma}(\Gamma))\overline{\varphi_\Gamma
(\hat\rho_{S_1}(\gamma_1)^{-1}\rho_{S_1}(\gamma_1),....,\hat\rho_{S_N}(\gamma_N)^{-1}\rho_{S_N}(\gamma_N))}
\eqs Notice $\CD(\bar G_{\breve S,\Gamma}\times \bar G_{\breve S,\Gamma})\subset \CD(\bar G_{\breve S,\Gamma},\bar G_{\breve S,\Gamma})$. Consequently $C_c(\bar G_{\breve S,\Gamma})$ is a full pre-Hilbert $\Alg_\Gamma$-module. 

\textit{Step 4.: full left Hilbert $\Alg$-module $\E_\Gamma$}\\ 
The completition of $\E_\Gamma$ is a Hilbert $C^*(\bar G_{\breve S,\Gamma},\bar G_{\breve S,\Gamma})$-module. 

\textit{Step 4.: $\Alg_\Gamma$-$\BAlg$-imprimitivity bimodule $\E_\Gamma$}\\ 
\textit{Step 4.1:}\\
Derive
\beqs
&\la\psi_\Gamma, F_\Gamma\phi_\Gamma\ra_{\CB}=\int_G\dif\mu_{\breve S,\Gamma}(\rho_{\breve S,\Gamma})\overline{\psi_\Gamma(\rho_{\breve S,\Gamma})} \pi_{I,\overleftarrow{L}}^N(F_\Gamma)\phi_\Gamma(\rho_{\breve S,\Gamma}(\Gamma))\\
&\la \psi_\Gamma,F_\Gamma\phi_\Gamma\ra_{\CB}
=\int_{\bar G_{\breve S,\Gamma}}\dif\mu_{\breve S,\Gamma}(\hat\rho_{\breve S,\Gamma}(\Gamma)) \int_{\bar G_{\breve S,\Gamma}}  \dif\mu_{\breve S,\Gamma}(\rho_{\breve S,\Gamma}(\Gamma))
\overline{\psi_\Gamma(\hat\rho_{S_1}(\gamma_1),...,\hat\rho_{S_N}(\gamma_N))}\\
&\qquad\qquad\qquad F_\Gamma(\rho_{\breve S,\Gamma}(\Gamma);\hat\rho_{\breve S,\Gamma}(\Gamma))\phi_\Gamma(\rho_{S_1}(\gamma_1)\hat\rho_{S_1}(\gamma_1),...,\rho_{S_N}(\gamma_N)\hat\rho_{S_N}(\gamma_N))\\
&=\int_{\bar G_{\breve S,\Gamma}}\dif\mu_{\breve S,\Gamma}(\hat\rho_{\breve S,\Gamma}(\Gamma)) \int_{\bar G_{\breve S,\Gamma}}  \dif\mu_{\breve S,\Gamma}(\rho_{\breve S,\Gamma}(\Gamma))
\overline{\psi_\Gamma(\rho_{S_1}(\gamma_1)^{-1}\hat\rho_{S_1}(\gamma_1),...,\rho_{S_N}(\gamma_N)^{-1}\hat\rho_{S_N}(\gamma_N))}\\
&\qquad F_\Gamma(\rho_{\breve S,\Gamma}(\Gamma);\rho_{S_1}(\gamma_1)^{-1}\hat\rho_{S_1}(\gamma_1),...,\rho_{S_N}(\gamma_N)^{-1}\hat\rho_{S_N}(\gamma_N))
\phi_\Gamma(\rho_{\breve S,\Gamma}(\Gamma))\\
&=\int_{\bar G_{\breve S,\Gamma}}\dif\mu_{\breve S,\Gamma}(\hat\rho_{\breve S,\Gamma}(\Gamma))\int_{\bar G_{\breve S,\Gamma}} \dif\mu_{\breve S,\Gamma}(\rho_{\breve S,\Gamma}(\Gamma))\overline{\psi_\Gamma(\rho_{S_1}(\gamma_1)^{-1}\hat\rho_{S_1}(\gamma_1),...,\rho_{S_N}(\gamma_N)^{-1}\hat\rho_{S_N}(\gamma_N))}\\
&\qquad\overline{\left(\overline{F_\Gamma(\rho_{\breve S,\Gamma}(\Gamma);
\rho_{S_1}(\gamma_1)^{-1}\hat\rho_{S_1}(\gamma_1),....,\rho_{S_N}(\gamma_N)^{-1}\hat\rho_{S_N}(\gamma_N))}\right)}\phi_\Gamma(\hat\rho_{\breve S,\Gamma}(\Gamma))\\
\eqs
\beqs
\la\psi_\Gamma, F_\Gamma\phi_\Gamma\ra_{\CB}&=
\int_{\bar G_{\breve S,\Gamma}} \dif\mu_{\breve S,\Gamma}(\rho_{\breve S,\Gamma}(\Gamma))\int_{\bar G_{\breve S,\Gamma}} \dif\mu_{\breve S,\Gamma}(\hat\rho_{\breve S,\Gamma}(\Gamma))\\
&\qquad\overline{\psi_\Gamma(\rho_{S_1}(\gamma_1)^{-1}\hat\rho_{S_1}(\gamma_1),...,\rho_{S_N}(\gamma_N)^{-1}\hat\rho_{S_N}(\gamma_N))}\\
&\qquad\overline{F^*_\Gamma(\rho_{\breve S,\Gamma}(\Gamma);
\hat\rho_{S_1}(\gamma_1)\rho_{S_1}(\gamma_1)^{-1},....,\hat\rho_{S_N}(\gamma_N)\rho_{S_N}(\gamma_N)^{-1})}\\
&\qquad\phi_\Gamma(\hat\rho_{\breve S,\Gamma}(\Gamma)) \\
&=\la F_\Gamma^*\psi_\Gamma,\phi_\Gamma\ra_{\CB}=\la \pi_I( F_\Gamma^*)\psi_\Gamma,\phi_\Gamma\ra_{\CB} 
\eqs
for $F_\Gamma\in \CD(\bar G_{\breve S,\Gamma},\bar G_{\breve S,\Gamma})$, $\psi_\Gamma,\phi_\Gamma\in C_c(\bar G_{\breve S,\Gamma})$
and
\beqs \la \lambda\psi_\Gamma ,\phi_\Gamma\ra_{C(\bar G_{\breve S,\Gamma},\bar G_{\breve S,\Gamma})}=\la \psi_\Gamma,\lambda^* \phi_\Gamma\ra_{C(\bar G_{\breve S,\Gamma},\bar G_{\breve S,\Gamma})}=\la \psi_\Gamma,\overline{\lambda} \phi_\Gamma\ra_{C(\bar G_{\breve S,\Gamma},\bar G_{\breve S,\Gamma})}
\eqs 
for $\lambda\in\CB$ and $\psi_\Gamma,\phi_\Gamma\in C_c(\bar G_{\breve S,\Gamma})$.

\textit{Step 4.2:}\\
\beqs \phi_\Gamma\la \psi_\Gamma,\varphi_\Gamma\ra_{\CB}= \pi_R(\la \psi_\Gamma,\varphi_\Gamma\ra_{\CB})\phi_\Gamma 
=\la \phi_\Gamma,\psi_\Gamma\ra_{\CD(\bar G_{\breve S,\Gamma},\bar G_{\breve S,\Gamma})}\varphi_\Gamma =\pi_I(\la \phi_\Gamma,\psi_\Gamma\ra_{\CD(\bar G_{\breve S,\Gamma},\bar G_{\breve S,\Gamma})})\varphi_\Gamma
\eqs for $\phi_\Gamma,\psi_\Gamma,\varphi_\Gamma\in\E$.

\textit{Step 5.: Morita equivalence}\\
Hence conclude that, the $C^*$-algebras $C^*(\bar G_{\breve S,\Gamma},\bar G_{\breve S,\Gamma})$ and $\CB$ are Morita equivalent. Moreover, for two Morita equivalent $C^*$-algebras there is a bijective correspondence between the non-degenerate representations of those two $C^*$-algebras. Consequently all irreducible representations of the $^*$-algebra $\CD(\bar G_{\breve S,\Gamma}, C_0( \bar G_{\breve S,\Gamma}))$ are unitarily equivalent to $\pi_{I,\overleftarrow{L}}^N$. Clearly for different unitarily inequivalent irreducible representations of $\bar G_{\breve S,\Gamma}$, there are different inequivalent irreducible representations of $\CD(\bar G_{\breve S,\Gamma},C_0(\bar G_{\breve S,\Gamma}))$, which corresponds, therefore, to possible superselections of the system. Remark that, every non-degenerate representation of the compact operators $\KD(\HS_\Gamma)$ is equivalent to a direct sum of copies of the identity representation. Hence it follows that, every non-degenerate representation of $C^*(\bar G_{\breve S,\Gamma},\bar G_{\breve S,\Gamma})$ is equivalent to a direct sum of copies of $\pi_{I,\overleftarrow{L}}^N:=\Phi_{M}\ltimes U_{\overleftarrow{L}}^N$, where $\Phi_{M}$ is the multiplication representation of $C_0(\bar G_{\breve S,\Gamma})$ on $\HS_\Gamma$. 
\end{proofs}

To summarise the Generalised Stone- von Neumann theorem \ref{Generalised Stone- von Neumann theorem} states that, there is a bijective correspondence strongly continuous unitary representations of a group $\bar G_{\breve S,\Gamma}$ on the $C^*$-algebra $\LD(\HS_\Gamma)$ and elements of $\Mor(C^*(\bar G_{\breve S,\Gamma},\bar G_{\breve S,\Gamma}),\KD(\HS_\Gamma))$. This correspondence preserves direct sums and irreducibility. 

Furthermore, all unitary representations of $\bar G_{\breve S,\Gamma}$ on $C_0(\Ab_\Gamma)$ for surface sets, which have the simple surface property for $\Gamma$, are naturally elements of the multiplier algebra $M(C^*(\bar G_{\breve S,\Gamma},\bar G_{\breve S,\Gamma}))$. Or equivalently all unitary representations of $\bar G_{S,\Gamma}$ for a surface $S$ having the same surface intersection property for $\Gamma$ are naturally elements of the multiplier algebra $M(C^*(\bar G_{S,\Gamma},\bar G_{S,\Gamma}))$.
Clearly the closed linear span $\{U(\rho_{S,\Gamma}(\Gamma)):\rho_{S,\Gamma}(\Gamma)\in\bar G_{\breve S,\Gamma}\}$ of all unitary representations of $\bar G_{\breve S,\Gamma}$ on the $C^*$-algebra $C(\bar G_{\breve S,\Gamma})$ forms a $C^*$-subalgebra of $M(C^*(\bar G_{\breve S,\Gamma},\bar G_{\breve S,\Gamma}))$.

In the next investigations the question is what happen if different surface sets are used for the construction of the flux transformation group $C^*$-algebra. In particular is there a generalised von Neumann theorem available?

For a simplification the following identifications are used. The flux group $\bar G_{\breve S,\Gamma}$ is identified with $G^N$. Then the following coset spaces (or space of orbits) are defined by the sets 
\beqs 
G^N/ G:=\{&(\rho_{S_1}(\gamma_1)\rho_{S}(\gamma_1),...,\rho_{S_N}(\gamma_N)\rho_{S}(\gamma_N)):\quad\rho_{S}\in \Gop_{S,\Gamma},\rho_{S_i}\in \Gop_{\breve S,\Gamma},\\&\qquad\quad \rho_S(\gamma_i)=\rho_S(\gamma_j)=g_S\in G; 1\leq i,j\leq N\}\\
G^N\setminus G:=\{&(\rho_{S}(\gamma_1)\rho_{S_1}(\gamma_1),...,\rho_{S}(\gamma_N)\rho_{S_N}(\gamma_N))\quad\rho_{S}\in \Gop_{S,\Gamma},\rho_{S_i}\in \Gop_{\breve S,\Gamma},\\&\qquad\quad \rho_S(\gamma_i)=\rho_S(\gamma_j)=g_S\in G; 1\leq i,j\leq N\}
\eqs whenever $\breve S$ is a surface set with simple surface intersection property for $\Gamma$ and $S$ has the same surface intersection property for $\Gamma$.

The space $G^N/G^2$ is identified with $G^2/G^{2}\times G^{N-2}$, which is given by
\beqs G^2/G^2\times G^{N-2}&:=\{(\rho_{S_1}(\gamma_1)\rho_{\bar S_1}(\gamma_1),\rho_{S_2}(\gamma_2)\rho_{\bar S_2}(\gamma_2),\rho_{S_3}(\gamma_3)...,\rho_{S_N}(\gamma_N)):\\
&\qquad\rho_{S_i}\in \Gop_{\breve S,\Gamma},\rho_{\bar S_l}\in \Gop_{\check S,\Gamma},\forall l=1,2; i=1,...,N\text{ and } (\rho_{\bar S_1}(\gamma_1),\rho_{\bar S_2}(\gamma_2))\in G^2 \}\\
&=G^{N-2}
\eqs whenever $\breve S$ is a surface set with simple surface intersection property for $\Gamma$ and $\check S:=\{\bar S_1,\bar S_2\}$ has the simple surface intersection property for $\{\gamma_1,\gamma_2\}$. The space $G^2/G\times G^{N-2}$ is derivable as
\beqs
 G^2/G\times G^{N-2}:=\{&(\rho_{S_1}(\gamma_1)\rho_{S}(\gamma_1),\rho_{S_2}(\gamma_2)\rho_{S}(\gamma_2),\rho_{S_3}(\gamma_3)...,\rho_{S_N}(\gamma_N)):\\
&\qquad\rho_{S_i}\in \Gop_{\breve S,\Gamma},\forall i=1,...,N;\rho_{S}\in \Gop_{S,\Gamma}\text{ and } \rho_{S}(\gamma_k)\in G,\forall k=1,2 \}
\eqs whenever $\breve S$ is a surface set with simple surface intersection property for $\Gamma$.

Or more general define  
\beqs &G^N/G^{N-M}=G^{N-M}/G^{N-M}\times G^{M}\\
&:=\{(\rho_{S_1}(\gamma_1)\rho_{\bar S_1}(\gamma_1),...,\rho_{S_{N-M}}(\gamma_{N-M})\rho_{\bar S_{N-M}}(\gamma_{N-M}),\rho_{S_{N-M+1}}(\gamma_{N-M+1}),...,\rho_{S_N}(\gamma_N)):\\
&\qquad\rho_{S_i}\in \Gop_{\breve S,\Gamma},\rho_{\bar S_i}\in \Gop_{\check S,\Gamma}\text{ and } (\rho_{\bar S_1}(\gamma_1),...,\rho_{\bar S_{N-M}}(\gamma_{N-M}))\in G^{N-M} \}\\
&\text{ or }\\
&G^{N-M}/G\times G^{M}\\
&:=\{(\rho_{S_1}(\gamma_1)\rho_{S}(\gamma_1),...,\rho_{S_{N-M}}(\gamma_{N-M})\rho_{S}(\gamma_{N-M}),\rho_{S_{N-M+1}}(\gamma_{N-M+1}),...,\rho_{S_N}(\gamma_N)):\\&\qquad \rho_{S_i}\in \Gop_{\breve S,\Gamma},\rho_{S}\in \Gop_{S,\Gamma},(\rho_{\bar S_1}(\gamma_1),...,\rho_{\bar S_{N-M}}(\gamma_{N-M}))\in G^{N-M}\text{ and }\\&\qquad\quad \rho_S(\gamma_i)=\rho_S(\gamma_j)\in G\quad i,j=1,...,N\}
\eqs for suitable surface sets $\breve S$ and $\check S$ and a surface $S$.
Hence the coset $G^N/G^{N-1}$ of a group $G^N$ and a closed subgroup $G^{N-1}$ is the set 
\beqs &G^N/G^{N-1}=G^{N-1}/G^{N-1}\times G \\&:=\{(\rho_{S_1}(\gamma_1)\rho_{\bar S_1}(\gamma_1),...,\rho_{S_{N-1}}(\gamma_{N-1})\rho_{\bar S_{N-1}}(\gamma_{N-1}),\rho_{S_N}(\gamma_N)):\\
&\qquad\rho_{S_i}\in \Gop_{\breve S,\Gamma},\rho_{\bar S_i}\in \Gop_{\check S,\Gamma}\text{ and } (\rho_{\bar S_1}(\gamma_1),...,\rho_{\bar S_{N-1}}(\gamma_{N-1}))\in G^{N-1} \}\eqs
for suitable surface sets $\breve S$ and $\check S$. For suitable surface sets $\breve S$, $\check S$ and a graph $\Gamma$ the following theorem is derivable. 

\begin{theo}\label{theo moritaequivgroup} It is true that:
\begin{enumerate}
\item\label{GGN} the algebras $C_0(G^N/G^{N-1})\rtimes G^N$ and $C^*(G^{N-1})$ are Morita equivalent $C^*$-algebras (for $N>1$).
\item\label{GMGNM} The algebras $C_0(G^N/G^{N-M})\rtimes G^N$ and $C^*(G^{N-M})$ are Morita equivalent $C^*$-algebras (for $N>1$ and $1\leq M<N$) 
\end{enumerate}
\end{theo}
\begin{proofs}
In the following the case \ref{GMGNM} is considered.

\textit{Step 1.: two pre-$C^*$-algebras $\Alg_\Gamma$,$\BAlg_\Gamma$ and a full pre-Hilbert $\BAlg$-module $\E_\Gamma$ }\\ 
Set $N$ be equivalent to $\vert \Gamma\vert$ for a graph $\Gamma$. Let $\Alg_\Gamma=\CD(G^N,G^N/G^{N-1})$ be the dense subalgebra of $C^*(G^N,G^N/G^{N-1})$ such that $\CD(G^N,G^N/G^{N-1})$ is a pre-$C^*$-algebra. Similarly let $\BAlg_\Gamma=\CD(G^{N-1})$ be a dense subalgebra of $C^*(G^{N-1})$ such that $\CD(G^{N-1})$ is a pre-$C^*$-algebra. Identify $G^N/G^{N-1}$ with $G$.

The full pre-Hilbert $\CD(G^{N-1})$-module is given by $C_c(G^N)$ and the right action $\psi_\Gamma f_\Gamma:=\pi_R(f_\Gamma)\psi_\Gamma$, which is of the form
\beqs &\pi_R(f_\Gamma)\psi_\Gamma
:= \int_{G^{N-1}}\dif\mu(\rho_{\bar S_1}(\gamma_1),...,\rho_{\bar S_{N-1}}(\gamma_{N-1}))\\
&\qquad\qquad\qquad\qquad 
\psi_\Gamma( \hat\rho_{S_1}(\gamma_1)\rho_{\bar S_1}(\gamma_1),...,\hat\rho_{S_{N-1}}(\gamma_{N-1})\rho_{\bar S_{N-1}}(\gamma_{N-1}),\hat\rho_{S_N}(\gamma_{N}))\\
&\qquad\qquad\qquad\qquad f_\Gamma(\rho_{\bar S_1}(\gamma_1),...,\rho_{\bar S_{N-1}}(\gamma_{N-1}))\\
&\pi_R(f^*_\Gamma)\psi_\Gamma
:= \int_{G^{N-1}}\dif\mu(\rho_{\bar S_1}(\gamma_1),...,\rho_{\bar S_{N-1}}(\gamma_{N-1}))f^*_\Gamma(\rho_{\bar S_1}(\gamma_1),...,\rho_{\bar S_{N-1}}(\gamma_{N-1}))\\
&\qquad\qquad\qquad\qquad 
\psi_\Gamma(\hat\rho_{S_1}(\gamma_1)\rho_{\bar S_1}(\gamma_1)^{-1},...,\hat\rho_{S_{N-1}}(\gamma_{N-1})\rho_{\bar S_{N-1}}(\gamma_{N-1})^{-1},\hat\rho_{S_N}(\gamma_{N}))
\eqs 
for $\psi_\Gamma\in C_c(G^N)$ and $f_\Gamma\in \CD(G^{N-1})$. The $\CD(G^{N-1})$-valued product on $C_c(G^N)$ is given by
\beqs \la \psi_\Gamma,\phi_\Gamma\ra_{\CD(G^{N-1})}&:=\int_{G^N}\dif\mu(\hat\rho_{S_1,\Gamma},...,\hat\rho_{S_N,\Gamma})\overline{\psi_\Gamma(\hat\rho_{S_1}(\gamma_1),...,\hat\rho_{S_N}(\gamma_N))}\\
&\qquad \phi_\Gamma(\hat\rho_{S_1}(\gamma_1)\rho_{\bar S_1}(\gamma_1),...,
\hat\rho_{S_{N-1}}(\gamma_{N-1})\rho_{\bar S_{N-1}}(\gamma_{N-1}),\hat\rho_{S_{N}}(\gamma_{N}))
\eqs
\textit{Step 2.: full right Hilbert $\BAlg$-module $\E_\Gamma$}\\ 
The completition of $C_c(G^{N})$ is a Hilbert $\CD(G^{N-1})$-module. 

\textit{Step 3.: left-action $\pi_L$ of $\Alg_\Gamma$ on $\E_\Gamma$ s.t. $\E_\Gamma$ is a full left pre-Hilbert }\\
Then there is a pre-Hilbert $\CD(G^N,G)$-module is given by $C_c(G^{N})$ and the left action $F_\Gamma\psi_\Gamma:=\pi_L(F_\Gamma)\psi_\Gamma$, which is of the form
\beqs \pi_{L}(F_\Gamma)\psi_\Gamma&:= \int_{G}\dif\mu(\rho_{S}(\gamma_N))\int_{G^N}\dif\mu(\rho_{S_1}(\gamma_1),...,\rho_{S_N}(\gamma_N))\\
&\quad\qquad F_\Gamma(\rho_{S}(\gamma_N);
\hat\rho_{S_1}(\gamma_1)\rho_{S_1}(\gamma_1)^{-1},...,\hat\rho_{S_{N}}(\gamma_{N})\rho_{S_N}(\gamma_{N})^{-1})\\
&\quad\qquad\psi_\Gamma(\rho_{S_1}(\gamma_1),...,\rho_{S_{N}}(\gamma_N))\\
&= \int_{G}\dif\mu(\rho_{S}(\gamma_N))\int_{G^N}\dif\mu(\rho_{S_1}(\gamma_1),...,\rho_{S_N}(\gamma_N))\\
&\quad\qquad F_\Gamma(\rho_{S}(\gamma_N);
\rho_{S_1}(\gamma_1),...,\rho_{S_N}(\gamma_{N}))\\
&\quad\qquad\psi_\Gamma(\rho_{S_1}(\gamma_1)^{-1}\hat\rho_{S_1}(\gamma_1),...,\rho_{S_{N}}(\gamma_{N})^{-1}\hat\rho_{S_{N}}(\gamma_N))
\eqs where $\rho_S(\gamma_i)=\rho_S(\gamma_j)$ for $i,j=1,...,N$
and
\beqs \pi_L(F^*_\Gamma)\psi_\Gamma&:= \int_{G}\dif\mu(\rho_{S}(\gamma_N))\int_{G^N}\dif\mu(\rho_{S_1}(\gamma_1),...,\rho_{S_N}(\gamma_N))\\
&\quad\qquad F^*_\Gamma(\rho_{S}(\gamma_N);
\rho_{S_1}(\gamma_1)\hat\rho_{S_1}(\gamma_1)^{-1},...,\rho_{S_N}(\gamma_{N})\hat\rho_{S_N}(\gamma_N)^{-1})\\
&\quad\qquad\psi_\Gamma(\rho_{S_1}(\gamma_1),...,\rho_{S_{N}}(\gamma_N))\\
&= \int_{G}\dif\mu(\rho_{S}(\gamma_N))\int_{G^N}\dif\mu(\rho_{S_1}(\gamma_1),...,\rho_{S_N}(\gamma_N))\\
&\quad\qquad F^*_\Gamma(\rho_{S}(\gamma_N);
\rho_{S_1}(\gamma_1),...,\rho_{S_N}(\gamma_{N}))\\
&\quad\qquad\psi_\Gamma(\rho_{S_1}(\gamma_1)\hat\rho_{S_1}(\gamma_1),...,\rho_{S_{N}}(\gamma_N)\hat\rho_{S_{N}}(\gamma_{N}))
\eqs
for $F_\Gamma\in \CD(G^N,G)$ and $\psi_\Gamma\in C_c(G^{N})$. 
The $\CD(G^N,G)$-valued inner product on $C_c(G^{N})$ is equal to
\beqs &\la\psi_\Gamma,\phi_\Gamma\ra_{\CD(G^N,G)}\\
&:=\int_{G^{N}}\dif\mu(\hat\rho_{S_1}(\gamma_1),...,\hat\rho_{S_N}(\gamma_N))
\psi_\Gamma(\rho_{S_1}(\gamma_1)^{-1}\hat\rho_{S_1}(\gamma_{1}),...,\rho_{S_N}(\gamma_N)^{-1}\hat\rho_{S_N}(\gamma_{N}))\\
&\qquad\qquad 
\overline{\phi_\Gamma(\hat\rho_{S_1}(\gamma_1),..., \rho_{S}(\gamma_N)^{-1}\hat\rho_{S_{N}}(\gamma_{N}))}
\eqs
for $\psi_\Gamma,\phi_\Gamma\in C_c(G^{N})$.

\textit{Step 4.: $\Alg_\Gamma$-$\BAlg_\Gamma$-imprimitivity bimodule $\E_\Gamma$}\\ 
\textit{Step 4.1:}\\
\beqs &\la \psi_\Gamma f_\Gamma ,\phi_\Gamma\ra_{\CD(G^{N},G)}
=\la \pi_R(f_\Gamma)\psi_\Gamma,\phi_\Gamma\ra_{\CD(G^{N},G)}\\
&=\int_{G^N}\dif\mu(\hat\rho_{S_1}(\gamma_1),...,\hat\rho_{S_N}(\gamma_N))\int_{G}\dif\mu(\rho_{S}(\gamma_N))\int_{G^N}\dif\mu(\rho_{S_1}(\gamma_1),...,\rho_{S_N}(\gamma_N))\\
&\qquad\psi_\Gamma(\rho_{S_1}(\gamma_1)^{-1}\hat\rho_{S_1}(\gamma_{1})\rho_{\bar S_1}(\gamma_{1}),...,\rho_{S_{N-1}}(\gamma_{N-1})^{-1}\hat\rho_{S_{N-1}}(\gamma_{N-1})\rho_{\bar S_{N-1}}(\gamma_{N-1}),\rho_{S_N}(\gamma_N)^{-1}\hat\rho_{S_{N}}(\gamma_{N}))\\ 
&\qquad f_\Gamma(\rho_{\bar S_1}(\gamma_1),...,\rho_{\bar S_{N-1}}(\gamma_{N-1}))
\overline{\phi_\Gamma(\rho_{S_1}(\gamma_1)^{-1}\hat\rho_{S_1}(\gamma_1),...,\rho_{S_N}(\gamma_N)^{-1}\rho_S(\gamma_N)^{-1}\hat\rho_{S_N}(\gamma_{N}))}\\
&=\int_{G^N}\dif\mu(\hat\rho_{S_1}(\gamma_1),...,\hat\rho_{S_N}(\gamma_N))\int_{G}\dif\mu(\rho_{S}(\gamma_N))\int_{G^N}\dif\mu(\rho_{S_1}(\gamma_1),...,\rho_{S_N}(\gamma_N))\\
&\qquad \psi_\Gamma(\rho_{S_1}(\gamma_1)\rho_{\bar S_1}(\gamma_1),...,\rho_{S_{N-1}}(\gamma_{N-1})\rho_{\bar S_{N-1}}(\gamma_{N-1}),\rho_{S_N}(\gamma_N))\\
&\qquad
\overline{\overline{f_\Gamma(\rho_{\bar S_1}(\gamma_1),...,\rho_{\bar S_{N-1}}(\gamma_{N-1}))}
\phi_\Gamma(\hat\rho_{S_1}(\gamma_1),...,\rho_S(\gamma_N)^{-1}\hat\rho_{S_{N}}(\gamma_{N}))}\\
&=\int_{G^N}\dif\mu(\hat\rho_{S_1}(\gamma_1),...,\hat\rho_{S_N}(\gamma_N))\int_{G}\dif\mu(\rho_{S}(\gamma_N))\int_{G^N}\dif\mu(\rho_{S_1}(\gamma_1),...,\rho_{S_N}(\gamma_N))\\
&\qquad \psi_\Gamma(\rho_{S_1}(\gamma_1),...,\rho_{S_{N-1}}(\gamma_{N-1}),\rho_{S_N}(\gamma_N))\\
&\qquad
\overline{f^*_\Gamma(\rho_{\bar S_1}(\gamma_1)^{-1},...,\rho_{\bar S_{N-1}}(\gamma_{N-1})^{-1})}\\
&\qquad\overline{
\phi_\Gamma(\hat\rho_{S_1}(\gamma_1)\rho_{\bar S_1}(\gamma_1)^{-1},...,
\hat\rho_{S_{N-1}}(\gamma_{N-1})\rho_{\bar S_{N-1}}(\gamma_{N-1})^{-1},
\rho_S(\gamma_N)^{-1}\hat\rho_{S_{N}}(\gamma_{N}))}\\
&=\la \psi_\Gamma,\pi_R(f^*_\Gamma)\phi_\Gamma\ra_{\CD(G^{N-1})}\\
&=\la \psi_\Gamma,\phi_\Gamma f^*_\Gamma\ra_{\CD(G^{N-1})}
\eqs
and
\beqs &\la \psi_\Gamma,F_\Gamma\phi_\Gamma\ra_{\CD(G^{N-1})}=\la \psi_\Gamma,\pi_L(F_\Gamma)\phi_\Gamma\ra_{\CD(G^{N-1})}\\
&=\int_{G}\dif\mu(\rho_{S}(\gamma_N))\int_{G^N}\dif\mu(\hat\rho_{S_1}(\gamma_1),...,\hat\rho_{S_N}(\gamma_N))\int_{G^N}\dif\mu(\rho_{S_1}(\gamma_1),...,\rho_{S_N}(\gamma_N))\\
&\qquad\overline{\psi_\Gamma(\hat\rho_{S_1}(\gamma_1),...,\hat\rho_{S_N}(\gamma_N))} 
F_\Gamma(\rho_{S}(\gamma_N);\rho_{S_1}(\gamma_1),...,
\rho_{S_{N}}(\gamma_{N}))\\
&\qquad \phi_\Gamma(\rho_{S_1}(\gamma_1)^{-1}\hat\rho_{S_1}(\gamma_1)\rho_{\bar S_1}(\gamma_1),...,\rho_{S_{N-1}}(\gamma_{N-1})^{-1}\hat\rho_{S_{N-1}}(\gamma_{N-1})\rho_{\bar S_{N-1}}(\gamma_{N-1}), \rho_{S_N}(\gamma_N)^{-1}\hat\rho_{S_{N}}(\gamma_{N}))
\eqs
\beqs\la \psi_\Gamma,F_\Gamma\phi_\Gamma\ra_{\CD(G^{N-1})}
&=\int_{G}\dif\mu(\rho_{S}(\gamma_N))\int_{G^N}\dif\mu(\hat\rho_{S_1}(\gamma_1),...,\hat\rho_{S_N}(\gamma_N))\int_{G^N}\dif\mu(\rho_{S_1}(\gamma_1),...,\rho_{S_N}(\gamma_N))\\
&\qquad\overline{\psi_\Gamma(\rho_{S_1}(\gamma_1)\hat\rho_{S_1}(\gamma_1),...,\rho_{S_N}(\gamma_N)\hat\rho_{S_N}(\gamma_N)) 
F^*_\Gamma(\rho_{S}(\gamma_N);\rho_{S_1}(\gamma_1),...,
\rho_{S_{N}}(\gamma_{N}))}\\
&\qquad\phi(\hat\rho_{S_1}(\gamma_1)\rho_{\bar S_1}(\gamma_1),...,\hat\rho_{S_{N}}(\gamma_{N})\rho_{\bar S_{N}}(\gamma_{N}))\\
&=\la \pi_L(F^*_\Gamma)\psi_\Gamma,\phi_\Gamma\ra_{\CD(G^{N-1})}\\
&=\la F^*_\Gamma\psi_\Gamma,\phi_\Gamma\ra_{\CD(G^{N-1})}
\eqs
\textit{Step 4.2:}\\
The following is true
\beqs &\phi_\Gamma\la \psi_\Gamma,\varphi_\Gamma\ra_{\CD(G^{N-1})}= \pi_R(\la \psi_\Gamma,\varphi_\Gamma\ra_{\CD(G^{N-1})})\phi_\Gamma\\ 
&= \int_{G^{N-1}}\dif\mu(\rho_{\bar S_1}(\gamma_1),...,\rho_{\bar S_{N-1}}(\gamma_{N-1}))\int_{G^{N}}\dif\mu(\hat\rho_{S_1}(\gamma_1),...,\hat\rho_{S_N}(\gamma_N))\\
&\qquad \phi_\Gamma(\hat\rho_{S_1}(\gamma_1)\rho_{\bar S_1}(\gamma_1),...,\hat\rho_{S_{N-1}}(\gamma_{N-1})\rho_{\bar S_{N-1}}(\gamma_{N-1}),\hat\rho_{S_{N}}(\gamma_{N}))\overline{\psi_\Gamma(\hat\rho_{S_{1}}(\gamma_{1}),...,\hat\rho_{S_{N}}(\gamma_{N}))}\\
&\qquad \varphi_\Gamma(\hat\rho_{S_{1}}(\gamma_{1})\rho_{\bar S_{1}}(\gamma_{1}),...,\hat\rho_{S_{N-1}}(\gamma_{N-1})\rho_{\bar S_{N-1}}(\gamma_{N-1}),\hat\rho_{S_{N}}(\gamma_{N}))\\
&= \int_{G^{N-1}}\dif\mu(\rho_{\bar S_1}(\gamma_1),...,\rho_{\bar S_{N-1}}(\gamma_{N-1}))\int_{G^{N}}\dif\mu(\hat\rho_{S_1}(\gamma_1),...,\hat\rho_{S_N}(\gamma_N))\\
&\qquad \phi_\Gamma(\rho_{\bar S_1}(\gamma_1)\hat\rho_{S_1}(\gamma_1),...,\rho_{\bar S_{N-1}}(\gamma_{N-1})\hat\rho_{S_{N-1}}(\gamma_{N-1}),\hat\rho_{S_{N}}(\gamma_{N}))\overline{\psi_\Gamma(\hat\rho_{S_{1}}(\gamma_{1}),...,\hat\rho_{S_{N}}(\gamma_{N}))}\\
&\qquad \varphi_\Gamma(\rho_{\bar S_{1}}(\gamma_{1})\hat\rho_{S_{1}}(\gamma_{1}),...,\rho_{\bar S_{N-1}}(\gamma_{N-1})\hat\rho_{S_{N-1}}(\gamma_{N-1}),\hat\rho_{S_{N}}(\gamma_{N}))
\eqs
\beqs
&\la \phi_\Gamma,\psi_\Gamma\ra_{\CD(G^N,G)}\varphi_\Gamma =\pi_L(\la \phi_\Gamma,\psi_\Gamma\ra_{\CD(G^N,G)})\varphi_\Gamma\\
&= \int_{G}\dif\mu(\rho_{S}(\gamma_N))\int_{G^{N}}\dif\mu(\rho_{S_1}(\gamma_1),...,\rho_{S_N}(\gamma_N))\int_{G^{N}}\dif\mu(\hat\rho_{S_1}(\gamma_1),...,\hat\rho_{S_N}(\gamma_N))\\
&\qquad\phi_\Gamma(\rho_{S_{1}}(\gamma_{1})^{-1}\hat\rho_{S_{1}}(\gamma_{1}),...,\rho_{S_N}(\gamma_N)^{-1}\hat\rho_{S_{N}}(\gamma_{N}))
\overline{\psi_\Gamma(\hat\rho_{S_{1}}(\gamma_{1}),...,\rho_{S}(\gamma_{N})^{-1}\hat\rho_{S_{N}}(\gamma_{N}))}\\
&\qquad\varphi_\Gamma(\rho_{S_1}^{-1}(\gamma_{1})\hat\rho_{S_{1}}(\gamma_{1}),...,\rho_{S_N}^{-1}(\gamma_{N})\hat\rho_{S_{N}}(\gamma_{N}))\\
&= \int_{G}\dif\mu(\rho_{S}(\gamma_N))\int_{G^{N}}\dif\mu(\rho_{S_1}(\gamma_1),...,\rho_{S_N}(\gamma_N))\int_{G^{N}}\dif\mu(\hat\rho_{S_1}(\gamma_1),...,\hat\rho_{S_N}(\gamma_N))\\
&\qquad\phi_\Gamma(\rho_{S_{1}}(\gamma_{1})^{-1}\hat\rho_{S_{1}}(\gamma_{1}),...,\rho_{S_N}(\gamma_N)^{-1}\rho_{S}(\gamma_{N})\hat\rho_{S_{N}}(\gamma_{N}))
\overline{\psi_\Gamma(\hat\rho_{S_{1}}(\gamma_{1}),...,\hat\rho_{S_{N}}(\gamma_{N}))}\\
&\qquad\varphi_\Gamma(\rho_{S_1}^{-1}(\gamma_{1})\hat\rho_{S_{1}}(\gamma_{1}),...,\rho_{S_N}^{-1}(\gamma_{N})\rho_{S}(\gamma_{N})\hat\rho_{S_{N}}(\gamma_{N}))
\eqs for $\phi_\Gamma,\psi_\Gamma,\varphi_\Gamma\in C_c(G^N)$.
Then 
\beqs &\phi_\Gamma\la \psi_\Gamma,\varphi_\Gamma\ra_{\CD(G^{N-1})}=\la \phi_\Gamma,\psi_\Gamma\ra_{\CD(G^N,G)}\varphi_\Gamma
\eqs
since the properties of the surfaces and paths force the identity
\beqs
&\int_{G^{N-1}}\dif\mu(\rho_{\bar S_1}(\gamma_1),...,\rho_{\bar S_{n-1}}(\gamma_{N-1}))\\
&\qquad\phi_\Gamma(\rho_{\bar S_1}(\gamma_1)\hat\rho_{S_1}(\gamma_1),...,\rho_{\bar S_{N-1}}(\gamma_{N-1})\hat\rho_{S_{N-1}}(\gamma_{N-1}),\hat\rho_{S_{N}}(\gamma_{N}))\\
&\qquad \varphi_\Gamma(\rho_{\bar S_{1}}(\gamma_{1})\hat\rho_{S_{1}}(\gamma_{1}),...,\rho_{\bar S_{N-1}}(\gamma_{N-1})\hat\rho_{S_{N-1}}(\gamma_{N-1}),\hat\rho_{S_{N}}(\gamma_{N}))\\
&= \int_{G}\dif\mu(\rho_{S}(\gamma_N))\int_{G^{N}}\dif\mu(\rho_{S_1}(\gamma_1),...,\rho_{S_N}(\gamma_N))\\
&\qquad\phi_\Gamma(\rho_{S_{1}}(\gamma_{1})^{-1}\hat\rho_{S_{1}}(\gamma_{1}),...,\rho_{S_N}(\gamma_N)^{-1}\rho_{S}(\gamma_{N})\hat\rho_{S_{N}}(\gamma_{N}))\\
&\qquad\varphi_\Gamma(\rho_{S_1}^{-1}(\gamma_{1})\hat\rho_{S_{1}}(\gamma_{1}),...,\rho_{S_N}^{-1}(\gamma_{N})\rho_{S}(\gamma_{N})\hat\rho_{S_{N}}(\gamma_{N}))\\
&= \int_{G}\dif\mu(\rho_{S}(\gamma_N))\int_{G^{N}}\dif\mu(\rho_{S_1}(\gamma_1),...,\rho_{S_N}(\gamma_N))\\
&\qquad\phi_\Gamma(\rho_{S_{1}}(\gamma_{1})^{-1}\hat\rho_{S_{1}}(\gamma_{1}),...,\rho_{S_{N-1}}(\gamma_{N-1})^{-1}\hat\rho_{S_{N-1}}(\gamma_{N-1}),
\hat\rho_{S_{N}}(\gamma_{N}))\\
&\qquad\varphi_\Gamma(\rho_{S_1}^{-1}(\gamma_{1})\hat\rho_{S_{1}}(\gamma_{1}),...,\rho_{S_{N-1}}^{-1}(\gamma_{N})\hat\rho_{S_{N-1}}(\gamma_{N-1}),
\hat\rho_{S_{N}}(\gamma_{N}))
\eqs

The case \ref{GGN} is derivable for the dense subalgebra $\Alg_\Gamma:=\CD(G^N,G^M)$ of $C^*(G^N,G^M)$ such that $\CD(G^N,G^M)$ is a pre-$C^*$-algebra. Similarly, let $\BAlg_\Gamma=\CD(G^{N-M})$ be a dense subalgebra of $C^*(G^{N-M})$ such that $\CD(G^{N-M})$ is a pre-$C^*$-algebra. Then $C_c(G^{N})$ is a full left Hilbert $\CD(G^{N},G^M)$-module or full right Hilbert $\CD(G^{N-M})$-module. Moreover, $C_c(G^{N})$ is a $\Alg_\Gamma$-$\BAlg_\Gamma$-imprimitivity bimodule.  
\end{proofs}

\subsection{The non-commutative holonomy and the heat-kernel holonomy $C^*$-algebra for graphs and a surface set}

Assume that, the configuration set $\Ab_\Gamma$ of generalised connections is naturally identified with $G^{\vert \Gamma\vert}$. Consider the convolution algebra $\CD(\Ab_\Gamma)$. Observe that, the convolution product is for example for a graph $\Gamma:=\{\gp\}$ defined by
\beqs (f_\Gamma\ast k_\Gamma)(\ho_\Gamma(\gp))
&=\int_{\Ab_\Gamma}\dif\mu_\Gamma(\go_\Gamma(\gp))f_\Gamma(\go_\Gamma(\gp))k_\Gamma(\go_\Gamma(\gp)^{-1}\ho_\Gamma(\gp))
\eqs
 
The \textbf{non-commutative holonomy $C^*$-algebra for a graph} is given by the object $C^*_{r}(\Ab_\Gamma)$ and reduces in the case of a compact Lie group $G$ to the following object. 

\begin{rem}
In the case of a compact group $G$ the holonomy algebra $C^*(\Ab_\Gamma)$ for a graph $\Gamma$ is equivalent to a $C^*$-algebra of matrices. 

The new algebra is given by the infinite matrix algebra \[M_\Gamma:=\bigotimes_{\gamma_i\in\Gamma}\bigoplus_{\pi_{s,\gamma_i}\in \hat G} M_{d_{s,\gamma_i}}(\CB),\] where $\hat G$ is the dual of $G$, $\pi_{s,\gamma_i}$ is a representation of $G$ associated to a path $\gamma_i$ and $d_{s,\gamma_i}$ is the dimension of the representation $\pi_{s,\gamma_i}$. Finally, the inductive limit of a increasing family of matrix algebras $M_{\Gamma_i}$ associated to graphs is considered.

For an inductive family $\{\Gamma_i\}$ of graphs there is an injective $^*$-homomorphism $\hat\beta_{\Gamma,\Gp}:C^*(\Ab_\Gamma)\rightarrow C^*(\Ab_\Gp)$ for all $\PD_\Gamma\leq\PD_\Gp$. This $^*$-homomorphism is for example given for a subgraph $\Gamma:=\{\gamma\}$ of $\Gp:=\{\gamma\circ\gp\}$ by
\beqs (\hat\beta_{\Gamma,\Gp}(f_{\Gamma}))(\ho_\Gamma(\gamma)):=f_{\Gp}(\ho_\Gp(\gamma\circ\gp))
\eqs  Consequently, there exists an inductive family of $C^*$-algebras $\{(C^*(\Ab_\Gamma),\hat\beta_{\Gamma,\Gp}):\PD_\Gamma\leq\PD_\Gp\}$.

An increasing family of finite matrix algebras are used to define UHF (uniformly hyperfinite) algebras, which are often used in quantum statistical mechanical systems. Furthermore, some of these algebras lead to KMS-states, which are fruitful states such that the dynamics of Loop Quantum Gravity is implemented. 
For further discussion refer to \cite{Kaminski0,Kaminski4,KaminskiPHD}. 
\end{rem}

Clearly, a similar algebra for the flux group and the flux transformation group $C^*$-algebra is constructed for generalised connections. In this case the holonomy transformation group $C^*$-algebra $C^*(\Ab_\Gamma,\Ab_\Gamma)$ is called the \textbf{heat-kernel holonomy $C^*$-algebra}.

\section{The holonomy-flux cross-product $C^*$-algebra for surface sets}\label{subsec holflux}

After the considerations of algebras generated by either quantum configuration or quantum momentum variables, algebras generated by both quantum variables simultaneously is studied in this section. 

There is no particular holonomy-flux cross-product $C^*$-algebra generated by all group-valued quantum flux operators and certain functions depending on holonomies along paths. But there exists a bunch of holonomy-flux cross-product $C^*$-algebra associated to a finite graph system and many different suitable surface sets. These algebras are developed in section \ref{subsubsec holfluxgraph}. The existence of this variety is the consequence of the following facts.

The group-valued quantum flux operators associated to certain surfaces and a graph $\Gamma$ form the flux group associated to a surface set $\breve S$ and a graph $\Gamma$. These elements are implemented as point-norm continuous and automorphic actions on the analytic holonomy $C^*$-algebra $C_0(\Ab_\Gamma)$ restricted to the finite orientation preserved graph system $\PD^{\op}_\Gamma$. For a short notation the analytic holonomy $C^*$-algebra $C_0(\Ab_\Gamma)$ is abreviated by the term \textit{analytic holonomy $C^*$-algebra associated to the graph $\Gamma$}.
It is assumed that, the configuration space is naturally identified with $G^{\vert\Gamma\vert}$. Then the elements of the flux group are represented as unitary operators on the Hilbert space $\HS_\Gamma$, which is given by $L^2(\Ab_\Gamma,\mu_{\Gamma})$. 

For each automorphic action of a certain flux group, which has been presented in \cite[Section 3.1]{Kaminski1}, \cite[Section 6.1]{KaminskiPHD}, a holonomy-flux cross-product $C^*$-algebra is constructed. Precisely, an automorphic action $\alpha$ of the flux group $\bar G_{\breve S,\Gamma}$ on $C_0(\Ab_\Gamma)$ defines a \textit{holonomy-flux cross-product $C^*$-algebra associated to the graph $\Gamma$ and the surface set $\breve S$}. This $C^*$-algebra is denoted by $C_0(\Ab_\Gamma)\rtimes_{\alpha}\bar G_{\breve S,\Gamma}$. 

There are many different possible actions of flux groups depending on a surface or a surface set. For example in \cite[Lemma 3.16]{Kaminski1}, \cite[Lemma 6.1.16]{KaminskiPHD} there is the point-norm continuous automorphic action $\alpha_{\overleftarrow{L}}^1$ of the flux group $\bar G_{S,\Gamma}$ associated to one suitable surface $S$ on the analytic holonomy $C^*$-algebra $C_0(\Ab_\Gamma)$. Moreover, the action $\alpha_{\overleftarrow{L}}^N$ is defined for a flux group associated to a set $\breve S$ of surfaces, which has the simple surface intersection property for a finite orientation preserved graph system associated to a graph $\Gamma$. In the following these two actions are often used.

Finally, there is an algebra, which unifies all cross-product algebras associated to a graph and different suitable sets of surfaces. This algebra is given by the \textit{multiplier algebra of the cross-product algebra $C_0(\Ab_\Gamma)\rtimes_{\alpha_{\overleftarrow{L}}^N}\bar G_{\breve S,\Gamma}$ of a certain surface set $\breve S$}. In theorem \ref{prop multilpiercrossprod} it is proven that this algebra contains the cross-product $C^*$-algebra associated to the graph $\Gamma$ and suitable surface sets and every Weyl element, which is obtained by the unitary representation of flux groups associated to the graph and suitable surface sets.

The inductive limit of  the inductive families of holonomy-flux cross-product $C^*$-algebras is studied in section \ref{subsec holfluxlimit}. There the inductive limit $C^*$-algebra is derived from the inductive limit of $C^*$-algebras restricted to finite orientation preserved graph systems. This algebra is called the \textit{holonomy-flux cross-product $C^*$-algebra} (of a special surface configuration $\breve S$).

\subsection{The holonomy-flux cross-product $C^*$-algebra associated to a graph and a surface set}\label{subsubsec holfluxgraph}

For the development of such a cross-product algebra generated by holonomies along paths and quantum fluxes the following Banach $^*$-algebra is fundamental.

\begin{defi}Let $\breve S$ be a set of surfaces with same surface intersection property for a finite orientation preserved graph system associated to a graph $\Gamma$ with $N$ independent edges. Furthermore, let $(\bar G_{\breve S,\Gamma},C_0(\Ab_\Gamma),\alpha)$ be a $C^*$-dynamical system defined by a point-norm continuous automorphic flux action $\alpha$ of $\bar G_{\breve S,\Gamma}$ on the analytic holonomy $C^*$-algebra $C_0(\Ab_\Gamma)$ associated to a graph $\Gamma$. 

The space $L^1(\bar G_{\breve S,\Gamma},C_0(\Ab_\Gamma),\alpha)$ consists of all measurable functions $F_\Gamma: \bar G_{\breve S,\Gamma}\rightarrow C_0(\Ab_\Gamma)$ for which
\beqs \|F_\Gamma\|_1:=\int_{\bar G_{\breve S,\Gamma}}d\mu_{\breve S,\Gamma}(\rho_{S_1}(\gamma_1),...,\rho_{S_N}(\gamma_N))\|F_\Gamma(\rho_{S_1}(\gamma_1),...,\rho_{S_N}(\gamma_N))\|_2<\infty\eqs yields whenever $\rho_{S}\in G_{\breve S,\Gamma}$.
\end{defi}

\begin{prop}\label{convolution}Let $\breve S$ be a set of surfaces with simple surface intersection property for a finite orientation preserved graph system associated to a graph $\Gamma$. Furthermore, let $(\bar G_{\breve S,\Gamma},C_0(\Ab_\Gamma),\alpha_{\overleftarrow{L}}^N)$ be a $C^*$-dynamical system where $\alpha_{\overleftarrow{L}}^N\in\Act(\bar G_{\breve S,\Gamma},C_0(\Ab_\Gamma))$. 

Then the multiplication operation between functions in $L^1(\bar G_{\breve S,\Gamma},C_0(\Ab_\Gamma),\alpha_{\overleftarrow{L}}^N)$:
\beqs &(F_\Gamma\ast \hat F_\Gamma) (\tilde\rho_{S_1}(\gamma_1),...,\tilde\rho_{S_N}(\gamma_N))\\
&= \int_{\bar G_{\breve S,\Gamma}}\dif\mu_{\breve S,\Gamma}(\rho_{\breve S,\Gamma}(\Gamma))\\ &\qquad\qquad F_\Gamma(\rho_{\breve S,\Gamma}(\Gamma))\left(\alpha_{\overleftarrow{L}}^N(\rho^N_{\breve S,\Gamma})(\hat F_\Gamma)\right)(\rho_{S_1}(\gamma_1)^{-1}\tilde\rho_{S_1}(\gamma_1),....,\rho_{S_N}(\gamma_N)^{-1}\tilde\rho_{S_N}(\gamma_N))
\eqs whenever $\rho_{\breve S,\Gamma}(\Gamma)=(\rho_{S_1}(\gamma_1),...,\rho_{S_N}(\gamma_N))=:\rho^N_{\breve S,\Gamma}$,  $\rho_{S_i},\tilde\rho_{S_i}\in G_{\breve S,\Gamma}$\\
the involution on $L^1(\bar G_{\breve S,\Gamma},C_0(\Ab_\Gamma),\alpha_{\overleftarrow{L}}^N)$:
\beqs
&F_\Gamma(\rho_{S_1}(\gamma_1),...,\rho_{S_N}(\gamma_N))^*=\left(\alpha_{\overleftarrow{L}}^N(\rho^N_{\breve S,\Gamma})(F_\Gamma^+)\right)(\rho_{S_1}(\gamma_1)^{-1},...,\rho_{S_N}(\gamma_N)^{-1})
\eqs where the involution $^+$ on $C_0(\Ab_\Gamma)$ is given by
\beqs &F_\Gamma^+(\rho_{S_1}(\gamma_1)^{-1},...,\rho_{S_N}(\gamma_N)^{-1}) := \overline{F_\Gamma(\rho_{S_1}(\gamma_1)^{-1},...,\rho_{S_N}(\gamma_N)^{-1})}\eqs
turn $L^1(\bar G_{\breve S,\Gamma},C_0(\Ab_\Gamma),\alpha_{\overleftarrow{L}}^N)$ into a Banach $^*$-algebra.
\end{prop}

In particular let $S$ be a surface having the same surface intersection property for a finite orientation preserved graph system associated to a graph $\Gamma$. Then the action $\alpha_{\overleftarrow{L}}^1$ is defined in \cite[Lemma 3.11]{Kaminski1} or \cite[Lemma 6.1.11]{KaminskiPHD} for a graph $\Gamma$ and the convolution product reads
\beq &(F_\Gamma\ast \hat F_\Gamma) (\tilde\rho_{S}(\gamma_i))\\
&= \int_{G}\dif\mu(\rho_{S}(\gamma_i)) F_\Gamma(\rho_{S}(\gamma_i))\left(\alpha_{\overleftarrow{L}}^1(\rho^1_{S,\Gamma}) (\hat F_\Gamma)\right)(\rho_{S}^{-1}(\gamma_i)\tilde\rho_{S}(\gamma_i))\\
&= \int_{G}\dif\mu(\rho_{S}(\gamma_i)) 
F_\Gamma(\rho_{S}(\gamma_i);\ho_\Gamma(\gamma_1),...,\ho_\Gamma(\gamma_N))\\
&\qquad \quad \hat F_\Gamma(\rho_{S}^{-1}(\gamma_i)\tilde\rho_{S}(\gamma_i);\rho_{S}(\gamma_i)\ho_\Gamma(\gamma_1),....,\rho_{S}(\gamma_i)\ho_\Gamma(\gamma_N))
\eq for any $i=1,..,N$. Since it is true that, $\rho_{S}(\gamma_i)=\rho_{S}(\gamma_j)=g_S\in G$ yields for all $i,j=1,...,N$. Clearly, this convolution algebra equipped with an appropriate involution and norm forms the $^*$-Banach algebra $L^1(\bar G_{S,\Gamma},C_0(\Ab_\Gamma),\alpha_{\overleftarrow{L}}^1)$. 
Notice that, the $^*$-Banach algebras  $L^1(\bar G_{S,\Gamma},C_0(\Ab_\Gamma),\alpha_{\overleftarrow{L}}^M)$ for $1\leq M\leq N$ exists. 

Indeed, there are a lot of different Banach $^*$-algebras depending on the choice of the set of surfaces $\breve S$. Let $\breve S$ has the same surface intersection property for a graph $\Gamma$ such that each path $\gamma_i$, that intersect the surface $S_i$, lie above and ingoing w.r.t. the surface orientation of $S_i$. There are no other intersection points of each path $\gamma_i$ with any other surface $S_j$ where $i\neq j$. Then for the map $F_\Gamma: \bar G_{\breve S,\Gamma}\rightarrow C_0(\Ab_\Gamma)$ write for the image of this function  $F_\Gamma(\rho_{S_1}(\gamma_1),...,\rho_{S_N}(\gamma_N))= F_\Gamma(\rho_{S_1}(\gamma_1),...,\rho_{S_N}(\gamma_N);\ho_\Gamma(\gamma_1),...,\ho_\Gamma(\gamma_N))$ and derive
\beqs &(F_\Gamma\ast \hat F_\Gamma) (\tilde\rho_{S_1}(\gamma_1),...,\tilde\rho_{S_N}(\gamma_N))\\
&= \int_{\bar G_{\breve S,\Gamma}}\dif\mu_{\breve S,\Gamma}(\rho_{\breve S,\Gamma}(\Gamma)) \\&\qquad \quad F_\Gamma(\rho_{\breve S,\Gamma}(\Gamma))\left(\alpha^{\overleftarrow{R}}_N(\rho^N_{\breve S,\Gamma})(\hat F_\Gamma)\right)\big(\rho_{S_1}(\gamma_1)^{-1}\tilde\rho_{S_1}(\gamma_1),....,\rho_{S_N}(\gamma_N)^{-1}\tilde\rho_{S_N}(\gamma_N)\big)\\
\eqs
Hence for a redefined convolution product and involution the $^*$-Banach algebras $L^1(\bar G_{S,\Gamma},C_0(\Ab_\Gamma),\alpha^{\overleftarrow{R}}_M)$ for $1\leq M\leq N$ are studied. Furthermore, it is also possible to construct the $^*$-Banach algebras $L^1(\bar G_{S,\Gamma},C_0(\Ab_\Gamma),\alpha_{\overleftarrow{L}}^{\overleftarrow{R},M})$ for $1\leq M\leq N$ and other algebras of that form for a modified convolution product, which is given in general by
\beqs &(F_\Gamma\ast \hat F_\Gamma) (\tilde\rho_{\breve S,\Gamma}(\Gamma))\\
&= \int_{\bar G_{\breve S,\Gamma}}\dif\mu_{\breve S,\Gamma}(\rho_{\breve S,\Gamma}(\Gamma)) F_\Gamma(\rho_{\breve S,\Gamma}(\Gamma))\left(\alpha(\rho_{\breve S,\Gamma}(\Gamma))\hat F_\Gamma\right)(L(\rho_{\breve S,\Gamma}(\Gamma)^{-1})(\tilde\rho_{\breve S,\Gamma}(\Gamma))
\eqs whenever $\rho_{\breve S,\Gamma}(\Gamma)=(\rho_{S_1}(\gamma_1),...,\rho_{S_N}(\gamma_N))$,  $\rho_{\breve S,\Gamma},\tilde\rho_{\breve S,\Gamma}\in G_{\breve S,\Gamma}$ and a modified involution
\beqs
&F^*_\Gamma(\rho_{\breve S,\Gamma}(\Gamma))=\alpha(\rho_{\breve S,\Gamma}(\Gamma))\left(F_\Gamma^+(\rho_{\breve S,\Gamma}(\Gamma)^{-1})\right)
\eqs is used whenever $\alpha\in\Act(\bar G_{\breve S,\Gamma},C_0(\Ab_\Gamma))$.
Hence, for all well-defined $C^*$-dynamical system $(\bar G_{\breve S,\Gamma},\Alg_\Gamma,\alpha)$ there exists a general Banach $^*$-algebra $L^1(\bar G_{\breve S,\Gamma},\Alg_\Gamma,\alpha)$.

\begin{theo}\label{representation}Let $\breve S$ be a set of surfaces with simple surface intersection property for a finite orientation preserved graph system associated to a graph $\Gamma$. Furthermore, let $(\bar G_{\breve S,\Gamma},C_0(\Ab_\Gamma),\alpha_{\overleftarrow{L}}^N)$ be a $C^*$-dynamical system where $\alpha_{\overleftarrow{L}}^N\in\Act(\bar G_{\breve S,\Gamma},C_0(\Ab_\Gamma))$.

There is a bijective correspondence between non-degenerate $L^1$-norm decreasing $^*$- representations $\pi$ of the Banach $^*$-algebra $L^1(\bar G_{S,\Gamma},C_0(\Ab_\Gamma),\alpha_{\overleftarrow{L}}^N)$ and covariant representations $(\Phi_M,U_{\overleftarrow{L}}^N)$ of the $C^*$-dynamical system\\ $(\bar G_{\breve S,\Gamma},C_0(\Ab_\Gamma),\alpha_{\overleftarrow{L}}^N)$ in $\LD(\HS_\Gamma)$. 

This correspondence is given in one direction by the fact that, the representation $\pi_{I,\overleftarrow{L}}^N$ of $L^1(\bar G_{S,\Gamma},C_0(\Ab_\Gamma),\alpha_{\overleftarrow{L}}^N)$ is defined by a covariant pair $(\Phi_M,U_{\overleftarrow{L}}^N)$ via
\beqs \pi_{I,\overleftarrow{L}}^N(F_\Gamma)\psi_\Gamma
:=\int_{\bar G_{\breve S,\Gamma}}\dif\mu_{\breve S,\Gamma}(\rho_{\breve S,\Gamma}^N)\Phi_M(F_\Gamma(\rho_{\breve S,\Gamma}^N))U_{\overleftarrow{L}}^N(\rho_{\breve S,\Gamma}^N)\psi_\Gamma
\eqs whenever $\rho_{\breve S,\Gamma}^N\in \bar G_{\breve S,\Gamma}$, $F_\Gamma\in L^1(\bar G_{S,\Gamma},C_0(\Ab_\Gamma),\alpha_{\overleftarrow{L}}^N)$ and $\psi_\Gamma\in \HS_\Gamma$.

The other direction is given by the definition of the covariant pair $(\Phi_M,U_{\overleftarrow{L}}^N)$ through the maps\\ 
$F_\Gamma: \rho_{\breve S,\Gamma}^N\mapsto F_\Gamma (\rho_{\breve S,\Gamma}^N)$ and\\ 
$\alpha_{\overleftarrow{L}}^N(\rho_{\breve S,\Gamma}^N)(F_\Gamma): \tilde\rho_{\breve S,\Gamma}^N\mapsto \left(\alpha_{\overleftarrow{L}}^N(\tilde\rho_{\breve S}^N)(F_\Gamma)\right)(L((\tilde\rho_{\breve S,\Gamma}^N)^{-1})(\rho_{\breve S,\Gamma}^N))$ such that
\beqs U_{\overleftarrow{L}}^N(\rho_{\breve S,\Gamma}^N)\pi_{I,\overleftarrow{L}}^N(F_\Gamma)\Omega 
&:=\pi_{I,\overleftarrow{L}}^N\left(\alpha_{\overleftarrow{L}}^N(\rho_{\breve S,\Gamma}^N)(F_\Gamma)\right)\Omega\\
\Phi_M(f_\Gamma )\pi_{I,\overleftarrow{L}}^N(F_\Gamma)\Omega 
&:= \pi_{I,\overleftarrow{L}}^N(f_\Gamma F_\Gamma)\Omega
\eqs where $\Omega$ denotes a cyclic vector for $\pi_{I,\overleftarrow{L}}^N(\CD(\bar G_{\breve S,\Gamma},C_0(\Ab_\Gamma)))$,
$f_\Gamma \in C_0(\Ab_\Gamma)$, $\rho_{\breve S,\Gamma}^N,\tilde\rho_{\breve S,\Gamma}^N\in\bar G_{\breve S,\Gamma}$ and\\ $F_\Gamma\in L^1(\bar G_{S,\Gamma},C_0(\Ab_\Gamma),\alpha_{\overleftarrow{L}}^N)$. 
This bijection preserves unitary equivalence, direct sums and irreducibility.

The \textbf{reduced holonomy-flux group $C^*$-algebra $C^*_r(\bar G_{\breve S,\Gamma},C_0(\Ab_\Gamma))$ associated to a graph $\Gamma$ and a set $\breve S$ of surfaces} is defined as the norm-closure of $L^1(\bar G_{S,\Gamma},C_0(\Ab_\Gamma),\alpha_{\overleftarrow{L}}^N)$ with respect to the norm $\|F_\Gamma\|:=\|\pi_{I,\overleftarrow{L}}^N(F_\Gamma)\|_2$.
\end{theo}

With no doubt there are a big bunch of reduced holonomy-flux group $C^*$-algebra $C^*_r(\bar G_{\breve S,\Gamma},C_0(\Ab_\Gamma))$ for different graph systems and different sets of surfaces.

\begin{defi}\label{defi int hol-flux-repr}Let $\breve S$ be a set of surfaces with simple surface intersection property for a finite orientation preserved graph system associated to a graph $\Gamma$.

Then the \textbf{Weyl-integrated holonomy-flux representation} w.r.t. a finite orientation preserved graph system associated to a graph $\Gamma$ and a set $\breve S$ of surfaces is given by
\beqs
&\pi^{I,\Gamma}_{E(\breve S)}( F_\Gamma) \psi_\Gamma
= \int_{\bar G_{\breve S,\Gamma}} \dif\mu_{\breve S,\Gamma}(\rho_{\breve S,\Gamma}(\Gamma)) \Phi_M\left(F_\Gamma(\rho_{\breve S,\Gamma}(\Gamma))\right)U(\rho_{\breve S,\Gamma}(\Gamma))\psi_\Gamma
\eqs
for $F_\Gamma\in C^*(\bar G_{\breve S,\Gamma}, C_0(\Ab_\Gamma))$, $\rho_{\breve S,\Gamma}(\Gamma)\in \bar G_{\breve S,\Gamma}$, $U\in\Rep(\bar G_{\breve S,\Gamma},\KD(L^2(\Ab_\Gamma,\mu_\Gamma)))$ and $\psi_\Gamma\in L^2(\Ab_\Gamma,\mu_{\Gamma})$. The Weyl-integrated holonomy-flux representation $\pi^{I,\Gamma}_{E(\breve S)}$ is a $^*$-representation of the $C^*$-algebra $C^*(\bar G_{\breve S,\Gamma}, C_0(\Ab_\Gamma))$ with a norm inherited from the representations $\pi^{I,\Gamma}_{E(\breve S)}$ on $L^2(\Ab_\Gamma,\mu_{\Gamma})$. The representation $\pi^{I,\Gamma}_{E(\breve S)}$ is also denoted by $\Phi_M\rtimes U$.  
\end{defi}

\begin{prop}Let $\breve S$ be a set of surfaces with simple surface intersection property for a finite orientation preserved graph system associated to a graph $\Gamma$. Furthermore, let $(\bar G_{\breve S,\Gamma},C_0(\Ab_\Gamma),\alpha_{\overleftarrow{L}}^N)$ is a $C^*$-dynamical system.

Define for each $F_\Gamma\in \CD (\bar G_{\breve S,\Gamma},C_0(\Ab_\Gamma))$ the norm
\beqs &\|F_\Gamma\|_u:=\sup\Big\{\|(\Phi_M\rtimes U_{\overleftarrow{L}}^N)(F_\Gamma)\|\Big\}\eqs 
where the supremum is taken over all covariant Hilbert space representations $(\Phi_M,U_{\overleftarrow{L}}^N)$ of the $C^*$-dynamical system $(\bar G_{\breve S,\Gamma},C_0(\Ab_\Gamma),\alpha^N_{\overleftarrow{L}})$.

Then $\|.\|_u$ is a norm on $\CD(\bar G_{\breve S,\Gamma},C_0(\Ab_\Gamma))$, which is called the universal norm. The universal norm is dominated by the $\|.\|_1$-norm, and the completition of $\CD(\bar G_{\breve S,\Gamma},C_0(\Ab_\Gamma))$ with respect to $\|.\|_u$ is a $C^*$-algebra called the \textbf{holonomy-flux cross-product $C^*$-algebra} of $C_0(\Ab_\Gamma)$ by $\bar G_{\breve S,\Gamma}$ for a finite orientation preserved graph system associated to a graph $\Gamma$ and a set $\breve S$ of surfaces. Shortly this algebra is denoted by $C_0(\Ab_\Gamma)\rtimes_{\alpha^N_{\overleftarrow{L}}}\bar G_{\breve S,\Gamma}$. 
\end{prop}
Notice that, for a surface $S$ having the same surface intersection property for a finite orientation preserved graph system associated to $\Gamma$, the $L^2(\Ab_\Gamma,\mu_\Gamma)$-norm of an element $F_\Gamma\in C_0(\Ab_\Gamma)\rtimes_{\alpha_{\overleftarrow{L}}^1}\bar G_{S,\Gamma}$ is given by 
\beq\|\pi^{I,\Gamma}_{E(S)}( F_\Gamma)\psi_\Gamma\|_2 
=&\int_{\Ab_{\Gamma}} \int_{\bar G_{S,\Gamma}}\dif\mu_{S,\Gamma}(\rho_{S}(\gamma_i))\dif\mu_{\Gamma}(\ho_\Gamma(\Gamma))\\
&\vert f_\Gamma(\rho_{S,\Gamma}(\gamma_i); \ho_\Gamma(\Gamma))\psi_\Gamma(L(\rho_{S}(\gamma_i))(\hgi),...,L(\rho_{S}(\gamma_i))(\hgnn))\vert^2\eq
whenever $\rho_S(\gamma_i)=\rho_S(\gamma_j)=g_S\in \bar G_{S,\Gamma}$ for $i\neq j$ and $1\leq i,j\leq N$.

The general holonomy-flux cross-product algebra $C_0(\Ab_\Gamma)\rtimes_\alpha\bar G_{\breve S,\Gamma}$ for an action $\alpha\in\Act(\bar G_{\breve S,\Gamma},C_0(\Ab_\Gamma))$ is in the case of a locally compact group $G$ a non-commutative and non-unital $C^*$-algebra.   

Refer to the definitions of resticted graph-diffeomorphisms presented in \cite[Definition 6.2.10]{KaminskiPHD} and consider the non-standard identification of the configuration space $\Ab_\Gamma$ with $G^{\vert\Gamma\vert}$.
\begin{prop}\label{prop crossprodstatenotdiffeo}
The state $\omega_{E(\breve S)}^\Gamma$ on $C_0(\Ab_\Gamma)\rtimes_{\alpha^N_{\overleftarrow{L}}}\bar\ZD_{\breve S,\Gamma}$ associated to the GNS-representation $(\HS_\Gamma,\pi^{I,\Gamma}_{E(\breve S)},\Omega^{I,\Gamma}_{E(\breve S)})$ is not surface-orientation-preserving graph-diffeomorphism invariant, but it is a surface-preserving graph-diffeomorphism invariant state.
\end{prop}
Notice that,
\beqs \zeta_{\sigma}\circ\alpha (\rho_{S,\Gamma}(\Gamma))\neq \alpha (\rho_{S,\Gamma}(\Gamma_\sigma))\circ\zeta_{\sigma}
\eqs holds for every bisection $\sigma\in\mathfrak{B}(\PD_{\Gamma}^{\op})$ and  $\rho_{S,\Gamma}\in G_{\breve S,\Gamma}$. Therefore, it is necessary to restrict the holonomy-flux cross-product $C^*$-algebra to $C_0(\Ab_\Gamma)\rtimes_{\alpha^N_{\overleftarrow{L}}}\bar \ZD_{\breve S,\Gamma}$.

\begin{proofs} Let $(\varphi_\Gamma,\Phi_\Gamma)$ be a graph-diffeomorphism on $\PD_\Gamma$ over $V_\Gamma$, which is surface orientation preserving. Then investigate the following computation:
\beqs
&\omega^\Gamma_E( \theta_{(\varphi_\Gamma,\Phi_\Gamma)}(F_\Gamma))\\ 
&=
\int_{\Ab_{\Gamma}}\int_{\bar G_{\breve S,\Gamma}}\dif\mu_{\Gamma}(\ho_\Gamma(\Phi_\Gamma(\gamma_1)),...,\ho_\Gamma(\Phi_\Gamma(\gamma_N))) 
\dif\mu_{\breve S,\Gamma}(\rho_{\varphi_\Gamma(S_1)}(\Phi_\Gamma(\gamma_1)),...,\rho_{\varphi_\Gamma(S_N)}(\Phi_\Gamma(\gamma_N)))\\ &\qquad\qquad\quad \vert F_\Gamma(\rho_{\breve S,\Gamma}(\Gamma_\sigma); \rho_{\breve S,\Gamma}(\Gamma_\sigma)^{-1}\ho_\Gamma(\Gamma_\sigma))\vert^2\\
&= \int_{\Ab_{\Gamma}}\int_{\bar G_{\breve S,\Gamma}}\dif\mu_{\Gamma}(\ho_\Gamma(\gamma_1),...,\ho_\Gamma(\gamma_N)) \dif\mu_{\breve S,\Gamma}(\rho_{\tilde S_1} (\gamma_1),...,\rho_{\tilde S_N}(\gamma_N)) \vert F_\Gamma(\rho_{\tilde S_1}(\gamma_1),...,\rho_{\tilde S_N}(\gamma_N))\vert^2\\
&\neq \int_{\Ab_{\Gamma}}\int_{\bar G_{\breve S,\Gamma}}\dif\mu_{\Gamma}(\ho_\Gamma(\gamma_1),...,\ho_\Gamma(\gamma_N)) \dif\mu_{\breve S,\Gamma}(\rho_{S_1} (\gamma_1),...,\rho_{S_N}(\gamma_N)) \vert F_\Gamma(\rho_{S_1}(\gamma_1),...,\rho_{S_N}(\gamma_N))\vert^2\\
&= \omega^\Gamma_E(F_\Gamma)
\eqs whenever $\varphi_\Gamma(S_i)=\tilde S_i$, $\tilde S_i\in\breve S$ for all $1\leq i\leq N$ and $\Gamma_\sigma=(\Phi_\Gamma(\gamma_1),...,\Phi_\Gamma(\gamma_N))$. Clearly for $\varphi_\Gamma(S_i)=S_i$ the invariance property is easy to deduce.
\end{proofs}

The different possibilities of orientation of surfaces and the graphs allow to define a bulk of automorphic actions and $C^*$-dynamical systems for the holonomy algebra $C_0(\Ab_\Gamma)$. Therefore, speak about different surface configurations with respect to graphs and define many different holonomy-flux cross-product $C^*$-algebras. For example there are the following holonomy-flux cross-product $C^*$-algebras constructable:  $C_0(\Ab_\Gamma)\rtimes_{\alpha_{\overleftarrow{L}}^M}\bar G_{\breve S_2,\Gamma}$, $C_0(\Ab_\Gamma)\rtimes_{\alpha_{\overrightarrow{L}}^M}\bar G_{\breve S_3,\Gamma}$,  $C_0(\Ab_\Gamma)\rtimes_{\alpha^{\overleftarrow{R}}_M}\bar G_{\breve S_5,\Gamma}$,  $C_0(\Ab_\Gamma)\rtimes_{\alpha^{\overrightarrow{R}}_M}\bar G_{\breve S_6,\Gamma}$, $C_0(\Ab_\Gamma)\rtimes_{\alpha^{\overleftarrow{R},M}_{\overleftarrow{L}}}\bar G_{\breve S_7,\Gamma}$ and $C_0(\Ab_\Gamma)\rtimes_{\alpha^{\overrightarrow{R},M}_{\overrightarrow{L}}}\bar G_{\breve S_8,\Gamma}$ whenever $1\leq M\leq N$ and for a set $\{\breve S_i\}$ of suitable surface sets. 

If the tensor $C^*$-algebra $C_0(\Ab_\Gamma)\otimes C_0(\Ab_\Gp)$ is used, then the $C^*$-algebra $C_0(\Ab_\Gamma)\rtimes_{\alpha_{\overleftarrow{R}}^N} \bar G_{\breve S,\Gamma}\otimes C_0(\Ab_\Gp)\rtimes_{\alpha_{\overleftarrow{L}}^{N^\prime}} \bar G_{\breve S,\Gp}$ with respect to the minimal $C^*$-norm is constructed.

Observe that, a generalised Stone - von Neumann theorem presented in \cite[Theorem 4.24]{Williams07} is not achievable, since the objects $\bar G_{\breve S, \Gamma}$ and $\Ab_{\Gamma}$ are not identified in general. It is necessary to distinguish between the two objects, since the holonomies are independent whereas the fluxes are dependent on a surface or surface set. Nevertheless, if it is assumed that $\bar G_{\breve S, \Gamma}$ is identified with $G^M$ and $\Ab_{\Gamma}$ is identified with $G^N$, then the holonomy-flux cross-product algebra is identified with $C_0(G^N)\rtimes_{\alpha_{\overleftarrow{L}}^M}G^M$. 
But a generalised Stone - von Neumann theorem is only available for $M$ equal to $N$. This is the result of theorem \ref{Generalised Stone- von Neumann theorem} and theorem \ref{theo moritaequivgroup}. 

Hence only for $M=N$ the $C^*$-algebra $C_0(G^N)\rtimes_{\alpha_{\overleftarrow{L}}^M}G^N$ is isomorphic to $\KD(L^2(G^N,\mu_N))$. 
Notice that, the state $\omega_E^\Gamma$ is now given in this particular case by
\beqs \omega_E^\Gamma(F_\Gamma)=\int_{G^N}\int_{G^N}\dif\mu_N(\textbf{g})\dif\mu_N(\textbf{h})\vert F_\Gamma(\textbf{g},\textbf{h})\vert^2\eqs for $F_\Gamma\in C_0(G^N)\rtimes_{\alpha_{\overleftarrow{L}}^M}G^N$ and does not depend on the surfaces anymore.
If $\bar G_{\breve S,\Gamma}$ for example is identified with $G^{N-1}$, then a problem occurs. The Morita equivalent $C^*$-algebra to $C_0(G^N)\rtimes_{\alpha_{\overleftarrow{L}}^M} G^M$ where $M<N$ is not of the form $C^*(G^K)$ for a suitable $K$ where $1\leq K\leq N$. The author does not know any Morita equivalent $C^*$-algebra to the $C^*$-algebra $C_0(G^N)\rtimes_{\alpha_{\overleftarrow{L}}^M} G^M$ where $M<N$. 

Pedersen \cite[section 7.7 ]{Pedersen} has presented a generalisation of regular representations of cross-products. His results are adapted to the case of a set $\breve S$ of surfaces with simple surface intersection property for a finite orientation preserved graph system associated to $\Gamma$ in the following paragraphs.

Set $\HS_{E(\breve S)}^\Gamma:=L^2(\bar G_{\breve S,\Gamma},\HS_\Gamma)$, where this Hilbert space is identified with $L^2(\bar G_{\breve S,\Gamma})\otimes\HS_\Gamma$.
In the following investigation the elements $F_\Gamma$ are understood as elements of $\KD(\bar G_{\breve S,\Gamma},C_0(\Ab_\Gamma))$.

First observe that, $\Psi_\Gamma(\rho_{\breve S,\Gamma}(\Gamma))$ is an element of $\HS_{E(\breve S)}^\Gamma$, if there is a map
\beqs \rho_{\breve S,\Gamma}(\Gamma)\mapsto \Psi^\Gamma_{E(\breve S)}(\ho_\Gamma(\Gamma),\rho_{\breve S,\Gamma}(\Gamma))
\eqs such that $\Psi^\Gamma_{E(\breve S)}(\rho_{\breve S,\Gamma}(\Gamma))\in \HS_\Gamma$.

Then recall the $C^*$-algebra dynamical system $(\bar G_{\breve S,\Gamma},C_0(\Ab_\Gamma),\alpha_{\overleftarrow{L}}^N)$ and the covariant pair $(\Phi_M,U_{\overleftarrow{L}}^N)$ of this $C^ *$-dynamical system. There is a morphism $\Phi^M_{\overleftarrow{L}}$ of the $C^*$-algebras, which maps from $C_0(\Ab_\Gamma)$ to $\KD\big(\HS_{E(\breve S)}^\Gamma\big)$, and a representation $U_{\overleftarrow{L}}^N$ of the group $\bar G_{\breve S,\Gamma}$ on the $C^*$-algebra $\KD\big(\HS_{E(\breve S)}^\Gamma\big)$. Both objects are defined by 
\beqs \left(\Phi^M_{\overleftarrow{L}}(f_\Gamma)\Psi^\Gamma_{E(\breve S)}\right)(\rho_{\breve S,\Gamma}(\Gamma))
\:=\Phi_M\left(\alpha_{\overleftarrow{L}}^N(\rho_{\breve S,\Gamma}(\Gamma))(f_\Gamma)\right)\Psi^\Gamma_{E(\breve S)}(\rho_{\breve S,\Gamma}(\Gamma))
\eqs for $\Psi^\Gamma_{E(\breve S)}\in\HS_{E(\breve S)}^\Gamma,f_\Gamma\in C_0(\Ab_\Gamma)$ and 
\beqs (U_{\overleftarrow{L}}^N(\hat\rho_{\breve S,\Gamma}(\Gamma))\Psi^\Gamma_{E(\breve S)})(\rho_{\breve S,\Gamma}(\Gamma))
:=\Psi^\Gamma_{E(\breve S)}(L(\hat\rho_{\breve S,\Gamma}(\Gamma))(\rho_{\breve S,\Gamma}(\Gamma)))
\eqs for $U_{\overleftarrow{L}}^N\in\Rep(\bar G_{\breve S,\Gamma},\KD(\HS^\Gamma_{E(\breve S)}))$, $\rho_{\breve S,\Gamma},\hat\rho_{\breve S,\Gamma},\tilde\rho_{\breve S,\Gamma}\in G_{\breve S,\Gamma}$. Then $(\Phi_{\overleftarrow{L}}^M,U_{\overleftarrow{L}}^N)$ defines a covariant representation of $(\bar G_{\breve S,\Gamma},C_0(\Ab_\Gamma),\alpha_{\overleftarrow{L}}^N)$ in $\KD\big(\HS_{E(\breve S)}^\Gamma\big)$.

\begin{defi}Let $\breve S$ be a set of surfaces with simple surface intersection property for a finite orientation preserved graph system associated to a graph $\Gamma$.

The \textbf{left regular representation of the holonomy-flux cross-product $C^*$-algebra} \\$C_0(\Ab_\Gamma)\rtimes_{\alpha_{\overleftarrow{L}}^N}\bar G_{\breve S,\Gamma}$ induced by $(\Phi_M,\HS_\Gamma)$ is the representation $\pi^{\Gamma,\breve S}_{\overleftarrow{L}}$ on $L^2(\bar G_{\breve S,\Gamma},\HS_\Gamma))$, which is expressed by
\beqs 
&((\pi^{\Gamma,\breve S}_{\overleftarrow{L}}(F_\Gamma))\Psi_{E(\breve S)}^\Gamma)(\rho_{\breve S,\Gamma}(\Gamma))
=\left(((\Phi^M_{\overleftarrow{L}}\rtimes U_{\overleftarrow{L}}^N)(F_\Gamma))\Psi_{E(\breve S)}^\Gamma\right)(\rho_{\breve S,\Gamma}(\Gamma))\\
&:=\int_{\bar G_{\breve S,\Gamma}}\Phi_M\left(\alpha_{\overleftarrow{L}}^N(\rho_{\breve S,\Gamma}(\Gamma))(F_\Gamma(\hat\rho_{\breve S,\Gamma}(\Gamma)))\right)U_{\overleftarrow{L}}^N(\hat\rho_{\breve S,\Gamma}(\Gamma))\Psi_{E(\breve S)}^\Gamma(\rho_{\breve S,\Gamma}(\Gamma))\dif \mu_{\breve S,\Gamma}(\hat\rho_{\breve S}(\Gamma))
\eqs for $F_\Gamma(\hat\rho_{\breve S,\Gamma}(\Gamma))\in C_0(\Ab_\Gamma)$, $\rho_{\breve S,\Gamma},\tilde\rho_{\breve S,\Gamma},\hat\rho_{\breve S,\Gamma}\in G_{\breve S,\Gamma}$ and $\Psi_{E(\breve S)}^\Gamma\in\HS_{E(\breve S)}^\Gamma$. The representation $\pi^{\Gamma,\breve S}_{\overleftarrow{L}}$ is also denoted by $\Phi^M_{E(\breve S)}\rtimes U_{\overleftarrow{L}}^N$.
\end{defi}

Then recall a general $C^*$-algebra dynamical system given by $(\bar G_{\breve S,\Gamma},C_0(\Ab_\Gamma),\alpha)$. There is a morphism $\Phi^M_{E(\breve S)}$ from the $C^*$-algebra $C_0(\Ab_\Gamma)$ to $\KD\big(\HS_{E(\breve S)}^\Gamma\big)$ and a representation $U$ of the group $\bar G_{\breve S,\Gamma}$ on the $C^*$-algebra $\KD\big(\HS_{E(\breve S)}^\Gamma\big)$. They are defined by 
\beqs \left(\Phi^M_{E(\breve S)}(f_\Gamma)\Psi^\Gamma_{E(\breve S)}\right)(\rho_{\breve S,\Gamma}(\Gamma))
:=\Phi_M\left(\alpha(\rho_{\breve S,\Gamma}(\Gamma))(f_\Gamma)\right)\Psi^\Gamma_{E(\breve S)}(\rho_{\breve S,\Gamma}(\Gamma))
\eqs for $\Psi^\Gamma_{E(\breve S)}\in\HS_{E(\breve S)}^\Gamma$, $f_\Gamma\in C_0(\Ab_\Gamma)$ and $\tilde\rho_{\breve S,\Gamma},\rho_{\breve S,\Gamma}\in G_{\breve S,\Gamma}$. Consequently, a general regular representation of the holonomy-flux cross-product is given by
\beq \pi^\Gamma_{E(\breve S)}(f_\Gamma)\Psi_{E(\breve S)}^\Gamma
=(\Phi^M_{E(\breve S)}\rtimes U)(f_\Gamma)\Psi_{E(\breve S)}^\Gamma
\eq whenever $U\in \Rep(\bar G_{\breve S,\Gamma},\KD(\HS_{E(\breve S)}^\Gamma))$ and $f_\Gamma\in C_0(\Ab_\Gamma)$.

Until now a unification of the different holonomy-flux cross-product $C^*$-algebras for certain surface sets has been not presented. The following algebra plays an important role.
\begin{defi}Let $\breve S$ be a set of surfaces with simple surface intersection property for a finite orientation preserved graph system associated to a graph $\Gamma$.

The \textbf{multiplier algebra of the holonomy-flux cross-product $C^*$-algebra}\\ $C_0(\Ab_\Gamma)\rtimes_{\alpha_{\overleftarrow{L}}^N}\bar G_{\breve S,\Gamma}$ is given by all linear operators \[M:C_0(\Ab_\Gamma)\rtimes_{\alpha_{\overleftarrow{L}}^N}\bar G_{\breve S,\Gamma}\longrightarrow C_0(\Ab_\Gamma)\rtimes_{\alpha_{\overleftarrow{L}}^N}\bar G_{\breve S,\Gamma}\] such that for any $\hat F_\Gamma\in C_0(\Ab_\Gamma)\rtimes_{\alpha_{\overleftarrow{L}}^N}\bar G_{\breve S,\Gamma}$ there exists a $\tilde F_\Gamma\in C_0(\Ab_\Gamma)\rtimes_{\alpha_{\overleftarrow{L}}^N}\bar G_{\breve S,\Gamma}$ such that for all\\ $F_\Gamma\in C_0(\Ab_\Gamma)\rtimes_{\alpha_{\overleftarrow{L}}^N}\bar G_{\breve S,\Gamma}$ it is true that
\beqs \hat F^*_\Gamma M(F_\Gamma)&=
\Big \la \hat F_\Gamma, M(F_\Gamma)\Big\ra_{C_0(\Ab_\Gamma)\rtimes_{\alpha_{\overleftarrow{L}}^N}\bar G_{\breve S,\Gamma}}
= \Big \la\tilde F_\Gamma, F_\Gamma\Big\ra_{C_0(\Ab_\Gamma)\rtimes_{\alpha_{\overleftarrow{L}}^N}\bar G_{\breve S,\Gamma}}\\
&=\tilde F^*_\Gamma F_\Gamma
\eqs 

In particular, the multiplier algebra of the reduced holonomy-flux group $C^*$-algebra $\CD_r^*(\bar G_{\breve S,\Gamma},C_0(\Ab_\Gamma))$ consists of such linear maps $M$ such that for any $\hat F_\Gamma(\hat\rho_{S,\Gamma}(\Gamma))\in C_0(\Ab_\Gamma)$ there exists a $\tilde F_\Gamma(\hat\rho_{S,\Gamma}(\Gamma))\in C_0(\Ab_\Gamma)$ such that for all $F_\Gamma(\hat\rho_{S,\Gamma}(\Gamma))\in C_0(\Ab_\Gamma)$ it is true that
\beq\label{eq multiholflux} &\Big\la (\pi^{\Gamma,\breve S}_{\overleftarrow{L}}(\hat F_\Gamma(\hat\rho_{S,\Gamma}(\Gamma)))\Psi_{E(\breve S)}^\Gamma,\pi^{\Gamma,\breve S}_{\overleftarrow{L}}(M(F_\Gamma(\hat\rho_{S,\Gamma}(\Gamma))))\Phi_{E(\breve S)}^\Gamma\Big\ra_{\HS_{E(\breve S)}^\Gamma}
\\&=\Big\la \pi^{\Gamma,\breve S}_{\overleftarrow{L}}(\tilde F_\Gamma(\hat\rho_{S,\Gamma}(\Gamma)))\Psi_{E(\breve S)}^\Gamma,\pi^{\Gamma,\breve S}_{\overleftarrow{L}}(F_\Gamma(\hat\rho_{S,\Gamma}(\Gamma)))\Phi_{E(\breve S)}^\Gamma\Big\ra_{\HS_{E(\breve S)}^\Gamma}
\eq yields whenever $\Psi^\Gamma_{E(\breve S)},\Phi_{E(\breve S)}^\Gamma\in\HS_{E(\breve S)}^\Gamma$.
\end{defi}
\begin{exa}\label{exa surfsetdiff}
In \cite[Definition 3.19]{Kaminski1},\cite[Definition 6.1.19]{KaminskiPHD} the following map $I$ has been introduced. The map $I: C_0(\Ab_\Gamma)\rightarrow C_0(\Ab_{\Gamma^{-1}})$ is given by
\beqs I:f_\Gamma \mapsto f_{\Gamma^{-1}}, \text{ where } (I\circ f_\Gamma)(\ho_{\Gamma}(\gamma_1),...,\ho_{\Gamma}(\gamma_N)):= f_{\Gamma^{-1}}(\ho_{\Gamma^{-1}}(\gamma_1)^{-1},...,\ho_{\Gamma^{-1}}(\gamma_N)^{-1})
\eqs such that $I^2=\text{id}$, where $\text{id}$ is the identical automorphism on $C_0(\Ab_\Gamma)$. 

Consider a suitable set $\bar S$ of surfaces that is contained in the set $\breve S$ and let $M\leq N$. Note that, if $M<N$, then there is a set of paths $\Gpp:=\Gamma\setminus\Gp$ such that each path of this set does not intersect a surface in $\bar S$. Each path in $\Gp$ intersects only one surface in $\bar S$ at the source vertes of this path. Then $\bar G_{\bar S,\Gp\leq \Gamma}$ is a subgroup of $\bar G_{\bar S,\Gamma}$ and is embedded by $\bar G_{\bar S,\Gamma}:=\bar G_{\bar S,\Gp} \times\{e_G\}\times....\times\{e_G\}$ in $\bar G_{\breve S,\Gamma}$. Denote the set of surfaces, which has the simple surface intersection property for the finite orientation preserved graph system $\PD^{\op}_{\Gp}$, which is contained in $\breve S$ and which is not contained in $\bar S$, by $\breve R$. Note that $\bar G_{\breve R,\Gpp\leq \Gamma}$ is a subgroup of $\bar G_{\breve R,\Gamma}$ and is embedded by $\bar G_{\breve R,\Gamma}:=\bar G_{\breve R,\Gpp} \times\{e_G\}\times....\times\{e_G\}$ in $\bar G_{\breve S,\Gamma}$. Let $\bar R$ be a set of surfaces, which has the same surface intersection property for a path $\gp$ in a graph, which is contained in the finite orientation preserved graph system $\PD_{\Gp}^{\op}$.

\textit{Situation 1}:\\
Then there is a $C^*$-dynamical system in $\KD(\HS^{\Gamma^{-1}}_{\E(\bar S)})$, which is given by $(\bar G_{\bar S,\Gamma^{-1}},C_0(\Ab_{\Gamma^{-1}}),\alpha_{\overleftarrow{R}}^M)$. Let $(\Phi_M,U_{\overleftarrow{R}}^M)$ be a covariant pair associated to the $C^*$-dynamical system.

Then observe that $\alpha_{\overleftarrow{R}}^M=I\circ\alpha_{\overleftarrow{L}}^M\circ I ^{-1}$ and $U_{\overleftarrow{R}}^M=I\circ U_{\overleftarrow{L}}^M\circ I ^{-1}$ hold. Then $(\bar G_{\bar S,\Gamma },C_0(\Ab_{\Gamma^{-1}}),I\circ\alpha_{\overleftarrow{L}}^M\circ I^{-1})$ is a $C^*$-dynamical system in $\KD(\HS^{\Gamma^{-1}}_{\E(\bar S)})$. Respectively, $(\bar G_{\bar S,\Gamma },C_0(\Ab_{\Gamma}),\alpha_{\overleftarrow{L}}^M)$ is a $C^*$-dynamical system in $\KD(\HS^{\Gamma}_{\E(\bar S)})$.

Note that, if $\bar S$ is equal to $\breve S$, then $\bar S$ has the simple surface intersection property for the finite orientation preserved graph system $\PD^{\op}_{\Gamma^{-1}}$ and $M=N$.Then $(\bar G_{\breve S,\Gamma },C_0(\Ab_{\Gamma^{-1}}),I\circ\alpha_{\overleftarrow{L}}^N\circ I^{-1})$ and $(\bar G_{\breve S,\Gamma },C_0(\Ab_{\Gamma}),\alpha_{\overleftarrow{L}}^N)$ are two $C^*$-dynamical systems in $\KD(\HS^{\Gamma}_{\E(\bar S)})$.

\textit{Situation 2}:\\
Furthermore, there is a $C^*$-dynamical system in $\KD(\HS^{\Gamma}_{\E(\breve R)})$ given by $(\bar G_{\breve R,\Gamma},C_0(\Ab_{\Gamma}),\alpha_{\overleftarrow{L}}^K)$ for $K$ suitable. 

\textit{Situation 3}:\\
There is a $C^*$-dynamical system in $\KD(\HS^{\gp}_{\E(\bar R)})$ given by $(\bar G_{\bar R,\gp},C_0(\Ab_{\gp}),\alpha_{\overleftarrow{L}}^1)$.

\textit{Situation 4}:\\
Finally, there is $C^*$-dynamical system in $\KD(\HS^{\Gamma}_{\E(\breve S)})$ given by $(\bar G_{\bar S,\Gamma^{\prime\text{ }-1}}\times\bar G_{\breve R,\Gpp},C_0(\Ab_{\Gamma}),(I^{-1}\circ\alpha_{\overleftarrow{R}}^M\circ I)\circ\alpha_{\overleftarrow{L}}^K)$. Note that, $(I^{-1}\circ\alpha_{\overleftarrow{R}}^M\circ I)\circ\alpha_{\overleftarrow{L}}^K=\alpha_{\overleftarrow{L}}^K\circ(I^{-1}\circ\alpha_{\overleftarrow{R}}^M\circ I)$. Reformulate 
$(\bar G_{\bar S,\Gamma^{\prime} }\times\bar G_{\breve R,\Gpp},C_0(\Ab_{\Gamma}),\alpha_{\overleftarrow{L}}^M\circ\alpha_{\overleftarrow{L}}^K)$.

For each $C^*$-dynamical system given above there is a cross-product $C^*$-algebra.
\end{exa}

In the following proposition the \textit{situation 1} is studied. 
\begin{prop}\label{prop proofmultiplier}Let $\breve T:=\{T_1,...,T_N\}$ be a set of surfaces with simple surface intersection property for the orientation preserved graph system $\PD^{\op}_{\Gamma}$. Let $\breve S:=\{S_1,...,S_M\}$ be a set of surfaces that is contained in $\breve T$ an such that $M\leq N$.
 
The unitaries $U_{\overleftarrow{R}}^M(\rho_{\breve S,\Gamma^{-1}}(\Gamma^{-1}))$, whenever $\rho_{\breve S,\Gamma^{-1}}(\Gamma^{-1})\in\bar G_{\breve S,\Gamma^{-1}}$, are elements of the multiplier algebra of the $C^*$-algebra $C_0(\Ab_\Gamma)\rtimes_{\alpha_{\overleftarrow{L}}^N}\bar G_{\breve T,\Gamma}$. Moreover, the elements of the holonomy-flux cross-product algebra $C_0(\Ab_\Gamma)\rtimes_{\alpha_{\overleftarrow{R}}^M\circ I}\bar G_{\breve S,\Gamma^{-1}}$ are multipliers of the $C^*$-algebra $C_0(\Ab_\Gamma)\rtimes_{\alpha_{\overleftarrow{L}}^N}\bar G_{\breve S,\Gamma}$. 
\end{prop}
\begin{proof}
Choose the two surface sets $\breve S$ and $\breve T$ and a graph $\Gamma$ such that $(C_0(\Ab_\Gamma),\bar G_{\breve S,\Gamma^{-1}},I^{-1}\circ\alpha_M^{\overleftarrow{R}}\circ I)$ and\\ $(C_0(\Ab_\Gamma),\bar G_{\breve T,\Gamma^{-1}},I^{-1}\circ\alpha_N^{\overleftarrow{R}}\circ I)$ are two $C^*$-dynamical systems.
Then notice that
\beqs &(F_\Gamma\ast \hat F_\Gamma) (\tilde\rho_{T_1}(\gamma_1^{-1}),...,\tilde\rho_{T_N}(\gamma_N^{-1}))\\
&= \int_{\bar G_{\breve S,\Gamma^{-1}}}\dif\mu_{\breve S,\Gamma^{-1}}(\rho_{\breve S,\Gamma^{-1}}(\Gamma^{-1})) \\&\qquad \quad F_\Gamma(\rho_{\breve S,\Gamma^{-1}}(\Gamma^{-1}))\left((I^{-1}\circ\alpha^{\overleftarrow{R}}_M(\rho^M_{\breve S,\Gamma^{-1}})\circ I)(\hat F_\Gamma)\right)\big(\rho_{S_1}(\gamma_1^{-1})^{-1}\tilde\rho_{T_1}(\gamma_1^{-1}),....,\rho_{S_N}(\gamma_N^{-1})^{-1}\tilde\rho_{T_N}(\gamma_N^{-1})\big)
\eqs holds whenever $F_\Gamma\in L^1(\bar G_{\breve S,\Gamma^{-1}},C_0(\Ab_\Gamma),I^{-1}\circ\alpha_M^{\overleftarrow{R}}\circ I)$ and $\hat F_\Gamma\in L^1(\bar G_{\breve S,\Gamma^{-1}},C_0(\Ab_\Gamma),I^{-1}\circ\alpha_N^{\overleftarrow{R}}\circ I)$. Furthermore recognize that,
\beqs 
&F^*_\Gamma(\tilde\rho_{\breve S,\Gamma^{-1}}(\Gamma^{-1}))=(I^{-1}\circ\alpha^{\overleftarrow{R}}_M(\tilde\rho_{\breve S,\Gamma^{-1}}^M)\circ I)\left(F_\Gamma^+(\tilde\rho_{\breve S,\Gamma^{-1}}(\Gamma^{-1})^{-1})\right)
\eqs is true.

Notice that,
\beqs \alpha^{\overleftarrow{R}}_N(\tilde\rho_{\breve S,\Gamma^{-1}}^N)(f_{\Gamma^{-1}})(\ho_{\Gamma^{-1}}( \Gamma^{-1}))
&= (I\circ \alpha_{\overleftarrow{L}}^N(\tilde\rho_{\breve S,\Gamma}^N)\circ I^{-1})(f_{\Gamma^{-1}})(\ho_{\Gamma^{-1}}( \Gamma^{-1}))\\
&= (I\circ f_{\Gamma})(\tilde\rho_{\breve S,\Gamma}(\Gamma)^{-1}\ho_{\Gamma}( \Gamma))\\
&=f_{\Gamma^{-1}}(\ho_{\Gamma^{-1}}( \Gamma^{-1})\tilde\rho_{\breve S,\Gamma^{-1}}(\Gamma^{-1}))
\eqs and
\beqs \int_{\Ab_{\Gamma}}\dif\mu_{\Gamma}(\ho_\Gamma( \Gamma)) (I^{-1}\circ\alpha^{\overleftarrow{R}}_N(\tilde\rho_{\breve S,\Gamma^{-1}}^N)\circ I)(f_\Gamma)(\ho_\Gamma( \Gamma))
= \int_{\Ab_{\Gamma}}\dif\mu_{\Gamma}(\ho_\Gamma( \Gamma)) \alpha^{\overleftarrow{L}}_N(\tilde\rho_{\breve S,\Gamma}^N)(f_\Gamma)(\ho_\Gamma( \Gamma))
\eqs yields whenever $f_\Gamma\in C_0(\Ab_\Gamma)$ and $\tilde\rho_{\breve S,\Gamma}^N\in\bar G_{\breve T,\Gamma}$.

Clearly, there is a representation $\pi^M_{I,\overleftarrow{R}}$ of $L^1(\bar G_{\breve S,\Gamma^{-1}},C_0(\Ab_\Gamma),I^{-1}\circ\alpha_M^{\overleftarrow{R}}\circ I)$ on $\HS_\Gamma$, which is given by
\beqs \pi_{I,\overleftarrow{R}}^M(F_\Gamma)\psi_\Gamma
&:=\int_{\bar G_{\breve S,\Gamma^{-1}}}\dif\mu_{\breve S,\Gamma^{-1}}(\rho_{\breve S,\Gamma^{-1}}^M)\Phi_M(F_\Gamma(\rho_{\breve S,\Gamma^{-1}}^M))(I^{-1}\circ U_{\overleftarrow{R}}^M(\rho_{\breve S,\Gamma^{-1}}^M)\circ I)\psi_\Gamma\\
&=\int_{\bar G_{\breve S,\Gamma}}\dif\mu_{\breve S,\Gamma}(\rho_{\breve S,\Gamma}^M)\Phi_M(F_\Gamma(\rho_{\breve S,\Gamma}^M))\big(U_{\overleftarrow{L}}^M(\rho_{\breve S,\Gamma}^M)\big)\psi_\Gamma
\eqs where $\rho_{\breve S,\Gamma^{-1}}^M\in \bar G_{\breve S,\Gamma^{-1}}$, $F_\Gamma\in L^1(\bar G_{S,\Gamma^{-1}},C_0(\Ab_\Gamma),I^{-1}\circ\alpha_M^{\overleftarrow{R}}\circ I)$ and $\psi_\Gamma\in \HS_\Gamma$. Then derive that there is an isomorphism $\mathcal{I}$ from $ L^1(\bar G_{\breve S,\Gamma^{-1}},C_0(\Ab_\Gamma),I^{-1}\circ\alpha_M^{\overleftarrow{R}}\circ I)$ to $ L^1(\bar G_{\breve S,\Gamma},C_0(\Ab_{\Gamma}),\alpha^M_{\overleftarrow{L}})$. 

Then the Hilbert space $L^2(\bar G_{\breve S,\Gamma},\mu_{\breve S,\Gamma})$ is embedded into $L^2(\bar G_{\breve T,\Gamma},\mu_{\breve T,\Gamma})$. The left regular representation of $C_0(\Ab_\Gamma)\rtimes_{\alpha_M^{\overleftarrow{R}}\circ I}\bar G_{\breve S,\Gamma^{-1}}$ on $L^2(\bar G_{\breve T,\Gamma},\mu_{\breve T,\Gamma})\otimes\HS_\Gamma$ is given by
\beqs 
&((\pi^{\Gamma^{-1},\breve S}_{\overleftarrow{R}}(F_\Gamma))\Psi_{E(\breve T)}^\Gamma)(\rho_{\breve T,\Gamma}(\Gamma))
=\left(((\Phi^M_{\overleftarrow{R}}\rtimes(I^{-1}\circ U_{\overleftarrow{R}}^M)\circ I)(F_\Gamma))\Psi_{E(\breve T)}^\Gamma\right)(\rho_{\breve T,\Gamma}(\Gamma))\\
&:=\int_{\bar G_{\breve S,\Gamma^{-1}}}\dif \mu_{\breve S,\Gamma^{-1}}(\hat\rho_{\breve S}(\Gamma^{-1}))\\
&\qquad\qquad\Phi_M\left((I^{-1}\circ \alpha_{\overleftarrow{R}}^M(\rho_{\breve S,\Gamma^{-1}}(\Gamma^{-1}))\circ I)(F_\Gamma(\hat\rho_{\breve S,\Gamma^{-1}}(\Gamma^{-1})))\right)(I^{-1}\circ U_{\overleftarrow{R}}^M(\hat\rho_{\breve S,\Gamma^{-1}}(\Gamma^{-1})) \circ I)\Psi_{E(\breve T)}^\Gamma(\rho_{\breve T,\Gamma}(\Gamma))
\eqs for $F_\Gamma(\hat\rho_{\breve S,\Gamma^{-1}}(\Gamma^{-1}))\in C_0(\Ab_\Gamma)$, $\rho_{\breve T,\Gamma}\in G_{\breve T,\Gamma}$, $\tilde\rho_{\breve S,\Gamma^{-1}},\hat\rho_{\breve S,\Gamma^{-1}}\in G_{\breve S,\Gamma^{-1}}$ and $\Psi_{E(\breve S)}^\Gamma\in\HS_{E(\breve S)}^\Gamma$.

Set $U_{\overleftarrow{L}}^N(\hat\rho_{\breve T,\Gamma}(\Gamma))\Psi_{E(\breve T)}^\Gamma(\rho_{\breve T,\Gamma}(\Gamma)):= \hat\Psi_{E(\breve T)}^\Gamma(\rho_{\breve T,\Gamma}(\Gamma))$ and $(I^{-1}\circ U_{\overleftarrow{R}}^N(\tilde\rho_{\breve S,\Gamma^{-1}}(\Gamma^{-1}))\circ I)\hat\Psi_{E(\breve T)}^\Gamma(\rho_{\breve T,\Gamma}(\Gamma)):= \tilde\Psi_{E(\breve T)}^\Gamma(\rho_{\breve T,\Gamma}(\Gamma))$.
Then the unitaries $I^{-1}\circ U_{\overleftarrow{R}}^N(\tilde\rho_{\breve S,\Gamma^{-1}}^N)\circ I$, whenever $\tilde\rho_{\breve S,\Gamma^{-1}}^N\in\bar G_{\breve S,\Gamma^{-1}}$, are multipliers. This is verified by the following computation: 
\beqs &\Big\la (\pi^{\Gamma,\breve T}_{\overleftarrow{L}}(\hat F_\Gamma))\Psi_{E(\breve T)}^\Gamma)(\rho_{\breve T,\Gamma}(\Gamma)), ((\pi^{\Gamma,\breve T}_{\overleftarrow{L}}(M_U(\hat F_\Gamma)))\Psi_{E(\breve T)}^\Gamma)(\rho_{\breve T,\Gamma}(\Gamma))\Big\ra_{\HS_{E(\breve T)}^\Gamma}\\
&:=\Big\la (\pi^{\Gamma,\breve T}_{\overleftarrow{L}}(\hat F_\Gamma)\Psi_{E(\breve T)}^\Gamma)(\rho_{\breve T,\Gamma}(\Gamma)), ((\pi^{\Gamma,\breve T}_{\overleftarrow{L}}((I^{-1}\circ U_{\overleftarrow{R}}^N(\tilde\rho_{\breve S,\Gamma^{-1}}^N)\circ I)(\hat F_\Gamma)))\Psi_{E(\breve T)}^\Gamma)(\rho_{\breve T,\Gamma}(\Gamma))\Big\ra_{\HS_{E(\breve T)}^\Gamma}\\
&=\Big\la \left(((\Phi^M_{\overleftarrow{L}}\rtimes U_{\overleftarrow{L}}^N)(\hat F_\Gamma))\Psi_{E(\breve T)}^\Gamma\right)(\rho_{\breve T,\Gamma}(\Gamma)), \left(((\Phi^M_{\overleftarrow{L}}\rtimes U_{\overleftarrow{L}}^N)((I^{-1}\circ U_{\overleftarrow{R}}^N(\tilde\rho_{\breve S,\Gamma^{-1}}^N)\circ I)(\hat F_\Gamma)))\Psi_{E(\breve T)}^\Gamma\right)(\rho_{\breve T,\Gamma}(\Gamma))\Big\ra_{\HS_{E(\breve T)}^\Gamma}\\
&=\int_{\bar G_{\breve S,\Gamma^{-1}}}\int_{\bar G_{\breve T,\Gamma}}\dif\mu_{\breve S,\Gamma^{-1}}(\tilde\rho_{\breve S,\Gamma^{-1}}(\Gamma^{-1})) \dif \mu_{\breve T,\Gamma}(\hat\rho_{\breve T}(\Gamma))\\
&\qquad\qquad\Big\la \Phi_M\left(\big(\alpha_{\overleftarrow{L}}^N(\rho_{\breve T,\Gamma}(\Gamma))(\hat F_\Gamma)\big)(\hat\rho_{\breve T,\Gamma}(\Gamma))\right) \hat\Psi_{E(\breve T)}^\Gamma(\rho_{\breve T,\Gamma}(\Gamma)), \\
&\qquad\qquad\qquad \Phi_M\left(\big((\alpha_{\overleftarrow{L}}^N(\rho_{\breve T,\Gamma}(\Gamma))\circ I^{-1}\circ \alpha_{\overleftarrow{R}}^N(\tilde\rho_{\breve S,\Gamma^{-1}}(\Gamma^{-1}))\circ I) \hat F_\Gamma)\big)(\hat\rho_{\breve T,\Gamma}(\Gamma))\right)\\
&\qquad\qquad\qquad\qquad (I^{-1}\circ U_{\overleftarrow{R}}^N(\tilde\rho_{\breve S,\Gamma^{-1}}^N)\circ I)\hat\Psi_{E(\breve T)}^\Gamma(\rho_{\breve T,\Gamma}(\Gamma))\Big\ra_{\HS_{E(\breve T)}^\Gamma}\\
&=\int_{\bar G_{\breve S,\Gamma^{-1}}}\int_{\bar G_{\breve T,\Gamma}}\dif\mu_{\breve S,\Gamma^{-1}}(\tilde\rho_{\breve S,\Gamma^{-1}}(\Gamma^{-1})) \dif \mu_{\breve T,\Gamma}(\hat\rho_{\breve T}(\Gamma))\\
&\qquad\qquad\Big\la \Phi_M\left(\big(\alpha_{\overleftarrow{L}}^N(\rho_{\breve T,\Gamma}(\Gamma))(\hat F_\Gamma)\big)(\hat\rho_{\breve T,\Gamma}(\Gamma))\right)(I^{-1}\circ U_{\overleftarrow{R}}^N(\tilde\rho_{\breve S,\Gamma^{-1}}^N)^*\circ I)\tilde\Psi_{E(\breve T)}^\Gamma(\rho_{\breve T,\Gamma}(\Gamma)), \\
&\qquad\qquad \Phi_M\left(\big((I^{-1}\circ \alpha_{\overleftarrow{R}}^N(\tilde\rho_{\breve S,\Gamma^{-1}}(\Gamma^{-1}))\circ I\circ \alpha_{\overleftarrow{L}}^N(\rho_{\breve T,\Gamma}(\Gamma))) \hat F_\Gamma)\big)(\hat\rho_{\breve T,\Gamma}(\Gamma))\right) \tilde\Psi_{E(\breve T)}^\Gamma(\rho_{\breve T,\Gamma}(\Gamma))\Big\ra_{\HS_{E(\breve T)}^\Gamma}\\
&=\Big\la \left(((\Phi^M_{\overleftarrow{L}}\rtimes U_{\overleftarrow{L}}^N)((I^{-1}\circ U_{\overleftarrow{R}}^N(\tilde\rho_{\breve S,\Gamma^{-1}}^N)^*\circ I)\hat F_\Gamma))\Psi_{E(\breve T)}^\Gamma\right)(\rho_{\breve T,\Gamma}(\Gamma)), \left(((\Phi^M_{\overleftarrow{L}}\rtimes U_{\overleftarrow{L}}^N)(\hat F_\Gamma))\Psi_{E(\breve T)}^\Gamma\right)(\rho_{\breve T,\Gamma}(\Gamma))\Big\ra_{\HS_{E(\breve T)}^\Gamma}\\
&=\Big\la (\pi^{\Gamma,\breve T}_{\overleftarrow{L}}(\tilde F_\Gamma))\Psi_{E(\breve T)}^\Gamma)(\rho_{\breve T,\Gamma}(\Gamma)), ((\pi^{\Gamma,\breve T}_{\overleftarrow{L}}(\hat F_\Gamma))\Psi_{E(\breve T)}^\Gamma)(\rho_{\breve T,\Gamma}(\Gamma))\Big\ra_{\HS_{E(\breve T)}^\Gamma}
\eqs holds whenever $(I^{-1}\circ U_{\overleftarrow{R}}^N(\tilde\rho_{\breve S,\Gamma^{-1}}^N)^*\circ I)\hat F_\Gamma:=\tilde F_\Gamma$.

Finally each element of the $C^*$-algebra $C_0(\Ab_\Gamma)\rtimes_{I^{-1}\circ \alpha^{\overleftarrow{R}}_M\circ I}\bar G_{\breve S,\Gamma^{-1}}$ defines a linear map $M$ from  $C_0(\Ab_\Gamma)\rtimes_{\alpha_{\overleftarrow{L}}^N}\bar G_{\breve T,\Gamma}$ to $C_0(\Ab_\Gamma)\rtimes_{\alpha_{\overleftarrow{L}}^N}\bar G_{\breve T,\Gamma}$ by
\beqs &((\pi^{\Gamma,\breve T}_{\overleftarrow{L}}(M(\hat F_\Gamma)))\Psi_{E(\breve T)}^\Gamma)(\rho_{\breve T,\Gamma}(\Gamma))\\
&:=((\pi^{\Gamma,\breve T}_{\overleftarrow{L}}(F_\Gamma\ast\hat F_\Gamma))\Psi_{E(\breve T)}^\Gamma)(\rho_{\breve T,\Gamma}(\Gamma))
=\left(((\Phi^M_{\overleftarrow{L}}\rtimes U_{\overleftarrow{L}}^N)(F_\Gamma\ast\hat F_\Gamma))\Psi_{E(\breve T)}^\Gamma\right)(\rho_{\breve T,\Gamma}(\Gamma))\\
&=\int_{\bar G_{\breve S,\Gamma^{-1}}}\int_{\bar G_{\breve T,\Gamma}}\dif\mu_{\breve S,\Gamma^{-1}}(\tilde\rho_{\breve S,\Gamma}(\Gamma)) \dif \mu_{\breve T,\Gamma}(\hat\rho_{\breve T}(\Gamma))\\
&\qquad\qquad\Phi_M\left(F_\Gamma(\tilde\rho_{\breve S,\Gamma^{-1}}(\Gamma^{-1}))((I^{-1}\circ\alpha^{\overleftarrow{R}}_M(\tilde\rho_{\breve S,\Gamma^{-1}}^M)\circ I\circ\alpha_{\overleftarrow{L}}^N(\rho_{\breve T,\Gamma}(\Gamma)))(\hat F_\Gamma))(\tilde\rho_{\breve S,\Gamma}(\Gamma)^{-1}\hat\rho_{\breve T,\Gamma}(\Gamma))\right)\\
&\qquad\qquad\qquad U_{\overleftarrow{L}}^N(\hat\rho_{\breve T,\Gamma}(\Gamma))\Psi_{E(\breve T)}^\Gamma(\rho_{\breve T,\Gamma}(\Gamma))\\
&=\int_{\bar G_{\breve S,\Gamma}}\int_{\bar G_{\breve T,\Gamma}}\dif\mu_{\breve S,\Gamma}(\tilde\rho_{\breve S,\Gamma}(\Gamma)) \dif \mu_{\breve T,\Gamma}(\hat\rho_{\breve T}(\Gamma))\\
&\qquad\qquad\Phi_M\left(F_\Gamma(\tilde\rho_{\breve S,\Gamma}(\Gamma))((\alpha_{\overleftarrow{L}}^N(\tilde\rho_{\breve S,\Gamma}(\Gamma)^{-1})\circ\alpha_{\overleftarrow{L}}^N(\rho_{\breve T,\Gamma}(\Gamma)))(\hat F_\Gamma))(\tilde\rho_{\breve S,\Gamma}(\Gamma)^{-1}\hat\rho_{\breve T,\Gamma}(\Gamma))\right)\\
&\qquad\qquad\qquad U_{\overleftarrow{L}}^N(\hat\rho_{\breve T,\Gamma}(\Gamma))\Psi_{E(\breve T)}^\Gamma(\rho_{\breve T,\Gamma}(\Gamma))\\
\eqs for $F_\Gamma(\tilde\rho_{\breve S,\Gamma}(\Gamma)), \hat F_\Gamma(\hat\rho_{\breve T,\Gamma}(\Gamma))\in C_0(\Ab_\Gamma)$, $\rho_{\breve S,\Gamma}\in G_{\breve S,\Gamma}$, $\tilde\rho_{\breve T,\Gamma},\hat\rho_{\breve T,\Gamma}\in G_{\breve T,\Gamma}$, $\tilde\rho_{\breve T,\Gamma}(\Gamma):=(\tilde\rho_{\breve T,\Gamma}^M, e_G,...,e_G)\in \bar G_{\breve T,\Gamma}$ and $\Psi_{E(\breve T)}^\Gamma\in\HS_{E(\breve T)}^\Gamma$. Clearly, if the set $\breve S$ is replaced by a set $\breve R^{-1}$, which is contained in $\breve T$, then $\tilde\rho_{\breve R^{-1},\Gamma}(\Gamma)^{-1}=\tilde\rho_{\breve R,\Gamma}(\Gamma)\in\bar G_{\breve R,\Gamma}$ and $\alpha_{\overleftarrow{L}}^N(\tilde\rho_{\breve R^{-1},\Gamma}(\Gamma)^{-1})=\alpha_{\overleftarrow{L}}^N(\tilde\rho_{\breve R,\Gamma}(\Gamma))\in\Aut(C_0(\Ab_\Gamma))$ yield.

Set $(I^{-1}\circ U_{\overleftarrow{L}}^N(\hat\rho_{\breve S,\Gamma^{-1}}(\Gamma^{-1}))\circ I)\Psi_{E(\breve T)}^\Gamma(\rho_{\breve T,\Gamma}(\Gamma)):=\hat\Psi_{E(\breve T)}^\Gamma(\rho_{\breve T,\Gamma}(\Gamma))$. Then $M$ is a multiplier since the following derivation is true:
\beqs &\Big\la ((\pi^{\Gamma,\breve T}_{\overleftarrow{L}}(\hat F_\Gamma))\Psi_{E(\breve T)}^\Gamma)(\rho_{\breve T,\Gamma}(\Gamma)),((\pi^{\Gamma,\breve T}_{\overleftarrow{L}}(M(\hat F_\Gamma)))\Psi_{E(\breve T)}^\Gamma)(\rho_{\breve T,\Gamma}(\Gamma))\Big\ra_{\HS_{E(\breve T)}^{\Gamma}}\\
&= \int_{\bar G_{\breve S,\Gamma^{-1}}}\int_{\bar G_{\breve T,\Gamma}}\dif\mu_{\breve S,\Gamma^{-1}}(\tilde\rho_{\breve S,\Gamma^{-1}}(\Gamma^{-1})) \dif \mu_{\breve T,\Gamma}(\hat\rho_{\breve T}(\Gamma))\\
&\qquad \la \Phi_M\big(\alpha_{\overleftarrow{L}}^N(\rho_{\breve T,\Gamma}(\Gamma))(\hat F_\Gamma)\big)(\hat\rho_{\breve T,\Gamma}(\Gamma))\Big) U_{\overleftarrow{L}}^N(\hat\rho_{\breve T,\Gamma}(\Gamma))\Psi_{E(\breve T)}^\Gamma(\rho_{\breve T,\Gamma}(\Gamma)),\\
&\qquad\qquad\Phi_M\left(F_\Gamma(\tilde\rho_{\breve S,\Gamma^{-1}}(\Gamma^{-1}))((I^{-1}\circ\alpha^{\overleftarrow{R}}_M(\tilde\rho_{\breve S,\Gamma^{-1}}^M)\circ I\circ\alpha_{\overleftarrow{L}}^N(\rho_{\breve T,\Gamma}(\Gamma)))(\hat F_\Gamma))(\tilde\rho_{\breve S,\Gamma}(\Gamma)^{-1}\hat\rho_{\breve T,\Gamma}(\Gamma))\right)\\
&\qquad\qquad\qquad U_{\overleftarrow{L}}^N(\hat\rho_{\breve T,\Gamma}(\Gamma))\Psi_{E(\breve T)}^\Gamma(\rho_{\breve T,\Gamma}(\Gamma))\ra_{\HS_{E(\breve T)}^{\Gamma}}\\
&= \int_{\bar G_{\breve S,\Gamma^{-1}}}\int_{\bar G_{\breve T,\Gamma}}\dif\mu_{\breve S,\Gamma^{-1}}(\tilde\rho_{\breve S,\Gamma^{-1}}(\Gamma^{-1})) \dif \mu_{\breve T,\Gamma}(\hat\rho_{\breve T}(\Gamma))\\
&\qquad \Big\la \Phi_M\big(\alpha_{\overleftarrow{L}}^N(\rho_{\breve T,\Gamma}(\Gamma))(\hat F_\Gamma)\big)(\hat\rho_{\breve T,\Gamma}(\Gamma))\Big) \hat\Psi_{E(\breve T)}^\Gamma(\rho_{\breve T,\Gamma}(\Gamma)),\\
&\qquad\qquad\Phi_M\Big(F_\Gamma(\tilde\rho_{\breve S,\Gamma^{-1}}(\Gamma^{-1}))((I^{-1}\circ\alpha^{\overleftarrow{R}}_M(\tilde\rho_{\breve S,\Gamma^{-1}}^M)\circ I\circ\alpha_{\overleftarrow{L}}^N(\rho_{\breve T,\Gamma}(\Gamma)))(\hat F_\Gamma))(\tilde\rho_{\breve S,\Gamma}(\Gamma)^{-1}\hat\rho_{\breve T,\Gamma}(\Gamma))\Big)\\
&\qquad\qquad\qquad\hat\Psi_{E(\breve T)}^\Gamma(\rho_{\breve T,\Gamma}(\Gamma))\Big\ra_{\HS_{E(\breve T)}^{\Gamma}}\\
&= \int_{\bar G_{\breve S,\Gamma}}\int_{\bar G_{\breve T,\Gamma}}\dif\mu_{\breve S,\Gamma}(\tilde\rho_{\breve S,\Gamma}(\Gamma)) \dif \mu_{\breve T,\Gamma}(\hat\rho_{\breve T}(\Gamma))\\
&\qquad \Big\la \Phi_M\Big(\big(\alpha_{\overleftarrow{L}}^N(\rho_{\breve T,\Gamma}(\Gamma))(\hat F_\Gamma)\big)(\hat\rho_{\breve T,\Gamma}(\Gamma))\Big) U_{\overleftarrow{L}}^N(\hat\rho_{\breve T,\Gamma}(\Gamma))\Psi_{E(\breve T)}^\Gamma(\rho_{\breve T,\Gamma}(\Gamma)),\\
&\qquad\qquad\Phi_M\left(F_\Gamma(\tilde\rho_{\breve S,\Gamma}(\Gamma))((\alpha^{\overleftarrow{R}}_M(\tilde\rho_{\breve S,\Gamma}^M)\circ\alpha_{\overleftarrow{L}}^N(\rho_{\breve T,\Gamma}(\Gamma)))(\hat F_\Gamma))(\tilde\rho_{\breve S,\Gamma}(\Gamma)^{-1}\hat\rho_{\breve T,\Gamma}(\Gamma))\right)U_{\overleftarrow{L}}^N(\hat\rho_{\breve T,\Gamma}(\Gamma))\Psi_{E(\breve T)}^\Gamma(\rho_{\breve T,\Gamma}(\Gamma))\Big\ra_{\HS_{E(\breve T)}^{\Gamma}}\\
&= \int_{\bar G_{\breve S,\Gamma^{-1}}}\int_{\bar G_{\breve T,\Gamma}}\dif\mu_{\breve S,\Gamma^{-1}}(\tilde\rho_{\breve S,\Gamma^{-1}}(\Gamma^{-1})) \dif \mu_{\breve T,\Gamma}(\hat\rho_{\breve T}(\Gamma))\\
&\qquad \Big\la \Phi_M\Big((I^{-1}\circ \alpha^{\overleftarrow{R}}_M(\tilde\rho_{\breve S,\Gamma^{-1}}^M)\circ I)\left(F_\Gamma^+(\tilde\rho_{\breve S,\Gamma^{-1}}(\Gamma^{-1})^{-1})\right)\\
&\qquad\qquad\big((I^{-1}\circ\alpha^{\overleftarrow{R}}_M(\tilde\rho_{\breve S,\Gamma^{-1}}^M)\circ I\circ\alpha_{\overleftarrow{L}}^N(\rho_{\breve T,\Gamma}(\Gamma)))(\hat F_\Gamma)\big)(\tilde\rho_{\breve S,\Gamma}(\Gamma)\hat\rho_{\breve T,\Gamma}(\Gamma))\Big) \hat\Psi_{E(\breve T)}^\Gamma(\rho_{\breve T,\Gamma}(\Gamma)),\\
&\qquad\qquad\qquad\Phi_M\left(\alpha_{\overleftarrow{L}}^N(\rho_{\breve T,\Gamma}(\Gamma))(\hat F_\Gamma)(\hat\rho_{\breve T,\Gamma}(\Gamma))\right)\hat\Psi_{E(\breve T)}^\Gamma(\rho_{\breve T,\Gamma}(\Gamma))\Big\ra_{\HS_{E(\breve T)}^{\Gamma}}\\
&=\la ((\pi^{\Gamma,\breve T}_{\overleftarrow{L}}(F_\Gamma^*\ast\hat F_\Gamma))\Psi_{E(\breve T)}^\Gamma)(\rho_{\breve T,\Gamma}(\Gamma)), \pi^{\Gamma,\breve T}_{\overleftarrow{L}}(\hat F_\Gamma)\Psi_{E(\breve T)}^\Gamma)(\rho_{\breve T,\Gamma}(\Gamma))\ra_{\HS_{E(\breve T)}^{\Gamma}}
\eqs 
\end{proof}

Notice that, the same arguments are used for a surface set $\breve T:=\{T_1,...,T_N\}$, which has the simple surface intersection property for the orientation preserved graph system $\PD^{\op}_{\Gamma}$ and where $\breve R^{-1}:=\{R_1^{-1},...,R^{-1}_N\}$ is a set of surfaces that has the simple surface intersection property for the orientation preserved graph system $\PD^{\op}_{\Gamma^{-1}}$. Indeed it can be shown that for all situations of example \thesection.\ref{exa surfsetdiff} except \textit{situation 2} similar results can be obtained. The \textit{situation 2} is not needed in the next theorem and hence is briefly discussed in the following remark.

\begin{rem}\label{rem situation2and4}
In \textit{situation 2} the sets $\breve R$ and $\bar S$ are disjoint. Let $\breve T_2$ and $\breve T_3$ be two disjoint surface sets such that the holonomy-flux cross-product algebras are given by $C_0(\Ab_\Gamma)\rtimes_{\alpha^{\overleftarrow{L}}_N}\bar G_{\breve T_2,\Gamma}$ and $C_0(\Ab_\Gamma)\rtimes_{\alpha^{\overleftarrow{R}}_N}\bar G_{\breve T_3,\Gamma}$. 

Then the elements of these algebras are represented on two different Hilbert spaces $\HS_{E(\breve T_2)}^\Gamma:=L^2(\bar G_{\breve T_2,\Gamma},\mu_{\breve T_2,\Gamma})\otimes\HS_\Gamma$ and $\HS_{E(\breve T_3)}^\Gamma:=L^2(\bar G_{\breve T_3,\Gamma},\mu_{\breve T_3,\Gamma})\otimes\HS_\Gamma$. Set $\HS_{E(\breve T_i)}:=L^2(\bar G_{\breve T_i,\Gamma},\mu_{\breve T_i,\Gamma})$ for $i=2,3$.
Hence there are two representations $\pi_{E(\breve T_2)}$ and $\pi_{E(\breve T_3)}$ such that $\pi_{E(\breve T_2)}\otimes \pi_{E(\breve T_3)}$
is a representation on $\HS_{E(\breve T_2)}\otimes \HS_{E(\breve T_3)}$. 

The holonomy-flux cross-product $C^*$-algebra $C_0(\Ab_\Gamma)\rtimes_{\alpha^{\overleftarrow{R}}_N}\bar G_{\breve T_3,\Gamma}$ is represented on $\HS_{E(\breve T_3)}^\Gamma$ by
\beqs &\big(\pi^{\Gamma,\breve T_3}_{\overleftarrow{R}}(F_\Gamma)\Psi_{E(\breve T_3)}^\Gamma\big)(\rho_{\breve T_3,\Gamma}(\Gamma))=((\Phi_M(F_\Gamma) \rtimes U_{\overleftarrow{R}}^N(\hat\rho_{\breve T_3,\Gamma}^N))\Psi_{E(\breve T_3)}^\Gamma)(\rho_{\breve T_3,\Gamma}(\Gamma))\\
&=\int_{\bar G_{\breve T_3,\Gamma}}\left(\Phi_M
\left(\alpha_N^{\overleftarrow{R}}(\rho_{\breve T_3,\Gamma}^N)(F_\Gamma(\hat\rho_{\breve T_3,\Gamma}(\Gamma)))\right)
\Psi_{E(\breve T_3)}^\Gamma\right)(L(\hat\rho_{\breve T_3,\Gamma}(\Gamma)^{-1})(\rho_{\breve T_3,\Gamma}(\Gamma)))\dif \mu_{\breve T_3,\Gamma}(\hat\rho_{\breve T_3,\Gamma}(\Gamma))
\eqs for $F_\Gamma(\hat\rho_{\breve T_3,\Gamma}(\Gamma))\in C_0(\Ab_\Gamma)$, $\rho_{\breve T_3,\Gamma}(\Gamma),\hat\rho_{\breve T_3,\Gamma}(\Gamma)\in \bar G_{\breve T_3,\Gamma}$ and $\Psi_{E(\breve T_3)}^\Gamma\in\HS_{E(\breve T_3)}^\Gamma$.\\ Similarly the elements of $C_0(\Ab_\Gamma)\rtimes_{\alpha^{\overleftarrow{L}}_N}\bar G_{\breve T_2,\Gamma}$ is represented on $\HS_{E(\breve T_2)}^\Gamma$.

Now the multiplier algebra of the cross product $C^*$-algebra $C_0(\Ab_\Gamma)\rtimes_{\alpha_{\overleftarrow{L}}^N}\bar G_{\breve T_2,\Gamma}$ is studied.
First of all unitary elements, i.e. $U_{\overleftarrow{R}}^N(\rho_{\breve T_3,\Gamma}(\Gamma))$ for $\rho_{\breve T_3,\Gamma}(\Gamma)\in \bar G_{\breve T_3,\Gamma}$, are elements of the multiplier algebra. This is verified by the identification of $M$ with the map \[\pi^{\Gamma,\breve T_2}_{\overleftarrow{L}}(F_\Gamma)\mapsto \pi^{\Gamma,\breve T_2}_{\overleftarrow{L}}(U_{\overleftarrow{R}}^N(\rho_{\breve T_3,\Gamma}(\Gamma)) F_\Gamma)\in C_0(\Ab_\Gamma)\rtimes_{\alpha_{\overleftarrow{L}}^N}\bar G_{\breve T_2,\Gamma}\] whenever $\pi^{\Gamma,\breve T_2}_{\overleftarrow{L}}(F_\Gamma)\in C_0(\Ab_\Gamma)\rtimes_{\alpha_{\overleftarrow{L}}^N}\bar G_{\breve T_2,\Gamma}$ and  $f_\Gamma\in C_0(\Ab_\Gamma)$ and the computation 
\beqs 
&\Big\la (\pi^{\Gamma,\breve T_2}_{\overleftarrow{L}}(\hat F_\Gamma)
\Psi_{E(\breve T_2)}^\Gamma,
\pi^{\Gamma,\breve T_2}_{\overleftarrow{L}}(M(F_\Gamma))
\Phi_{E(\breve T_2)}^\Gamma\Big\ra_{\HS_{E(\breve T_2)}\otimes\HS_\Gamma}\\
&=\Big\la \pi^{\Gamma,\breve T_2}_{\overleftarrow{L}}(\hat F_\Gamma)
\Psi_{E(\breve T_2)}^\Gamma, 
\pi^{\Gamma,\breve T_2}_{\overleftarrow{L}}(U_{\overleftarrow{R}}^N(\rho_{\breve T_3,\Gamma}(\Gamma))F_\Gamma)
\Phi_{E(\breve T_2),}^\Gamma\Big\ra_{\HS_{E(\breve T_2)}\otimes\HS_\Gamma}\\
&=\int_{\bar G_{\breve T_2,\Gamma}} \dif \mu_{\breve T_2,\Gamma}(\hat\rho_{\breve T_2}(\Gamma))
\Big\la \Phi_M\left(\alpha_{\overleftarrow{L}}^N(\rho_{\breve T_2,\Gamma}(\Gamma))(\hat F_\Gamma(\hat\rho_{\breve T_2,\Gamma}(\Gamma)))\right)U_{\overleftarrow{L}}^N(\hat\rho_{\breve T_2,\Gamma}(\Gamma))\Psi_{E(\breve T_2)}^\Gamma(\rho_{\breve T_2,\Gamma}(\Gamma)),\\
&\qquad\qquad\Phi_M\left((\alpha^{\overleftarrow{R}}_N(\rho_{\breve T_3,\Gamma}(\Gamma))\circ\alpha_{\overleftarrow{L}}^N(\rho_{\breve T_2,\Gamma}(\Gamma)))(F_\Gamma(\hat\rho_{\breve T_2,\Gamma}(\Gamma)))\right)U_{\overleftarrow{R}}^N(\rho_{\breve T_3,\Gamma}(\Gamma))U_{\overleftarrow{L}}^N(\hat\rho_{\breve T_2,\Gamma}(\Gamma))\Psi_{E(\breve T_2)}^\Gamma(\rho_{\breve T_2,\Gamma}(\Gamma))\Big\ra_{\HS^\Gamma_{E(\breve T_2)}}\\
&=\Big\la \pi^{\Gamma,\breve T_2}_{\overleftarrow{L}}(U_{\overleftarrow{R}}^N(\rho_{\breve T_3,\Gamma}(\Gamma)) ^*\hat F_\Gamma)
\Psi_{E(\breve T_2)}^\Gamma, 
\pi^{\Gamma,\breve T_2}_{\overleftarrow{L}}(F_\Gamma)
\Phi_{E(\breve T_2),}^\Gamma\Big\ra_{\HS_{E(\breve T_2)}\otimes\HS_\Gamma}\\
\eqs 

Then show that, the holonomy-flux cross-product $C^*$- algebra $C_0(\Ab_\Gamma)\rtimes_{\alpha_N^{\overleftarrow{R}}}\bar G_{\breve T_3,\Gamma}$ is a subset of the multiplier algebra  $M( C_0(\Ab_\Gamma)\rtimes_{\alpha^N_{\overleftarrow{L}}}\bar G_{\breve T_2,\Gamma})$. The multiplier $M$ is assumed to be the map \[C_0(\Ab_\Gamma)\rtimes_{\alpha^N_{\overleftarrow{L}}}\bar G_{\breve T_2,\Gamma}\ni F_\Gamma\mapsto \hat F_\Gamma \ast F_\Gamma\in C_0(\Ab_\Gamma)\rtimes_{\alpha^N_{\overleftarrow{L}}}\bar G_{\breve T_2,\Gamma}\] 
for a $\hat F_\Gamma\in C_0(\Ab_\Gamma)\rtimes_{\alpha^{\overleftarrow{R}}_N}\bar G_{\breve T_3,\Gamma}$. But since
$L(\rho_{\breve T_3,\Gamma}(\Gamma)^{-1})(\tilde\rho_{\breve T_2,\Gamma}(\Gamma))$ is not well-defined, the convolution
\beqs (\hat F_\Gamma\ast F_\Gamma)(\tilde\rho_{\breve T_2,\Gamma}(\Gamma))&=\int_{\bar G_{\breve T_3,\Gamma}}\dif \mu_{\breve T_3,\Gamma}(\rho_{\breve T_3,\Gamma}(\Gamma))\hat F_\Gamma(\rho_{\breve T_3,\Gamma}(\Gamma))\big(\alpha^{\overleftarrow{R}}_N(\rho_{\breve T_3,\Gamma}(\Gamma))F_\Gamma\big)(L(\rho_{\breve T_3,\Gamma}(\Gamma)^{-1})(\tilde\rho_{\breve T_2,\Gamma}(\Gamma)))
\eqs is not well-defined, too. Consequently, it has to be assumed that either $\bar G_{\breve T_3,\Gamma}$ is embedded into $\bar G_{\breve T_2,\Gamma}$ as a subgroup or the other way arround. Clearly the \textit{situation 4} is of this form.
\end{rem}

\begin{rem}\label{rem situation3}
Let $\breve S$ contains only the surface $S$ and let $\bar S$ be a surface set with same surface intersection property for a path $\gamma$. Then $U_{\overleftarrow{R}}^1(\rho_{\bar S,\gamma}(\gamma))$ is contained in the multiplier algebra of $C_0(\Ab_\gamma)\rtimes_{\alpha_{\overleftarrow{L}}^1}\bar G_{S,\gamma}$. This follows by showing that, the map  \[C_0(\Ab_\gamma)\rtimes_{\alpha_{\overleftarrow{L}}^1}\bar G_{S,\gamma}\ni\pi^{\gamma,S}_{\overleftarrow{L}}(F_\gamma)\mapsto U_{\overleftarrow{R}}^1(\rho_{\bar S,\gamma}(\gamma))\pi^{\gamma,S}_{\overleftarrow{L}}(F_\gamma)\in C_0(\Ab_\gamma)\rtimes_{\alpha_{\overleftarrow{L}}^1}\bar G_{S,\gamma}\] defines a multiplier map. Furthermore it can be shown that,
\[C_0(\Ab_\gamma)\rtimes_{\alpha_{\overleftarrow{L}}^1}\bar G_{S,\gamma}\ni\pi^{\gamma,S}_{\overleftarrow{L}}(F_\gamma)\mapsto \pi^{\gamma,S}_{\overleftarrow{L}}(\hat F_\Gamma\ast F_\gamma)\in C_0(\Ab_\gamma)\rtimes_{\alpha_{\overleftarrow{L}}^1}\bar G_{S,\gamma}\] defines a multiplier map for each function $\hat F_\Gamma\in C_0(\Ab_\gamma)\rtimes_{\alpha^{\overleftarrow{R}}_1}\bar G_{\bar S,\gamma}$.
\end{rem}

\begin{theo}\label{prop multilpiercrossprod} 
Let $\breve S$ be a set of surfaces with simple surface intersection property for a finite orientation preserved graph system associated to a graph $\Gamma$. Let $\{\breve S_i\}$ be a set of sets of surface such that each surface set $\breve S_i$ is suitable for a finite (orientation preserved) graph system associated to a graph $\Gamma$. 

Then the following statements are true:
\begin{enumerate}
 \item The algebra $C_0(\Ab_\Gamma)$, the group $\bar G_{\breve S_i,\Gamma}$ and the group $\bar G_{\breve S,\Gamma}$ are not contained in $C_0(\Ab_\Gamma)\rtimes_{\alpha_{\overleftarrow{L}}^N}\bar G_{\breve S,\Gamma}$. 
\item The analytic holonomy algebra $C_0(\Ab_\Gamma)$ and the unitaries $U_{\overleftarrow{R}}^M(\rho_{\breve S_2,\Gamma}(\Gamma))$, whenever $\rho_{\breve S_i,\Gamma}(\Gamma)\in\bar G_{\breve S_2,\Gamma}$ where $1\leq M\leq N$, are elements of the multiplier algebra of the $C^*$-algebra $C_0(\Ab_\Gamma)\rtimes_{\alpha_{\overleftarrow{L}}^N}\bar G_{\breve S,\Gamma}$.
 \item The unitaries $U_{\overleftarrow{L}}^M(\rho_{\breve S_1,\Gamma}(\Gamma))$, $U_{\overleftarrow{L}}^{\overleftarrow{R},M}(\rho_{\breve S_3,\Gamma}(\Gamma))$, $U_{\overrightarrow{L}}^{\overrightarrow{R},M}(\rho_{\breve S_4,\Gamma}(\Gamma))$, $U_{\overrightarrow{L}}^M(\rho_{\breve S_5,\Gamma}(\Gamma))$ and so on, whenever $\rho_{\breve S_i,\Gamma}(\Gamma)\in\bar G_{\breve S_i,\Gamma}$ where $1\leq M\leq N$ and all $i$,\\[3pt] are elements of the multiplier algebra of the $C^*$-algebra $C_0(\Ab_\Gamma)\rtimes_{\alpha_{\overleftarrow{L}}^N}\bar G_{\breve S,\Gamma}$.
\item The elements of the holonomy-flux cross-product algebra $C_0(\Ab_\Gamma)\rtimes_{\alpha_{\overleftarrow{R}}^M}\bar G_{\breve S_2,\Gamma}$ are multipliers of the $C^*$-algebra $C_0(\Ab_\Gamma)\rtimes_{\alpha_{\overleftarrow{L}}^N}\bar G_{\breve S,\Gamma}$. 
\item Moreover, all elements of $C_0(\Ab_\Gamma)\rtimes_{\alpha^{\overleftarrow{L}}_M}\bar G_{\breve S_1,\Gamma}$, $C_0(\Ab_\Gamma)\rtimes_{\alpha^{\overrightarrow{R}}_M}\bar G_{\breve S_6,\Gamma}$,
$C_0(\Ab_\Gamma)\rtimes_{\alpha_{\overleftarrow{L}}^M}\bar G_{\breve S_7,\Gamma}$, $C_0(\Ab_\Gamma)\rtimes_{\alpha_{\overrightarrow{L}}^M}\bar G_{\breve S_5,\Gamma}$,
$C_0(\Ab_\Gamma)\rtimes_{\alpha^{\overleftarrow{R},M}_{\overleftarrow{L}}}\bar G_{\breve S_3,\Gamma}$ and 
$C_0(\Ab_\Gamma)\rtimes_{\alpha^{\overrightarrow{R},M}_{\overrightarrow{L}}}\bar G_{\breve S_4,\Gamma}$ for $1\leq M\leq N$ are contained in the multiplier algebra the $C^*$-algebra 
$C_0(\Ab_\Gamma)\rtimes_{\alpha_{\overleftarrow{L}}^N}\bar G_{\breve S,\Gamma}$.
\end{enumerate}
\end{theo}
\begin{proofs}
The proof is similar to proposition \ref{prop proofmultiplier} and remarks \ref{rem situation2and4} and \ref{rem situation3}. 
\end{proofs}

In \cite{Kaminski3}, \cite[Section 8.2]{KaminskiPHD} the Lie algebra-valued quantum flux operators $E_S(\Gamma)$ for different surfaces $S$ are considered. Similarly, they are not contained in $C_0(\Ab_\Gamma)\rtimes_{\alpha_L}\bar G_{\breve S,\Gamma}$ or $C_0(\Ab_\Gamma)\rtimes_{\alpha_R}\bar G_{\breve S,\Gamma}$, but they are affiliated in the sense of Woronowicz \cite{WoroNap}. 

\begin{rem}\label{rem minimaltensorprodholflux}
If the action of the flux group $\bar G_{\breve S,\Gamma}$ on $C_0(\Ab_\Gamma)$ is assumed to be the identity, then $ C_0(\Ab_\Gamma)\rtimes_{\id}\bar G_{\breve S,\Gamma}$ is equivalent to $C_0(\Ab_\Gamma)\otimes_{\text{max}}C^*(\bar G_{\breve S,\Gamma})$ where $\otimes_{\text{max}}$ denotes the maximal $C^*$-tensor product. 
\end{rem}

\subsection{The holonomy-flux cross-product $C^*$-algebra for surface sets}\label{subsec holfluxlimit}

Let $G$ be a compact group and $F_\Gamma\in C^*(\bar G_{\breve S,\Gamma},C(\Ab_\Gamma))$. Recall the Weyl-integrated holonomy-flux representation $\pi^{I,\Gamma}_{E(\breve S)}(F_\Gamma)=(\Phi_M\rtimes U_{\overleftarrow{L}}^N)(F_\Gamma)$ of the $C^*$-algebra $C(\Ab_\Gamma)\rtimes_{\alpha_{\overleftarrow{L}}^N}\bar G_{\breve S,\Gamma}$ presented in equation \eqref{defi int hol-flux-repr}. Consider a $^*$-homomorphisms $\beta_{\Gamma,\Gamma^\prime}$ from $C(\Ab_\Gamma)\rtimes_{\alpha_{\overleftarrow{L}}^N}\bar G_{\breve S,\Gamma}$ to $C(\Ab_{\Gp})\rtimes_{\alpha_{\overleftarrow{L}}^N}\bar G_{\breve S,\Gp}$ which satisfies 
\beq
&\beta_{\Gamma,\Gamma^\prime}((\Phi_M\rtimes U_{\overleftarrow{L}}^N)(F_\Gamma(\ho_\Gamma,\rho_{\breve S,\Gamma}(\Gamma))))
=(\Phi_M\rtimes U_{\overleftarrow{L}}^N)(F_{\Gamma^\prime}(\ho_{\Gamma^\prime},\rho_{\breve S,\Gamma}(\Gamma^\prime)))
\eq whenever $\Gamma\leq\Gp$. Then there is an inductive family \[\{(C(\Ab_{\Gamma_i})\rtimes_{\alpha_L}\bar G_{S,\Gamma_i},\beta_{\Gamma_i,\Gamma_j})\quad \beta_{\Gamma_i,\Gamma_j}:\text{ }^*\text{- homomorphisms }\text{ s.t. } \beta_{\Gamma_i,\Gamma_j}=\beta_{\Gamma_i,\Gamma_k}\circ \beta_{\Gamma_k,\Gamma_j}\text{ for } \Gamma_i\leq \Gamma_k\leq\Gamma_j\}\] of $C^*$-algebras derivable.

\begin{defi}Let $\Gamma_\infty$ be the inductive limit of a family of graphs $\{\Gamma_i\}$ such that each graph $\Gamma_i$ of the family has the same intersection surface property for the set $\breve S$ (or the set $\check S$) of surfaces. Set $\vert\Gamma_i\vert=N_i$.
Then $\PD_{\Gamma_\infty}^{\op}$ is the inductive limit of an inductive family $\{\PD_{\Gamma_i}^{\op}\}$ of finite orientation preserved graph systems. 
 
The \textbf{holonomy-flux cross-product $C^*$-algebra $\Alg\rtimes_{\alpha_{\overleftarrow{L}}} \bar G_{\breve S}$ } (of a special surface configuration $\breve S$) is an inductive limit $C^*$-algebra $\limPDi C(\Ab_{\Gamma_i})\rtimes_{\alpha_{\overleftarrow{L}}^{N_i}}\bar G_{S,\Gamma_i}$ of the inductive family of $C^*$-algebras given by 
\[\{(C(\Ab_{\Gamma_i})\rtimes_{\alpha_L}\bar G_{S,\Gamma_i},\beta_{\Gamma_i,\Gamma_j})\quad \beta_{\Gamma_i,\Gamma_j}:\text{ }^*\text{- homomorphisms }\text{ s.t. } \beta_{\Gamma_i,\Gamma_j}=\beta_{\Gamma_i,\Gamma_k}\circ \beta_{\Gamma_k,\Gamma_j}\text{ for } \Gamma_i\leq \Gamma_k\leq\Gamma_j\}\] completed in the norm (where elements of norm $0$ are devided out)
\beq \|F\|:=\inf_{\PD_{\Gamma_j}\supseteq\PD_{\Gamma_i}}\|\beta_{\Gamma_i,\Gamma_j}(F_{\Gamma_i})\|_{\Gamma_j}\text{ for } F_{\Gamma_i}\in  \Alg_{\Gamma_i}\rtimes_{\alpha_{\overleftarrow{L}}^{N_i}}\bar G_{\breve S,\Gamma_i}
\eq
with $\|F_{\Gamma_i}\|_{\Gamma_i}:=\sup_{\pi_E}\|\pi_E(F_{\Gamma_i})\|_2$ where the supremum is taken over all non-degenerate $L^1$-norm decreasing $*$-representations of $L^1(\bar G_{\breve S,\Gamma_i},C(\Ab_{\Gamma_i}))$.
\end{defi}

\begin{prop}\label{prop crossprodstatenotdiffeo2}
Let $\Gamma_\infty$ be the inductive limit of a family of graphs $\{\Gamma_i\}$ such that each graph $\Gamma_i$ of the family has the same intersection surface property for the set $\breve S$ (or the set $\check S$) of surfaces and such that there is only a finite number of intersections of $\breve S$ and all graphs in $\Gamma_\infty$. Set $\vert\Gamma_i\vert=N_i$.
Then $\PD_{\Gamma_\infty}^{\op}$ is the inductive limit of an inductive family $\{\PD_{\Gamma_i}^{\op}\}$ of finite orientation preserved graph systems. Denote the center of the inductive limit group $\bar G_{\breve S}$ by $\bar\ZD_{\breve S}$.

The state $\omega_{E(\breve S)}$ on $\Alg\rtimes_{\alpha_{\overleftarrow{L}}}\bar \ZD_{\breve S}$ associated to the GNS-representation $(\HS_\Gamma,\pi^{I}_{E(\breve S)},\Omega^{I}_{E(\breve S)})$ is not surface-orientation preserving graph-diffeomorphism invariant, but it is a surface preserving graph-diffeomorphism invariant state.
\end{prop}
\begin{proofs}
This is deduced similarly to proposition \ref{prop crossprodstatenotdiffeo}. 
\end{proofs}

\begin{theo}
The \textbf{multiplier algebra $M(\Alg\rtimes_{\alpha_{\overleftarrow{L}}} \bar G_{\breve S})$ of the holonomy-flux cross-product $C^*$-algebra $\Alg\rtimes_{\alpha_{\overleftarrow{L}}} \bar G_{\breve S}$} contains all elements of the holonomy-flux cross-product $C^*$-algebra of any suitable surface set $\breve S$ in $\surf$.   
\end{theo}
\begin{proofs}
 This is derived by using theorem \ref{prop multilpiercrossprod}.
\end{proofs}

\section{The holonomy-flux-graph-diffeomorphism cross-product \\ $C^*$-algebra}\label{subsec holfluxdiffcrossalg}

In this section the holonomy-flux cross-product $C^*$-algebra is enlarged further such that the new $C^*$-algebra contains in a suitable sense the finite graph-diffeomorphisms. Hence, this algebra contains some constraints of the theory of quantum gravity. This is one further step to the aim of the project \textit{AQV}. Notice that, the construction in this section is restricted to surface-preserving graph-diffeomorphisms, but the development can be generalised to surface-orientation-preserving graph-diffeomorphisms. The latter are necessary for the interplay with the quantum flux operators.  

Recall the $C^*$-dynamical system $(\mathfrak{B}(\PD_\Gamma),C_0(\Ab_\Gamma),\zeta)$ defined in \cite[Section 3.2]{Kaminski1}, \cite[Proposition 6.2]{KaminskiPHD}. 
Similarly to the construction of the Banach $^*$-algebra $L^1(\bar G_{\breve S,\Gamma},C_0(\Ab_\Gamma),\zeta)$ in subsection \ref{subsubsec holfluxgraph} the Banach $^*$-algebra $l^1(\mathfrak{B}^\Gamma_{\breve S,\diff}(\PD_\Gamma),C_0(\Ab_\Gamma),\zeta)$ is developed in the next paragraphs. 

Due to the fact that, the number of subgraphs of $\Gamma$ generated by the edges of $\Gamma$ is finite, there exists a finite set $\mathfrak{B}^\Gamma_{\breve S,\ori}(\PD_\Gamma)$ of bisections, such that each of bisection is a map from the set $V_\Gamma$ to a distinct subgraph of $\Gamma$ such that all elements of $\PD_\Gamma$ are construced from the finite set $\mathfrak{B}^\Gamma_{\breve S,\ori}(\PD_\Gamma)$. Call such a set of bisections a \textbf{generating system of bisections for a graph} $\Gamma$. 

The function $F_{\Gamma,\B}$ is contained in $l^1(\mathfrak{B}^\Gamma_{\breve S,\diff}(\PD_\Gamma),C_0(\Ab_\Gamma),\zeta)$ if $F_{\Gamma,\B}$ satisfies
\beqs \|F_{\Gamma,\B}\|_1:=\sum_{l=1,...,k_\Gamma}\|F_{\Gamma,\B}(\ho_{\Gamma}(\Gamma^\prime_{\sigma_l}))\|_2<\infty
\eqs
Then the product of two elements $F_{\Gamma,\B}, K_{\Gamma,\B}\in l^1(\mathfrak{B}^\Gamma_{\breve S,\diff}(\PD_\Gamma),C_0(\Ab_\Gamma),\zeta)$ is defined by
\beqs (F_{\Gamma,\B}\ast K_{\Gamma,\B})(\ho_{\Gamma}(\Gamma^\prime_{\sigma}))
=\sum_{\overset{\tilde\sigma,\breve\sigma\in\mathfrak{B}^\Gamma_{\breve S,\diff}(\PD_\Gamma)}{\tilde\sigma\ast_2\breve\sigma=\sigma}}F_{\Gamma,\B}(\ho_{\Gamma}(\Gamma^\prime_{\tilde\sigma}))K_{\Gamma,\B}(\ho_{\Gamma}(\Gamma^\prime_{\breve\sigma}))
\eqs and the involution is
\beqs F_{\Gamma,\B}(\ho_{\Gamma}(\Gp_{\sigma})):=\overline{F_{\Gamma,\B}(\ho_{\Gamma}(\Gamma^\prime_{\sigma^{-1}}))}
\eqs

There is a $^*$-representation $\pi_{I,\mathfrak{B}}^\Gamma$ of $l^1(\mathfrak{B}^\Gamma_{\breve S,\diff}(\PD_\Gamma),C_0(\Ab_\Gamma),\zeta)$ on $l^2(\mathfrak{B}^\Gamma_{\breve S,\diff}(\PD_\Gamma),C_0(\Ab_\Gamma),\zeta)$ given by
\beqs \pi_{I,\mathfrak{B}}^\Gamma(F_{\Gamma,\B})=\sum_{\sigma\in\mathfrak{B}^\Gamma_{\breve S,\diff}(\PD_\Gamma)} F_{\Gamma,\B}(\ho_{\Gamma}(\Gamma^\prime_{\sigma}))U(\ho_{\Gamma}(\Gamma^\prime_{\sigma}))
\eqs where $U(\ho_{\Gamma}(\Gamma^\prime_{\sigma}))=\delta_{\sigma}$ and $\delta_{\sigma}(\ho_{\Gamma}(\Gamma^\prime_{\breve\sigma})):=\delta(\ho_{\Gamma}(\Gamma^\prime_{\sigma\ast\breve\sigma}))$.

\begin{lem}Let $\mathfrak{B}^\Gamma_{\breve S,\diff}(\PD_\Gamma):=\{\sigma_l\in \mathfrak{B}(\PD_\Gamma)\}_{1\leq l\leq k}$ be a subset of $\mathfrak{B}(\PD_\Gamma)$ that forms a generating system of bisections for the graph $\Gamma$.

The integrated $^*$-representation $\pi_{I,\mathfrak{B}}^\Gamma$ of $l^1(\mathfrak{B}^\Gamma_{\breve S,\diff}(\PD_\Gamma),C_0(\Ab_\Gamma),\zeta)$ is non-degenerate.
\end{lem}
\begin{proofs}
This follows from the fact that, $\pi_{I,\mathfrak{B}}^\Gamma(\FD_{\Gamma,\mathfrak{B}}(\ho_{\Gamma}(\Gamma^\prime)) \delta_{\id}(\ho_\Gamma(\Gamma^\prime))= \FD_{\Gamma,\mathfrak{B}}(\ho_{\Gamma}(\Gamma^\prime))$. 
\end{proofs}

Since $\mathfrak{B}^\Gamma_{\breve S,\diff}(\PD_\Gamma)$ is finite-dimensional and discrete, the reduced holonomy-graph-diffeomorphism group $C^*$-algebra coincides with the holonomy-graph-diffeomorphism cross-product $C^*$-algebra $C_0(\Ab_\Gamma)\rtimes_{\zeta}\mathfrak{B}^\Gamma_{\breve S,\diff}(\PD_\Gamma)$.

But this algebra does not contain any flux variables. Hence recall that in proposition \cite[Proposition 6.2.15]{KaminskiPHD},\cite[Proposition 3.39]{Kaminski1}, it was shown that, the triple $(\mathfrak{B}(\PD_\Gamma\Sigma),\mathsf{W}( \bar G_{\breve S,\Gamma}),\zeta)$ of a surface preserving group $\mathfrak{B}(\PD_\Gamma\Sigma)$ of bisections, a $C^*$-algebra $\mathsf{W}(\bar G_{\breve S,\Gamma})$ associated to a suitable set $\breve S$ of surfaces and a graph $\Gamma$ is a $C^*$-dynamical system in $\LD(\HS_\Gamma)$. 

The pair $(\Phi,V)$, which consists of a morphism $\Phi\in \Mor(\mathsf{W}(\bar G_{\breve S,\Gamma}),\LD(\HS_\Gamma))$ and a unitary representation $V$ of $\mathfrak{B}(\PD_\Gamma\Sigma)$ on $\LD(\HS_\Gamma)$, i.e. $V\in\Rep(\mathfrak{B}(\PD_\Gamma\Sigma),\KD(\HS_\Gamma))$ such that
\beqs \Phi(\zeta_{\sigma} (W))
=V(\sigma) \Phi(W) V^*(\sigma)
\eqs is a covariant representation of $(\mathfrak{B}(\PD_\Gamma\Sigma),\mathsf{W}(\bar G_{\breve S,\Gamma}),\zeta)$ in $\LD(\HS_\Gamma)$.
 
\begin{lem}Let $\breve S$ be a set of surfaces with same surface intersection property for $\Gamma$. Furthermore, let $\mathfrak{B}^\Gamma_{\breve S,\diff}(\PD_\Gamma):=\{\sigma_l\in \mathfrak{B}(\PD_\Gamma)\}_{1\leq l\leq k}$ be a subset of $\mathfrak{B}(\PD_\Gamma)$ that forms a generating system of bisections for the graph $\Gamma$.

Then the triple $(\mathfrak{B}^\Gamma_{\breve S,\diff}(\PD_\Gamma\Sigma),C_0(\Ab_\Gamma)\rtimes_{\alpha}\bar \ZD_{\breve S,\Gamma},\zeta)$ is a $C^*$-dynamical system in $\LD( \HS)$.
\end{lem}
\begin{proofs}Set $\Gamma=\{\gamma_1,...,\gamma_N\}$, $\Gamma_\sigma=\{\gamma_1\circ\sigma(v_1),...,\gamma_N\circ\sigma(v_N)\}$.

Let $F_\Gamma: C_c( \bar \ZD_{\breve S,\Gamma})\rightarrow C_0(\Ab_\Gamma)$ and denote the image of $F_\Gamma(\rho_{S_1}(\gamma_1), ...,\rho_{S_1}(\gamma_1))$ by\\ $F_\Gamma(\rho_{S_1}(\gamma_1), ...,\rho_{S_N}(\gamma_N); \ho_\Gamma(\gamma_1),..., \ho_\Gamma(\gamma_N))$. Notice that,
\beqs &(\zeta_\sigma F_\Gamma)(\rho_{S_1}(\gamma_1), ...,\rho_{S_1}(\gamma_1))\\
&= F_{\Gamma_\sigma}(\rho_{S_1}(\gamma_1\circ\sigma(v_1)), ...,\rho_{S_N}(\gamma_N\circ\sigma(v_N)); \ho_{\Gamma_\sigma}(\gamma_1\circ\sigma(v_1)),..., \ho_{\Gamma_\sigma}(\gamma_N\circ\sigma(v_N)))
\eqs holds. Clearly this defines a point-norm continuous automorphic action. 
\end{proofs}

\begin{prop}Let $\breve S$ be a set of surfaces with same surface intersection property for $\Gamma$. Furthermore, let $\mathfrak{B}^\Gamma_{\breve S,\diff}(\PD_\Gamma):=\{\sigma_l\in \mathfrak{B}(\PD_\Gamma)\}_{1\leq l\leq k}$ be a subset of $\mathfrak{B}(\PD_\Gamma)$ that forms a generating system of bisections for the graph $\Gamma$.

The pair $(\pi_{E(\breve S)}^{I,\Gamma},V)$ is a covariant pair of the $C^*$-dynamical system $(\mathfrak{B}^\Gamma_{\breve S,\diff}(\PD_\Gamma\Sigma),C_0(\Ab_\Gamma)\rtimes_{\alpha}\bar \ZD_{\breve S,\Gamma},\zeta)$ in $\LD( \HS)$.
\end{prop}
\begin{proofs}
Take the $\pi_{E(\breve S)}^{I,\Gamma}$ $^*$-representation of $C_c(\bar \ZD_{\breve S,\Gamma},C_0(\Ab_\Gamma))$ on $\HS_{\Gamma}$ and $V$ a regular representation of $\mathfrak{B}^\Gamma_{\breve S,\diff}(\PD_\Gamma\Sigma)$ on $\HS_{\Gamma}$ to observe that,
\beqs 
&\pi_{E(\breve S)}^{I,\Gamma}(\zeta_\sigma (F_\Gamma))\Omega^{I}_{E(\breve S)}\\ 
&= \int_{\bar \ZD_{\breve S,\Gamma}} \dif\mu_{\breve S,\Gamma}(\rho_{S_1}(\gamma_1),...,\rho_{S_N}(\gamma_N))(\zeta_\sigma F_\Gamma)(\rho_{S_1}(\gamma_1),...,\rho_{S_N}(\gamma_N))\Omega^{I}_{E(\breve S)}\\
&= \int_{\bar \ZD_{\breve S,\Gamma}} \dif\mu_{\breve S,\Gamma}(\rho_{S_1}(\gamma_1\circ\sigma(v_1)),...,\rho_{S_N}(\gamma_N\circ\sigma(v_N)))\\
&\qquad\qquad F_\Gamma(\rho_{S_1}(\gamma_1\circ\sigma(v_1)),...,\rho_{S_N}(\gamma_N\circ\sigma(v_N)))\Omega^{I}_{E(\breve S)}\\
&= \int_{\bar \ZD_{\breve S,\Gamma}} \dif\mu_{\breve S,\Gamma}(\rho_{S_1}(\gamma_1\circ\sigma(v_1)),...,\rho_{S_N}(\gamma_N\circ\sigma(v_N)))\\
&\qquad\qquad F_\Gamma(\rho_{S_1}(\gamma_1\circ\sigma(v_1)),...,\rho_{S_N}(\gamma_N\circ\sigma(v_N)))V^*_\sigma\Omega^{I}_{E(\breve S)}\\
&=V_\sigma \pi_{E(\breve S)}^{I,\Gamma}(F_\Gamma)V^*_\sigma\Omega^{I}_{E(\breve S)}
\eqs yields if $v_i=t(\gamma_i)$ for $i=1,...,N$.
Consequently, $(\pi_{E(\breve S)}^{I,\Gamma},V)$ is a covariant representation. 
\end{proofs}

In proposition \ref{prop crossprodstatenotdiffeo} it has been shown that, the state $\omega_{E(\breve S)}$ of $C_0(\Ab_\Gamma)\rtimes_{\alpha}\bar \ZD_{\breve S,\Gamma}$ is graph-diffeomorphism invariant in general. 
There is a finite surface-orientation-preserving graph-diffeomorphism and hence a $\mathfrak{B}_{\breve S,\diff}(\PD_\Gamma)$-invariant state of $C_0(\Ab_\Gamma)\rtimes_{\alpha}\bar \ZD_{\breve S,\Gamma}$ on $\HS_\Gamma$ given by
\beqs \omega_{E(\breve S)}^ \Gamma(\zeta_{\sigma}(\FD_{\Gamma,\breve S}))
&=
\la \Omega_{E(\breve S)}^ \Gamma,V_\sigma \pi_{E(\breve S)}^{I,\Gamma}(\FD_{\Gamma,\breve S})V_\sigma^*\Omega_{E(\breve S)}^ \Gamma\ra\\
&= \omega_{E(\breve S)}^ \Gamma(\FD_{\Gamma,\breve S})
\eqs whenever $\sigma\in \mathfrak{B}_{\breve S,\diff}(\PD_\Gamma)$ and $\FD_{\Gamma,\breve S}\in C_0(\Ab_\Gamma)\rtimes_{\alpha}\bar \ZD_{\breve S,\Gamma}$.

\begin{prop}Let $\breve S$ be a set of surfaces with same surface intersection property for $\Gamma$. Furthermore, let $\mathfrak{B}^\Gamma_{\breve S,\diff}(\PD_\Gamma):=\{\sigma_l\in \mathfrak{B}(\PD_\Gamma)\}_{1\leq l\leq k}$ be a subset of $\mathfrak{B}(\PD_\Gamma)$ that forms a generating system of bisections for the graph $\Gamma$.

The space $l^1(\mathfrak{B}^\Gamma_{\breve S,\diff}(\PD_\Gamma),C_0(\Ab_\Gamma)\rtimes_{\alpha}\bar \ZD_{\breve S,\Gamma},\zeta)$ is defined by all functions $\FD_{\Gamma,\breve S}:\mathfrak{B}^\Gamma_{\breve S,\diff}(\PD_\Gamma)\rightarrow C_0(\Ab_\Gamma)\rtimes_{\alpha}\bar \ZD_{\breve S,\Gamma}$ for which
\beqs \|\FD_{\Gamma,\breve S}\|_1= \sum_{\sigma\in\mathfrak{B}^\Gamma_{\breve S,\diff}(\PD_\Gamma)}\|\FD_{\Gamma,\breve S}(\sigma(t(\gamma_1)),...,\sigma(t(\gamma_N)))\|_2<\infty
\eqs is true.

The convolution $^*$-algebra $l^1(\mathfrak{B}^\Gamma_{\breve S,\diff}(\PD_\Gamma),C_0(\Ab_\Gamma)\rtimes_{\alpha}\bar \ZD_{\breve S,\Gamma},\zeta)$ is presented by the multiplication
\beqs &(\GG_{\Gamma,\breve S}\ast\FD_{\Gamma,\breve S})(\sigma(t(\gamma_1)),...,\sigma(t(\gamma_N)))\\
&=\sum_{\overset{\tilde\sigma,\breve\sigma\in\mathfrak{B}^\Gamma_{\breve S,\diff}(\PD_\Gamma)}{}} \GG_{\Gamma,\breve S}(\sigma(t(\gamma_1)),...,\sigma(t(\gamma_N)))\zeta_{\sigma}\Big(\FD_{\Gamma,\breve S}((\sigma^{-1}\ast \sigma^\prime)(t(\gamma_1)),...,(\sigma^{-1}\ast \sigma^\prime)(t(\gamma_N)))\Big)
\eqs and the involution
\beqs \FD_{\Gamma,\breve S}^*(\sigma(t(\gamma_1)),...,\sigma(t(\gamma_N)))=\zeta_\sigma(\FD_{\Gamma,\breve S}(\sigma^{-1}(t(\gamma_1)),...,\sigma^{-1}(t(\gamma_N)))^*)
\eqs
where the involution $^*$ of $l^1(\mathfrak{B}^\Gamma_{\breve S,\diff}(\PD_\Gamma),C_0(\Ab_\Gamma)\rtimes_{\alpha}\bar \ZD_{\breve S,\Gamma},\zeta)$ is inherited from the involution $^*$ of $C_0(\Ab_\Gamma)\rtimes_{\alpha}\bar \ZD_{\breve S,\Gamma}$
\beqs \FD_{\Gamma,\breve S}^*(\sigma(t(\gamma_1)),...,\sigma(t(\gamma_N)))
=\alpha(\rho_{\breve S}(\Gamma))\Big(\FD_{\Gamma,\breve S}^+(\sigma(t(\gamma_1)),...,\sigma(t(\gamma_N));\rho_{S_1}(\gamma_1)^{-1},...,\rho_{S_N}(\gamma_N)^{-1})\Big)
\eqs and
\beqs &\FD_{\Gamma,\breve S}^+(\sigma(t(\gamma_1)),...,\sigma(t(\gamma_N));\rho_{S_1}(\gamma_1)^{-1},...,\rho_{S_N}(\gamma_N)^{-1})\\&=\overline{\FD_{\Gamma,\breve S}(\sigma(t(\gamma_1)),...,\sigma(t(\gamma_N));\rho_{S_1}(\gamma_1)^{-1},...,\rho_{S_N}(\gamma_N)^{-1})}
\eqs
where the map 
\beqs(\sigma(t(\gamma_1)),...,\sigma(t(\gamma_N)))\mapsto \FD_{\Gamma,\breve S}(\sigma(t(\gamma_1)),...,\sigma(t(\gamma_N)))\eqs defines an element in $l^1(\mathfrak{B}^\Gamma_{\breve S,\diff}(\PD_\Gamma),C_0(\Ab_\Gamma)\rtimes_{\alpha}\bar \ZD_{\breve S,\Gamma},\zeta)$, the map \beqs(\sigma(t(\gamma_1)),...,\sigma(t(\gamma_N));\rho_{S_1}(\gamma_1)^{-1},...,\rho_{S_N}(\gamma_N)^{-1})\mapsto \FD_{\Gamma,\breve S}(\sigma(t(\gamma_1)),...,\sigma(t(\gamma_N));\rho_{S_1}(\gamma_1)^{-1},...,\rho_{S_N}(\gamma_N)^{-1})\eqs
defines an element in $C_0(\Ab_\Gamma)\rtimes_{\alpha}\bar \ZD_{\breve S,\Gamma}$ and, finally, the map 
\beqs &(\sigma(t(\gamma_1)),...,\sigma(t(\gamma_N));\rho_{S_1}(\gamma_1)^{-1},...,\rho_{S_N}(\gamma_N)^{-1};\ho_\Gamma(\gamma_1),...,\ho_\Gamma(\gamma_N))\\&\qquad\mapsto \FD_{\Gamma,\breve S}(\sigma(t(\gamma_1)),...,\sigma(t(\gamma_N));\rho_{S_1}(\gamma_1)^{-1},...,\rho_{S_N}(\gamma_N)^{-1};\ho_\Gamma(\gamma_1),...,\ho_\Gamma(\gamma_N))\eqs defines an element in $C_0(\Ab_\Gamma)$.

The space $l^1(\mathfrak{B}^\Gamma_{\breve S,\diff}(\PD_\Gamma),C_0(\Ab_\Gamma)\rtimes_{\alpha}\bar \ZD_{\breve S,\Gamma},\zeta)$ is a well-defined Banach $^*$-algebra. 
\end{prop}

\begin{defi}Let $\breve S$ be a set of surfaces with same surface intersection property for $\Gamma$. Furthermore, let $\mathfrak{B}^\Gamma_{\breve S,\diff}(\PD_\Gamma):=\{\sigma_l\in \mathfrak{B}(\PD_\Gamma)\}_{1\leq l\leq k}$ be a subset of $\mathfrak{B}(\PD_\Gamma)$ that forms a generating system of bisections for the graph $\Gamma$.

Let $(\pi_{E(\breve S)}^{I,\Gamma},V)$ be a covariant representation of $(\mathfrak{B}^\Gamma_{\breve S,\diff}(\PD_\Gamma),C_0(\Ab_\Gamma)\rtimes_{\alpha}\bar \ZD_{\breve S,\Gamma},\zeta)$ in $\LD(\HS_{\Gamma})$.
 
Define the \textbf{integrated holonomy-flux-graph-diffeomorphism representation of\\ $l^1(\mathfrak{B}^\Gamma_{\breve S,\diff}(\PD_\Gamma),C_0(\Ab_\Gamma)\rtimes_{\alpha}\bar \ZD_{\breve S,\Gamma},\zeta)$} by
\beqs &\pi_{I,\mathfrak{B}}(\FD_{\Gamma,\breve S}(\sigma(t(\gamma_1)),...,\sigma(t(\gamma_N))))
=\sum_{\sigma\in\mathfrak{B}^\Gamma_{\breve S,\diff}(\PD_\Gamma\Sigma)}\pi_{E(\breve S)}^{I,\Gamma}\big(\FD_{\Gamma,\breve S,\sigma}(\sigma(t(\gamma_1)),...,\sigma(t(\gamma_N)))\big)V_{\sigma}\\
&= \sum_{\overset{\delta_i\in\PD_\Gamma\Sigma^{t(\gamma_i)}}{i=1,..,N}}
\pi_{E(\breve S)}^{I,\Gamma}\big(\FD_{\Gamma,\breve S}(\delta_1,...,\delta_N)\big)V(\delta_1,...,\delta_N)
\eqs 
such that the sum is over all paths $\delta_i$, which start at $t(\gamma_i)$ and $\delta_i\in\PD_\Gamma\Sigma$. 
\end{defi}
\begin{defi}Let $\breve S$ be a set of surfaces with same surface intersection property for $\Gamma$. Furthermore, let $\mathfrak{B}^\Gamma_{\breve S,\diff}(\PD_\Gamma):=\{\sigma_l\in \mathfrak{B}(\PD_\Gamma)\}_{1\leq l\leq k}$ be a subset of $\mathfrak{B}(\PD_\Gamma)$ that forms a generating system of bisections for the graph $\Gamma$.

The \textbf{reduced holonomy-flux-graph-diffeomorphism group $C^*$-algebra $C^*_r(\mathfrak{B}^\Gamma_{\breve S,\diff}(\PD_\Gamma),C_0(\Ab_\Gamma)\rtimes_{\alpha}\bar \ZD_{\breve S,\Gamma})$ of a graph $\Gamma$ and a set of surfaces $\breve S$} is defined as the closure of $l^1(\mathfrak{B}^\Gamma_{\breve S,\diff}(\PD_\Gamma),C_0(\Ab_\Gamma)\rtimes_{\alpha}\bar \ZD_{\breve S,\Gamma},\zeta)$ in the norm $\|\FD_{\Gamma,\breve S}\|:=\|\pi_{I,\mathfrak{B}}(\FD_{\Gamma,\breve S})\|_2$.
\end{defi}

\begin{prop}Let $\breve S$ be a set of surfaces with same surface intersection property for $\Gamma$. Furthermore, let $\mathfrak{B}^\Gamma_{\breve S,\diff}(\PD_\Gamma):=\{\sigma_l\in \mathfrak{B}(\PD_\Gamma)\}_{1\leq l\leq k}$ be a subset of $\mathfrak{B}(\PD_\Gamma)$ that forms a generating system of bisections for the graph $\Gamma$.

Suppose that $(\mathfrak{B}^\Gamma_{\breve S,\diff}(\PD_\Gamma),C_0(\Ab_\Gamma)\rtimes_{\alpha}\bar \ZD_{\breve S,\Gamma},\zeta)$ in $\LD(\HS)$ is a $C^*$-dynamical system and that for each\\ $F_\Gamma\in l^1 (\mathfrak{B}^\Gamma_{\breve S,\diff}(\PD_\Gamma),C_0(\Ab_\Gamma)\rtimes_{\alpha}\bar \ZD_{\breve S,\Gamma},\zeta)$ define
\beqs \|\FD_{\Gamma,\breve S}\|:=\sup\Big\{\|(\pi\rtimes V)(\FD_{\Gamma,\breve S})\| :&(\pi,V)\text{ is a covariant representation of } (\mathfrak{B}^\Gamma_{\breve S,\diff}(\PD_\Gamma),C_0(\Ab_\Gamma)\rtimes_{\alpha}\bar \ZD_{\breve S,\Gamma},\zeta)\Big\}\eqs 

Then $\|.\|$ is a norm on $l^1 (\mathfrak{B}^\Gamma_{\breve S,\diff}(\PD_\Gamma),C_0(\Ab_\Gamma)\rtimes_{\alpha}\bar \ZD_{\breve S,\Gamma},\zeta)$ called the universal norm. The universal norm is dominated by the $\|.\|_1$-norm, and the completition of $l^1 (\mathfrak{B}^\Gamma_{\breve S,\diff}(\PD_\Gamma),C_0(\Ab_\Gamma)\rtimes_{\alpha}\bar \ZD_{\breve S,\Gamma},\zeta)$ with respect to $\|.\|$ is a $C^*$-algebra. 
This $C^*$-algebra is called the \textbf{holonomy-flux-graph-diffeomorphism cross-product $C^*$-algebra} $\left(C_0(\Ab_\Gamma)\rtimes_{\alpha}\bar \ZD_{\breve S,\Gamma}\right)\rtimes_{\zeta}\mathfrak{B}^\Gamma_{\breve S,\diff}(\PD_\Gamma)$ associated to a graph $\Gamma$ and a set $\breve S$ of surfaces. 
\end{prop}

In \cite[Propostion 3.32]{Kaminski1} or \cite[Proposition 6.2.2]{KaminskiPHD}, it has been argued that, there are several $C^*$-dynamical systems available for the analytic holonomy $C^*$-algebra and the group of bisections. This is used to define a bunch of holonomy-flux-graph-diffeomorphism cross-product $C $-algebras, which are constructed from $C^*$-dynamical systems. These cross-product $C^*$-algebra are exterior equivalent, too. Clearly there is a multiplier algebra of the holonomy-flux-graph-diffeomorphism cross-product algebra associated to a graph and a set of surfaces is derivable. The author of this article suggests that it can be proven that, the different holonomy-flux-graph-diffeomorphism cross-product $C $-algebras are contained in this multiplier algebra by using similiar arguments used in the proof of theorem \ref{prop multilpiercrossprod}. The construction of the inductive limit $C^*$-algebra of a family of $C^*$-algebras defined above is not mathematically understood very well until now. A detailed study of these objects is a further project.

\section{Comparison table}\label{end}

The Weyl $C^*$-algebra for surfaces and the holonomy-flux cross-product $C^*$-algebra associated to a certain surface set are constructed from functions depending on holonomies along paths of a graph, and the strongly continuous unitary representation of the quantum flux group for surfaces. In contrast to the Weyl algebra, where the group-valued quantum flux operators are implemented as unitary operators, the elements of the holonomy-flux cross-product $C^*$-algebra are operator-valued functions depending on group-valued quantum flux variables for surfaces. In both cases these operators are represented on Hilbert spaces.

\begin{longtable}[ht]{|l|l|l|l|}
\hline &&\\ 

& Weyl $C^*$-algebra for surfaces &   holonomy-flux cross-product $C^*$-algebra \\ &&\\ \hline\hline &&\\
ingredients&set of fin. set of surfaces $\breve S$&set of fin. set of surfaces $\breve S$\\[5pt]
&$G$ locally compact group&$G$ locally compact group\\[5pt]
inductive family of& fin. graph systems&fin. orientation-preserved graph systems\\[5pt]
config. space& $\Ab_\Gamma$ & $\Ab_\Gamma$\\[3pt]
assumption&natural or non-standard identific. of &natural or non-standard identification of \\[3pt]
&a set of independent paths in $\PD_\Gamma\Sigma$&a set of independent paths in $\PD_\Gamma\Sigma$\\[3pt]
mom. space &the flux groups $\bar G_{\breve S,\Gamma}$ or $\bar G_{\breve S,\Gamma_\infty}$ & the flux groups $\bar G_{\breve S,\Gamma}$ or $\bar G_{\breve S,\Gamma_\infty}$\\[5pt]
Hilbert space&$\HS_\Gamma:=L^2(\Ab_\Gamma,\dif\mu_{\Gamma})$ or &$\HS_\Gamma:=L^2(\Ab_\Gamma,\dif\mu_{\Gamma})$ or\\[3pt]
&$\HS_\infty:=L^2(\Ab,\dif\mu_{\infty})$&  $\HS_\infty:=L^2(\Ab,\dif\mu_{\infty})$\\[3pt]
&&$\HS_{E(\breve S)}^\Gamma:=L^2(\bar G_{\breve S,\Gamma},\mu_{\breve S,\Gamma})\otimes\HS_\Gamma$\\[5pt]
representations & $\Phi_M\in\Rep(C_0(\Ab_\Gamma),\LD(\HS_\Gamma))$ & $\Phi_M\in\Rep(C_0(\Ab_\Gamma),\LD(\HS_\Gamma))$ \\[5pt]
& left regular representation & Weyl-integrated holonomy-flux repres. \\[3pt]
& of the flux group& of the holonomy-flux cross-product $^*$-alg.\\[5pt]
&$U\in\Rep(\bar G_{\breve S,\Gamma},\KD(\HS_\Gamma))$&$\pi^{I,\Gamma}_{E(\breve S )}\in\Rep(L^1(\bar G_{\breve S,\Gamma},C_0(\Ab_\Gamma)),\KD(\HS_\Gamma))$\\[5pt]
\hline
\hline &&\\[3pt]
$^*$-algebra& Weyl algebra generated by &  \\[3pt]
&$C_0(\Ab_\Gamma)$ and $\{U\in\Rep(\bar G_{\breve S,\Gamma},\KD(\HS_\Gamma))\}$&$L^1(\bar G_{\breve S,\Gamma},C_0(\Ab_\Gamma))$\\[5pt]
completion w.r.t.&$L^2(\Ab_\Gamma,\mu_\Gamma)$-norm& universal-norm\\[5pt]
$C^*$-algebra & $\WF\text{eyl}(\breve S,\Gamma)$  &$C_0(\Ab_\Gamma)\rtimes_{\alpha}\bar G_{\breve S,\Gamma}$\\[5pt]
&&multiplier algebra of $C_0(\Ab_\Gamma)\rtimes_{\alpha}\bar G_{\breve S,\Gamma}$\\[5pt]
induct. limit $C^*$-alg.& $\WF\text{eyl}(\breve S)$  &$C(\Ab)\rtimes_{\alpha}\bar G_{\breve S}$\hspace*{3pt} ($G$ compact)\\[5pt]

state& unique and pure state $\bar\omega_M$ on $\WF\text{eyl}_\ZD(\breve S)$& state $\omega_{E(\breve S)}$ on $C(\Ab)\rtimes_{\alpha}\bar\ZD_{\breve S}$ s.t.\\[5pt]
&  s.t.  $\bar\omega_M\circ\zeta_{(\Phi,\varphi)}=\bar\omega_M$&$\omega_{E(\breve S)}\circ\zeta_{(\Phi,\varphi)}=\omega_{E(\breve S)}$\\[3pt]
&$\forall (\Phi,\varphi)$ certain diffeomorphism&$\forall (\Phi,\varphi)$ certain diffeomorphism, which\\[3pt]
&& preserve the surfaces in $\breve S$\\[5pt]
\hline
\end{longtable}

\section*{Acknowledgements}
The work has been supported by the Emmy-Noether-Programm (grant FL 622/1-1) of the Deutsche Forschungsgemeinschaft.

\end{document}